\documentclass{llncs}
\newif\ifconference
\conferencefalse

\usepackage{cite}
\usepackage{todonotes}
\usepackage{alltt}
\usepackage{amsfonts}
\usepackage{listings}
\usepackage{longtable}
\usepackage{multirow}
\usepackage{pdflscape}
\usepackage{pgf}
\usepackage{pifont}
\usepackage{color}
\usepackage[english,rounding]{rccol}
\usepackage{rotating}
\usepackage{setspace}
\usepackage{tabularx}
\usepackage{tipa}
\usepackage{wrapfig}
\usepackage{xcolor}
\usepackage{xspace}
\usepackage{textcomp}
\usepackage{subcaption}
\usepackage{enumitem}
\usepackage{multicol}
\usepackage{style-bib}

\usepackage{algorithmicx}
\usepackage{algpseudocode}
\usepackage{algorithm}
\usepackage{macros}
\usepackage{ioa_code}
\usepackage{amsmath, amssymb}
\usepackage{graphicx}
\setcitestyle{number}

\rcDecimalSignOutput{.} % rccol

\newcommand{\MSG}{{\scriptsize $\lit{MSG}$}}
\newcommand{\INIT}{{\ms{INIT}}}
\newcommand{\ECHO}{{\ms{ECHO}}}
\newcommand{\READY}{{\ms{READY}}}
\newcommand{\EST}{\ms{EST}}
\newcommand{\AUX}{{\ms{AUX}}}
\newcommand{\DECS}{{\ms{DECS}}}
\newcommand{\COORDVALUE}{{\ms{COORD\_VALUE}}}
\newcommand{\ESTONES}{{\ms{EST\_ONES}}}
\newcommand{\AUXONES}{{\ms{AUX\_ONES}}}
\newcommand{\TEMP}{\ms{TEMP}}
\newcommand{\TAG}{{\ms{TAG}}}

\title{Anonymity Preserving Byzantine\\ Vector Consensus}
\ifconference
\author{}
\institute{}
\else
\author{Christian Cachin\thanks{University of Bern, \texttt{cachin@inf.unibe.ch}}
	\and Daniel Collins\thanks{University of Sydney, \texttt{dcol9436@uni.sydney.edu.au}}
	\and Tyler Crain\thanks{University of Sydney, \texttt{tycrain@gmail.com}}
	\and Vincent Gramoli\thanks{University of Sydney, \texttt{vincent.gramoli@sydney.edu.au}}
}
\institute{}
\fi

\begin{document}
%\setstretch{0.9}
%\setstretch{0.82}
\maketitle

\begin{abstract}
	\noindent
	Collecting anonymous opinions finds various applications ranging from simple whistleblowing,  
	releasing secretive information, to complex forms of voting, where participants rank candidates by order of preferences.
	Unfortunately, as far as we know there is no efficient distributed solution to this problem. Previously,
	participants had to trust third parties, run expensive cryptographic protocols or sacrifice anonymity.
	In this paper, we propose a resilient-optimal solution to this problem called AVCP, which tolerates up to a third of Byzantine participants.
	AVCP combines traceable ring signatures to detect double votes with a reduction from vector consensus to binary consensus to ensure all valid votes are taken into account.
	We prove our algorithm correct and show that it preserves anonymity with at most a linear communication overhead and constant message overhead when compared to a recent consensus baseline.
	Finally, we demonstrate empirically that the protocol is practical by deploying it on 100 machines geo-distributed in 3 continents: America, Asia and Europe.
	Anonymous decisions are reached within 10 seconds with a conservative choice of traceable ring signatures.

\end{abstract}

%\newpage
%\setcounter{page}{0}
%\pagenumbering{arabic}

\section{Introduction and related work}

%\input oldintro.tex

%\iffalse
%Blockchain technology promises to distribute decision power over a group of
%participants and is often used to run distributed computations in an open
%environment or to track valuable items in applications accessible to any
%Internet user.  Yet, the so-called \emph{permissionless} blockchains like
%Bitcoin~\cite{Nakamoto2008} or Ethereum~\cite{Wood} suffer from heavy computational effort and from
%external, inefficient governance models, because decisions about the
%protocols themselves are relegated to an external, human level.  Multiple
%recent forks in those systems demonstrate this.
%The alternative model of a \emph{consortium blockchain} relies on a set of
%predetermined decision makers. They operate much faster~\cite{CGLR18}
%than permissionless protocols and do not rely on costly computations for 
%decision making. Moreover, reconfiguration~\cite{Vizier2018} allows for
%flexibility in the set of decision makers over time.

%This paper addresses the problem of voting and governance in permissioned blockchains such that the opinion of every participant is accounted for.
%In accordance with established election protocols, voting procedures, and other decision-making rules \cite{Cramer1997, Kiayias2002, Adida2008}, we demand a level of anonymity to protect the voters.
%\fi

%%% [V:
Consider a distributed survey where a group of mutually distrusting participants wish to exchange their opinions about some issue.
%Distributing surveys requires a typically large group of mutually distrusting participants to exchange their opinions.
%%% ]
%Their network of communication may be subject to arbitrary message delays---that is, it may be asynchronous---and some participants may act arbitrarily faulty (or be \textit{Byzantine} faulty).
For example, participants may wish to communicate over the Internet to
rank candidates in order of preference to change the governance of a blockchain.
%~\cite{Comchain}, 
%In this scenario, one cannot reasonably predict the time it will take for them to communicate, as in asynchronous networks.
%%% [V:
%One way to achieve this involves homomorphic tallying~\cite{Cramer1996, Adida2008}, where votes are encoded as values that are summed in encrypted form prior to decryption.
%%% ]
%Importantly, using vector consensus in an electronic election allows the contents of a proposal (vote) to be arbitrary.
%%% [V
%However, if the set of choices or candidates in the election is sufficiently large, such as a population, the cryptographic overhead of homomorphic tallying can be impractically high.
%Indeed, votes for a particular choice or choices must be meaningfully encoded mathematically.
%%% ]
%%% [V:
Without making additional trust assumptions~\cite{Cramer1996, Adida2008, Juels2010}, one promising approach is to run a Byzantine consensus algorithm~\cite{Lamport1982}, or more generally a vector consensus algorithm~\cite{Doudou1997, Neves2005, Correia2006} to allow for arbitrary votes.
In vector consensus, a set of participants decide on a common vector of values, each value being proposed by one process.
Unlike interactive consistency~\cite{Lamport1982}, a protocol solving vector consensus can be executed without fully synchronous communication channels, and as such is preferable for use over the Internet.
%When communicating over networks, like the Internet, that are subject to unpredictable delays, it is paramount not to assume synchrony.
%To this end, the Internet is
%%% ]
%At least $n - t$ values are decided by each non-faulty process; the classical interactive cons
%%% [V:
%Although vector consensus can fulfil the requisite assumptions on the network and fault-tolerance, processes are trivially tied to the opinion that they propose.
%When this is problematic---in electronic voting and when whistle-blowing, for example---it is desirable to allow processes to anonymously propose their value.
Unfortunately, vector consensus protocols tie each participant's opinion to its identity to ensure one opinion is not overrepresented in the decision and to avoid double voting.
%The problem then becomes the lack of anonymity as consensus inherently ties the identity of a participant to its opinion, hence disincentivising whistleblowing or controversial opinions being expressed.
There is thus an inherent difficulty in solving vector consensus while preserving anonymity.
%%% ]
%%% [V:
%Motivated by this, we define the \textit{anonymised vector consensus} problem, enriching vector consensus by asserting that all non-faulty processes cannot be associated with their proposal by an adversary via any better strategy than randomly guessing their identity.

\sloppy{In this paper, we introduce the \textit{anonymity-preserving vector consensus} problem that prevents an adversary from discovering the identity of non-faulty participants that propose values in the consensus and we introduce a solution called Anonymised Vector Consensus Protocol (AVCP).}
%%% ]
%%% [V:
%Our primary contribution, Anonymised Vector Consensus Protocol (AVCP), is a reduction to binary consensus, which reduces the influence of any singular process, such as a leader in an protocol like PBFT~\cite{Castro1999}, who could selectively include or remove proposals at its whim.
To prevent the leader in some Byzantine consensus algorithms~\cite{Castro1999} from influencing the outcome of the vote by discarding proposals, AVCP reduces the problem to binary consensus that is executed without the need for a traditional leader~\cite{CGLR18}.

We provide a mechanism to prevent Byzantine processes from double voting whilst also decoupling their ballots from their identity. 
In particular, we adopt \textit{traceable ring signatures}~\cite{Fujisaki2007, gu2019efficient}, which enable participants to anonymously prove their membership in a set of participants, exposing a signer's identity if and only if they sign two different votes.
This can disincentivise participants from proposing multiple votes to the consensus.
%%% ]
Alternatively, we could use linkable ring signatures~\cite{Liu2004}, but would not ensure that Byzantine processes are held accountable when double-signing.
%which ensure that different signatures signed by one signer can be associated together without exposing their identity.
%They can be more efficient in size~\cite{Tsang2005a} and computational effort~\cite{Liu2005}, but cannot ensure that Byzantine processes are held accountable when double-signing.
We could also have used blind signatures \cite{Chaum1983, Camp1996}, but this would have required an additional trusted authority.

%%% [V:
%Process must propagate their signatures anonymously, else their identity would be trivially exposed.
%To do this, we could take advantage of in-house protocols.
%However, these protocols generally either assume the existence of additional trusted parties~\cite{Chaum1988, Jakobsson1998, Golle2004} or require $O(n)$ message steps to perform all-to-all broadcast~\cite{Golle2004a}.
%It appears impossible construct such a protocol without some synchrony, and appears non-trivial to withstand Byzantine failures with a small (constant) number of message steps without additional trust.
%Thus, our solution assumes that anonymous communication channels can be instantiated.
%In practice, one can use publicly-deployed \cite{Dingledine2004, Zantout2011} networks, which generally requires sustained network observation or large amounts of resources~\cite{Gilad2012, Murdoch2005} to de-anonymise with high probability.
We also identify interesting conditions to ensure anonymous communication.
Importantly, participants must propagate their signatures anonymously to hide their identity.
To this end, we could construct anonymous channels directly.
However, these protocols require additional trusted parties, are not robust or require $O(n)$ message delays to terminate~\cite{Chaum1988, Jakobsson1998, Golle2004, Golle2004a}.
%It appears impossible construct such a protocol without some synchrony, and appears non-trivial to withstand Byzantine failures with a small (constant) number of message steps without additional trust.
%Thus, our solution assumes that anonymous communication channels can be instantiated.
Thus, we assume access to anonymous channels, such as through publicly-deployed \cite{Dingledine2004, Zantout2011} networks which often require sustained network observation or large amounts of resources to de-anonymise with high probability~\cite{Gilad2012, Murdoch2005}.
%%% ]
%
%%% [V:
%Particularly because our protocols rely on both anonymous and regular channels, latency~\cite{Back2001} and message transmission ordering~\cite{Raymond2001} can diminish the anonymity provided by the anonymous channels themselves.
%For example, if a process is much slower than all other processes, an adversary may associate their identity with the anonymous message that they deliver last.
%To cope, we assume non-faulty processes anonymously broadcast after all non-faulty processes have terminated in any previous instance of AVCP, and that the timing and order of messages exchanged over anonymous channels are statistically independent from that of those exchanged over regular channels.
%These measures ensure that, in the view of the adversary, each message received by a process that they corrupt could have been sent by any non-faulty process.
Anonymity is ensured then as processes do not reveal their identity via ring signatures and communicate over anonymous channels.
When correlation-based attacks on certain anonymous networks may be viable~\cite{Mathewson2004} and for efficiency's sake, processes may anonymously broadcast their proposal and then continue the protocol execution over regular channels.
%For efficiency and to reduce the impact of correlating anonymous message passing with regular message passing in practice~\cite{Mathewson2004}, processes only anonymously broadcast once in a given protocol execution.
However, anonymous channels alone cannot ensure anonymity when combined with regular channels as an adversary could infer identities through latency~\cite{Back2001} and message transmission ordering~\cite{Raymond2001}.
For example, an adversary may relate late message arrivals from a single process over both channels to deduce the identity of a slow participant.
This is why, to ensure anonymity, the timing and order of messages exchanged over anonymous channels should be statistically independent (or computationally indistinguishable with a computational adversary) from that of those exchanged over regular channels.
In practice, one may attempt to ensure that there is a low correlation by ensuring that message delays over anonymous and regular channels are sufficiently randomised.
We construct our solution iteratively first by defining the \textit{anonymity-preserving all-to-all reliable problem} that may be of independent interest.
Here, anonymity and comparable properties to reliable broadcast~\cite{Bracha1985} with $n$ broadcasters are ensured.
%Our solution called
%Anonymised All-to-all Reliable Broadcast Protocol (AARBP) terminates after two regular and one anonymous message delay and has identical message complexity to $n$ instances of Bracha's seminal reliable broadcast protocol~\cite{Bracha1987}.
By construction a solution to this problem with Bracha's reliable broadcast~\cite{Bracha1987}, AVCP can terminate after three regular and one anonymous message delay.
With this approach, our experimental results are promising---with 100 geo-distributed nodes, AVCP generally terminates in less than ten seconds.
We remark that, to ensure confidentiality of proposals until after termination, threshold encryption~\cite{Desmedt1992, Shoup2002}, in which a set threshold of participants must cooperate to decrypt any message, can be used at the cost of an additional message delay. \\

\noindent \textbf{Related constructions.}
We consider techniques without additional trusted parties.
Homomorphic tallying~\cite{Cramer1996} encodes opinions as values that are summed in encrypted form prior to decryption.
Such schemes that rely on a public bulletin board for posting votes~\cite{Groth2004} could use Byzantine consensus~\cite{Lamport1982} and make timing assumptions on vote submission to perform an election.
Unfortunately, homomorphic tallying is impractical when the pool of candidates is large and impossible when arbitrary.
Using a multiplicative homomorphic scheme~\cite{elgamal1985public}, decryption work is exponential in the amount of candidates, and additive homomorphic encryption like Paillier's scheme~\cite{paillier1999public} incur large (RSA-size) parameters and more costly operations.
Fully homomorphic encryption~\cite{gentry2009fully, brakerski2014efficient} is suitable in theory, but at present is untenable for computing complex circuits efficiently.
Self-tallying schemes~\cite{Groth2004}, which use homomorphic tallying, are not appropriate as at some point all participants must be correct, which is untenable with arbitrary (Byzantine) behaviour.
Constructions involving mix-nets~\cite{Chaum1981} allow for arbitrary ballot structure.
However, decryption mix-nets are not a priori robust to a single Byzantine or crash failure~\cite{Chaum1981}, and re-encryption mix-nets which use proofs of shuffle are generally slow in tallying~\cite{Adida2008, Kulyk2014}, requiring hours to tally in larger elections since $O(t)$ processes need to perform proofs of shuffle~\cite{neff2001verifiable} in sequence.
DC-nets and subsequent variations~\cite{Chaum1988, Golle2004a} are not sufficiently robust and generally require $O(n)$ message delays for an all-to-all broadcast.
On multi-party computation, general techniques like Yao's garbled-circuits~\cite{yao1986generate} incur untenable overhead given enough complexity in the structure of ballots.
Private-set intersection~\cite{freedman2004efficient, ye2008distributed} can be efficient for elections that require unanimous agreement, but do not generalise arbitrarily.
\\

\noindent \textbf{Roadmap.}
The paper is structured as follows.
Section~\ref{sec:model:model} provides preliminary definitions and specifies the model.
Section~\ref{sec:arb:arb} presents protocols and definitions required for our consensus protocol.
Section \ref{sec:avc:avc} presents our anonymity-preserving vector consensus solution.
\ifconference
\else We consider the case where regular and anonymous message channels are combined in Section~\ref{sec:anon}.
\fi
Section \ref{sec:exp} benchmarks AVCP on up to 100 geo-distributed located on three continents.
%For simplicity, we present AARBP and AVCP in Sections~\ref{sec:arb:arb} and~\ref{sec:avc:avc} without considering anonymity in the presence of timing attacks, which we consider in Section~\ref{sec:anon}.
%Section~\ref{sec:rw} describes the related work.
\ifconference
%Section \ref{sec:anon} describes the assumptions needed to ensure anonymity in the presence of timing attacks.
To conclude, Section~\ref{sec:conclusion} discusses the use of anonymous and regular channels in practice and future work. %and proposes directions for future research.
The full paper published on arXiv~\cite{avcpaper} provides constructions for all-to-all anonymity-preserving reliable broadcast and an election scheme, correctness proofs for our algorithms and additional experimental results.
\else
Section~\ref{sec:conclusion} concludes. %and proposes directions for future research.
\fi
\ifconference
\else
Appendix~\ref{sec:avcproof} proves the correctness of AVCP.
Appendix~\ref{sec:bvbincons} describes handling the termination of binary consensus instances used in AVCP.
Appendix~\ref{sec:arbproof} details and proves correct AARBP, presented using a single anonymous broadcast per process.
%, \ref{sec:bvbincons} and \ref{sec:avcproof} contain the correctness proof of AARBP, the binary consensus algorithm AVCP relies on~\cite{CGLR18}, and the correctness proof of AVCP, respectively.
Appendix~\ref{sec:evoting:abe} combines AVCP and threshold encryption, forming a voting scheme which is then analysed. 
% and characterises the resulting voting scheme.
Appendix~\ref{sec:bench:bench} evaluates our other distributed protocols and necessary cryptographic schemes.
\fi
%provides empirical evaluations of our protocols on up to 100 nodes distributed on three continents. %benchmarks and an evaluation of the practical feasibility of our solutions.

\section{Model} \label{sec:model:model}
% if using randomization, need to talk about trusted setup / static corruption of t processes
% without randomization - actively corruptive adv?
We assume the existence of a set of processes $P = \{p_1, \dots, p_n\}$ (where $|P| = n$, and the $i$th process is $p_i$), an adversary $A$ who corrupts $t < \frac{n}{3}$ processes in $P$, and a trusted dealer $D$ who generates the initial state of each process.
%Formally, we model each such party as a deterministic Turing machine.
For simplicity of exposition, we assume that cryptographic primitives are unbreakable. %authenticated model cite?
With concrete primitives that provide computational guarantees, each party could be modelled as being able to execute a number of instructions each message step bounded by a polynomial in a security parameter~$k$~\cite{Cachin2005}, in which case the transformation of our proofs
\ifconference
\cite{avcpaper}
\else \fi
is straight forward given the hardness assumptions required by the underlying cryptographic schemes. \newline % cite

%\textbf{Negligible functions:}
%A function $\epsilon(k)$ is called negligible if, for all $c > 0$, there exists $k_0$ s.t. $\epsilon(k) < \frac{1}{k^c}$ for all $k > k_0$.
%A computational problem is infeasible if any probabilistic polynomial-time algorithm solves it with negligible probability.
%Similarly, we say a property holds with overwhelming probability if it holds with probability $1 - \epsilon(k)$.

\noindent \textbf{Network:}
We assume that $P$ consists of asynchronous, sequential processes that communicate over reliable, point-to-point channels in an asynchronous network.
An \textit{asynchronous} process is one that executes instructions at its own pace.
%We note that sequential processes may multiplex instruction execution.
An \textit{asynchronous} network is one where message delays are unbounded.
A \textit{reliable} network is such that any message sent will eventually be delivered by the intended recipient.
We assume that processes can also communicate using reliable one-way \textit{anonymous} channels, which we soon describe.

Each process is equipped with the primitive ``$\lit{send}$ $M$ to $p_j$'', which sends the message $M$ (possibly a tuple) to process $p_j \in P$.
For simplicity, we assume that $p_j$ can send a message to itself.
A process receives a message $M$ by invoking the primitive ``$\lit{receive}$ $m$''.
Each process may invoke ``$\lit{broadcast}$ $M$'', which is short-hand for ``\textbf{for each} $p_i \in P$ \textbf{do} $\lit{send}$ $M$ to $p_i$ \textbf{end for}''.
Analogously, processes may invoke ``$\lit{anon\_send}$ $M$ to $p_j$'' and ``$\lit{anon\_broadcast}$ $M$'' over anonymous channels, which we characterise below.
%A process receives a message $m$ from a call to ``$\lit{anon\_send}$'' by invoking the primitive ``$\lit{anon\_receive}$ $m$''.

Since reaching consensus deterministically is impossible in failure-prone asynchronous message-passing systems~\cite{Fischer1985}, we assume that \textit{partial synchrony} holds among processes in $P$ in Section~\ref{sec:avc:avc}.
That is, we assume there exists a point in time in protocol execution, the global stabilisation time (GST), unknown to all processes, after which the speed of each process and all message transfer delays are upper bounded by a finite constant~\cite{Dwork1988}. \newline %unknown point in time?
%To mitigate de-anonymisation attacks based on observing message timing and order, we make an assumption on the independence of anonymous and regular message passing.
%The assumption and its motivation is described in Section~\ref{sec:anon}.

\noindent \textbf{Adversary:}
We assume that the adversary $A$ schedules message delivery over the regular channels, restricted to the assumptions of our model (such as partial synchrony).
For each $\lit{send}$ call made, $A$ determines when the corresponding $\lit{receive}$ call is invoked.
A portion of processes--- exactly $t < \frac{n}{3}$ members of $P$---are initially corrupted by $A$ and may exhibit Byzantine faults \cite{Lamport1982} over the lifetime of a protocol's execution.
That is, they may deviate from the predefined protocol in an arbitrary way.
We assume $A$ can see all computations and messages sent and received by corrupted processes.
%We assume processes can only be corrupted by the adversary.
A \textit{non-faulty} process is one that is not corrupted by $A$ and therefore follows the prescribed protocol.
$A$ can only observe $\lit{anon\_send}$ and $\lit{anon\_receive}$ calls made by corrupted processes.
$A$ cannot see the (local) computations that non-faulty processes perform.
We do not restrict the amount or speed of computation that $A$ can perform. \newline

\iffalse
\noindent \textbf{Anonymity assumption:}
Consider the following scenario.
Let $D \notin P$ be an adversary, modelled as a deterministic Turing machine.
Suppose that $p_i \in P$ is non-faulty and invokes ``$\lit{anon\_send}$ $m$ to $D$'', and $D$ invokes $\lit{anon\_receive}$ with respect to $m$.
$D$ is allowed to output a single guess, $g \in \{1, \dots, n\}$ as to the identity of $p_i$.
$D$ is allowed access to the message $m$, but not any of its message history or previous computations, and may not communicate further with other parties.
$D$ is then allowed to perform an arbitrary number of computations and make oracle calls as described below before outputting its guess.
Then $Pr(i = g) = \frac{1}{n - t}$.

As $A$ corrupts $t$ processes, the anonymity set~\cite{Danezis2008}---the set of processes by which the invoker of $\lit{anon\_send}$ is indistinguishable from---comprises $n - t$ non-faulty processes.
We implicitly assume that non-faulty processes do not reveal additional information about their identity based on $m$'s contents. \newline
\fi
\noindent \textbf{Anonymity assumption:}
Consider the following experiment.
Suppose that $p_i \in P$ is non-faulty and invokes ``$\lit{anon\_send}$ $m$ to $p_j$'', where $p_j \in P$ and is possibly corrupted, and $p_j$ invokes $\lit{anon\_receive}$ with respect to $m$.
No process can directly invoke $\lit{send}$ or invoke $\lit{receive}$ in response to a $\lit{send}$ call at any time.
$p_j$ is allowed to use $\lit{anon\_send}$ to message corrupted processes if it is corrupted, and can invoke $\lit{anon\_recv}$ with respect to messages sent by corrupted processes.
Each process is unable to make oracle calls (described below), but is allowed to perform an arbitrary number of local computations.
$p_j$ then outputs a single guess, $g \in \{1, \dots, n\}$ as to the identity of $p_i$.
Then for any run of the experiment, $Pr(i = g) \leq \frac{1}{n - t}$.

As $A$ can corrupt $t$ processes, the anonymity set~\cite{Danezis2008}, i.e. the set of processes $p_i$ is indistinguishable from, comprises $n - t$ non-faulty processes.
Our definition captures the anonymity of the anonymous channels, but does not consider the effects of regular message passing and timing on anonymity.
As such, we can use techniques to establish anonymous channels in practice with varying levels of anonymity with these factors considered.
\noindent \textbf{Traceable ring signatures:}
Informally, a ring signature \cite{Rivest2001, Fujisaki2007} proves that a signer has knowledge of a private key corresponding to one of many public keys of their choice without revealing the corresponding public key.
%Variant schemes with weakened anonymity levels \cite{Liu2004, Fujisaki2007} allows two messages signed by the same process to be linked or traced together.
%While linking does not reveal the signer's identity, tracing does.
%Hereafter, we consider \textit{traceable ring signatures} (or \textit{TRSs}), which are ring signatures that provide a traceability guarantee.
Hereafter, we consider \textit{traceable ring signatures} (or \textit{TRSs}), which are ring signatures that expose the identity of a signer who signs two different messages.
To increase flexibility, we can consider traceability with respect to a particular string called an \textit{issue}~\cite{Tsang2005}, allowing signers to maintain anonymity if they sign multiple times, provided they do so each time with respect to a different issue.
We now present relevant definitions of the ring signatures which are analogous to those of Fujisaki and Suzuki~\cite{Fujisaki2007}.
Let $\ms{ID} \in \{0, 1\}^{*}$, which we denote as a \textit{tag}.
We assume that all processes may query an idealised distributed oracle, which implements the following \textit{four} operations:
\begin{enumerate}
%[noitemsep,topsep=5pt,parsep=2pt,partopsep=5pt,itemindent=1em,listparindent=1em,leftmargin=1em]
\item $\sigma \leftarrow \lit{Sign}(i, \ms{ID}, m)$, which takes the integer $i \in \{1, \dots, n\}$, tag $\ms{ID} \in \{0, 1\}^{*}$ and message $m \in \{0, 1\}^{*}$, and outputs the signature $\sigma \in \{0, 1\}^{*}$. We restrict $\lit{Sign}$ such that only process $p_i \in P$ may invoke $\lit{Sign}$ with first argument~$i$.
  %CC: you could also make the first argument only "i" not "p_i"
		%Non-faulty processes can only query $\lit{Sign}$ with respect to their secret key $sk_i$.
\item $b \leftarrow \lit{VerifySig}(\ms{ID}, m, \sigma)$, which takes the tag $\ms{ID}$, message $m \in \{0, 1\}^{*}$, and signature $\sigma \in \{0, 1\}^{*}$, and outputs a bit $b \in \{0, 1\}$. All parties may query $\lit{VerifySig}$.
\item $out \leftarrow \lit{Trace}(\ms{ID}, m, \sigma, m', \sigma')$, which takes the tag $\ms{ID} \in \{0, 1\}^{*}$, messages $m, m' \in \{0, 1\}^{*}$ and signatures $\sigma, \sigma' \in \{0, 1\}^{*}$, and outputs $out \in \{0, 1\}^{*} \cup \{1, \dots, n\}$ (possibly corresponding to a process $p_i$). All parties may query $\lit{Trace}$.
	\item $x \leftarrow \lit{FindIndex}(\ms{ID}, m, \sigma)$ takes a tag $\ms{ID} \in \{0, 1\}^{*}$, a message $m \in \{0, 1\}^{*}$, and a signature $\sigma \in \{0, 1\}^{*}$, and outputs a value $x \in \{1, \dots, n\}$.
	$\lit{FindIndex}$ may not be called by any party, and exists only for protocol definitions.
	%	\item Trace(), which, given two signatures, outputs ``traced'' and the double signer's public key if they have signed two different messages with respect to the same issue, ``linked'' if the signer has double-signed with respect to the same message \textit{and} issue, and ``independent'' otherwise.
\end{enumerate}

\ifconference
The distributed oracle satisfies the following relations:
	\begin{itemize}
		\item $\lit{VerifySig}(\ms{ID}, m, \sigma) = 1$ $\Longleftrightarrow$ $\exists \ p_i \in P$ which invoked $\sigma \leftarrow \lit{Sign}(i, \ms{ID}, m)$.
		\item $\lit{Trace}$ is as below $\Longleftrightarrow$ $\sigma \leftarrow \lit{Sign}(i, \ms{ID}, m)$ and $\sigma' \leftarrow \lit{Sign}(i', \ms{ID}, m')$ where:
		\[\lit{Trace}(\ms{ID}, m, \sigma, m', \sigma')=
		\begin{cases}
			\text{``indep''} & \text{ if } i \neq i', \\
			\text{``linked''} & \text{ else if } m = m', \\
			i & \text{ otherwise } (i = i' \wedge m \neq m')\text{.}
		\end{cases}
		\]
	\item If adversary $D$ is given an arbitrary set of signatures $S$ and must identify the signer $p_i$ of a signature $\sigma \in S$ by guessing $i$,  $Pr(i = g) \leq \frac{1}{n - t}$ for any $D$.
	\end{itemize}
\else

We describe the behaviour of the idealised distributed oracle:
\begin{itemize}
%[noitemsep,topsep=5pt,parsep=2pt,partopsep=5pt,itemindent=1em,listparindent=1em,leftmargin=1em]
	\item \textbf{Signature correctness and unforgeability:} $\lit{VerifySig}(\ms{ID}, m, \sigma) = 1$ $\Longleftrightarrow$ there exists some process $p_i \in P$ that previously invoked $\lit{Sign}(i, \ms{ID}, m)$ and obtained $\sigma$ as a response. Unforgeability is captured by the ``$\Rightarrow$'' claim.
	\item \textbf{Traceability and entrapment-freeness:} The function $\lit{Trace}$ behaves as defined below $\Longleftrightarrow$ $\sigma \leftarrow \lit{Sign}(i, \ms{ID}, m)$ and $\sigma' \leftarrow \lit{Sign}(i', \ms{ID}, m')$:
		\[\lit{Trace}(\ms{ID}, m, \sigma, m', \sigma')=
		\begin{cases}
			\text{``indep''} & \text{ if } i \neq i', \\
			\text{``linked''} & \text{ else if } m = m', \\
			i & \text{ otherwise } (i = i' \wedge m \neq m')\text{.}
		\end{cases}
		\]
		Traceability is captured by the ``$\Leftarrow$'' claim: if $m = m'$, then the messages are linked, otherwise, the identity of the signing process $p_i = p_{i'}$ is exposed. Entrapment-freeness is captured by the ``$\Rightarrow$'' claim: loosely, processes cannot be falsely accused of double-signing.
%	\item \textbf{Signature anonymity (old):} Suppose $\sigma \leftarrow \lit{Sign}(i, \ms{ID}, m)$, where $p_i \in P$ is non-faulty. Then, it is impossible for the adversary to determine the value of $i$ with probability greater than $\frac{1}{n - t}$.
	\item \textbf{Signature anonymity:} Consider the following scenario.
		Let $D$ be an adversary.%, modelled as a deterministic Turing machine.
		Let $i \in \{1, \dots, n\}$ be the identifier of a non-faulty process.
		Suppose that $D$ is given a set of signatures of arbitrary length $S = \{\sigma_1, \dots\}$ such that, for each pair $(\sigma_x, \sigma_y) =  (\lit{Sign}(i, \ms{ID}_x, m_x), \lit{Sign}(i, \ms{ID}_y, m_y))$, either $\ms{ID}_x = \ms{ID}_y$ and $m_x = m_y$ holds or $\ms{ID}_x \neq \ms{ID}_y$ holds.
		$D$ is allowed access to the set $S$, but not any of its previous computations, and may not communicate with other parties.
		Then, suppose that $D$ is required to output a value $g \in \{1, \dots, n\}$, corresponding to its guess of the value of $i$, after performing an arbitrary number of computations and oracle calls subject to the restrictions described above.
		Then $Pr(i = g) = \frac{1}{n - t}$.
		Signature anonymity ensures that a process that signs with respect to a single message $m$ per tag $\ms{ID}$ maintains anonymity.
%	\item \textbf{Traceability:} $\forall k, n \in \mathbb{N}$, $i, i' \in \{1..n\}$, $issue, m, m' \in \{0, 1\}^{*}$, if $\sigma \leftarrow \lit{Sign}_{sk_i}(L, m)$ and $\sigma \leftarrow \lit{Sign}_{sk_i'}(L, m')$, then with overwhelming probability:
%		\[\lit{Trace}(L, m, \sigma, m', \sigma')=
%		\begin{cases}
%			\text{``indep''} & \text{ if } i \neq i', \\
%			\text{``linked''} & \text{ else if } m = m', \\
%			pk_i & \text{ otherwise } (m \neq m')\text{.}
%		\end{cases}
%		\]
%		Given a process signs two messages $m, m'$, if $m = m'$, then the messages are linked; otherwise, their identity is exposed. This definition characterises the correctness of traceability, but does not ensure signatures cannot be forged or that processes cannot be falsely accused of double-signing. Definitions of properties that ensure this, namely tag-linkability, anonymity and exculpability, can be found in \cite{Fujisaki2007}. We assume that these properties hold.
%	\item \textbf{Index deduction:} Fix $P$ and the corresponding keying material $(pk_i, sk_i)_{i = \{1..n\}}$. Then, $\forall k \in \mathbb{N}$, $n \in \mathbb{N}$, $i \in \{1..n\}$, $\ms{issue} \in \{0, 1\}^{*}$, $m \in \{0, 1\}^{*}$, if $\sigma \leftarrow \lit{Sign}_{sk_i}(L, m)$, then with overwhelming probability $\lit{FindIndex}(sk_1, \cdots, sk_n, L, m, \sigma) = i$. This property ensures that $\lit{FindIndex}$ can always deduce the index of the signer provided that the signature is well-formed. Other behaviour is undefined; for example, when $\lit{VerifySig}(L, m, \sigma) = 0$.

\end{itemize}

With respect to $\lit{FindIndex}(\ms{ID}, m, \sigma)$, $i \in \{1, \dots, n\}$ is outputted if and only if $\sigma$ is the result of process $p_i$ having previously queried $\lit{Sign}(i, \ms{ID}, m)$.
The unforgeability property implies \textit{signature uniqueness}: If calls $\sigma \leftarrow \lit{Sign}(i, \ms{ID}, m)$ and $\sigma' \leftarrow \lit{Sign}(j, \ms{ID}, m')$ are made, then $\sigma \neq \sigma'$ holds unless $i = j$ and $m = m'$.
\fi

The concrete scheme proposed by Fujisaki and Suzuki~\cite{Fujisaki2007} computationally satisfies these properties in the random oracle model~\cite{Bellare1993} provided the Decisional Diffie-Hellman problem is intractable.
It has signatures of size $O(kn)$, where $k$ is the security parameter.
%\footnote{Alternatively, the scheme by Gu et al. boasts $O(k)$-sized signatures~\cite{gu2019efficient}, but requires a more elaborate trusted setup and interaction for tracing.}
To simplify the presentation, we assume that its properties hold unconditionally in the following.
%In implementation, Fujisaki and Suzuki's scheme~\cite{Fujisaki2007} requires $\ms{ID} = \ms{issue} \mid \mid pk_1 \mid \mid \cdots \mid \mid pk_n$, where $pk_i$ denotes the $i$th process' public key, and $\ms{issue}$ is a string.

%We remark that tag-linkability and exculplability imply standard signature unforgeability, as detailed in \cite{Fujisaki2007}. 

%We emphasise that if a Trace() operation outputs ``linked'', the identity of the signer is not revealed.
%Then, the following security properties, presented informally for readability, are fulfilled by the scheme:
%
%\begin{itemize}
%	\item \textbf{TRS-Public-traceability}: Two messages signed by with the same private key with respect to the same tag can be traced in a pairwise fashion.
%	\item \textbf{TRS-Tag-linkability}: Two signatures generated by the same signer are linked. As such, given a ring of size $n$, at most $n$ messages that are not linked can be produced for a particular tag.
%	\item \textbf{TRS-Anonymity}: Given the signer does not sign two different messages with respect to the same tag, their identity is indistinguishable from others in the ring. Further, it is infeasible to determine if two signatures generated with respect to distinct tags are generated by the same signer.
%	\item \textbf{TRS-Exculpability}: A honest ring member cannot be accused of signing twice with respect to the same tag.
%\end{itemize}

%\section{Anonymised all-to-all reliable broadcast} \label{sec:arb:arb}
\section{Communication primitives} \label{sec:arb:arb}
\ifconference
\else
In this section, we detail communication primitives that are invoked in our consensus algorithm presented in Section~\ref{sec:avc:avc}.
In particular, we describe and present a mechanism for processes to replace communication calls over regular channels to ones over anonymous channels, we present the binary consensus problem, and define the properties of anonymity-preserving all-to-all reliable broadcast, an extension of reliable broadcast, initially described by Bracha~\cite{Bracha1987}.
\fi
\subsection{Traceable broadcast}
Suppose a process $p$ wishes to anonymously message a given set of processes $P$.
By invoking anonymous communication channels, $p$ can achieve this, but $P \setminus \{p\}$ is unable to verify that $p$ resides in $P$, and so cannot meaningfully participate in protocol execution.
By using (traceable) ring signatures, $p$ can verify its membership in $P$ over anonymous channels without revealing its identity.
To this end, we outline a simple mechanism to replace the invocation of $\lit{send}$ and $\lit{receive}$ primitives (over regular channels) with calls to ring signature and anonymous messaging primitives $\lit{anon\_send}$ and $\lit{anon\_receive}$.

\ifconference
Let $(\ms{ID}, \TAG, \lit{label}, m)$ be a tuple for which $p_i \in P$ has invoked $\lit{send}$ with respect to.
Instead of invoking $\lit{send}$, $p_i$ first calls $\sigma \leftarrow \lit{Sign}(i, T, m)$, where $T$ is a tag uniquely defined by $\TAG$ and $\lit{label}$.
Then, $p_i$ invokes $\lit{anon\_send}$ $M$, where $M = (\ms{ID}, \TAG, \lit{label}, m, \sigma)$.
Upon an $\lit{anon\_receive}$ $M$ call by $p_j \in P$, $p_j$ verifies that $\lit{VerifySig}(T, m, \sigma) = 1$ and that they have not previously received $(m', \sigma')$ such that $\lit{Trace}(T, m, \sigma, m', \sigma') \neq$ ``indep''.
Given this check passes, $p_j$ invokes $\lit{receive}$ $(\ms{ID}, \TAG, \lit{label}, m)$.

By the properties of the anonymous channels and the signatures, it follows that anonymity as defined in the previous section holds with additional adversarial access to the distributed oracles.
We present the proof in the full paper~\cite{avcpaper}.
\else

\sloppy{
\paragraph{State.} Each process in $P$ tracks $\ms{msg\_buf}[]$, which maps a key (an ordered list of strings) to a set of messages of the form $(m, \sigma)$, where $\sigma$ is the output of a call to $\lit{Sign}$.
Each set is initially empty.
Messages of the form $(\ms{ID}, \TAG, \lit{label}, m, \sigma)$ (where $\lit{label} = (l_1, \dots, l_k)$ for some $k \geq 0$) are deposited as $\ms{msg\_buf}[\ms{ID}, \TAG, \lit{label}] \leftarrow \ms{msg\_buf}[\ms{ID}, \TAG, \lit{label}] \cup \{m, \sigma\}$.
%Considering the example $(\ms{ID}, \TAG, r, \lit{label}, m, \sigma)$, $\{m, \sigma\}$ is added to $\ms{msg\_buf}[\TAG, r, \lit{label}]$.
}

\begin{algorithm}[ht]
	% state in pseudo-code
	\caption{Traceable broadcast} %as invoked by $p'$}
	{\footnotesize
	\label{alg:tb}
	\begin{algorithmic}[1]
		\Upon{invocation of $\lit{broadcast}$ $(\ms{ID}, \TAG, \lit{label},  m)$}
		\State $\sigma \leftarrow \lit{Sign}(i, T, m)$ \label{line:tb:sign} \Comment{$T = \ms{ID}||\TAG||l_1||\cdots||l_k$, where $\lit{label} = (l_1, \ldots, l_k)$}
		\State $\lit{anon\_broadcast}$ $(\ms{ID}, \TAG, \lit{label}, m, \sigma)$ \label{line:tb:anonbcast}
	\EndUpon
	\Upon{invocation of $\lit{anon\_receive}$ $(\ms{ID}, \TAG, \lit{label}, m, \sigma)$}
	\State $\ms{valid} \leftarrow (\lit{VerifySig}(T, m, \sigma) = 1)$ \label{line:tb:isvalid}
	\If{\ms{valid}}
	\For{\textbf{each} $(m', \sigma') \in \ms{msg\_buf}[\ms{ID}, \TAG, \lit{label}]$} \label{line:tb:for}
	\If{$\lit{Trace}(T, m, \sigma, m', \sigma') \neq$ ``indep''}\label{line:tb:trace}
	\State $\ms{valid} \leftarrow \lit{false}$ \Comment{Double-signing detected}
	\State \textbf{break}
	\EndIf
	\EndFor
	\State $\ms{msg\_buf}[\ms{ID}, \TAG, \lit{label}] \leftarrow \ms{msg\_buf}[\ms{ID}, \TAG, \lit{label}] \cup \{(\ms{m, \sigma})\}$ \label{line:tb:store}
	\If{\ms{valid}}
	\State $\lit{receive}$ $(\ms{ID}, \TAG, \lit{label}, m)$ \label{line:tb:recv}
	\EndIf
	\EndIf
	\EndUpon
	\end{algorithmic}
	}
\end{algorithm}

\paragraph{Protocol.}
Suppose that a process $p$ wishes to invoke $\lit{broadcast}$ with respect to a tuple $(\ms{ID}, \TAG, \lit{label}, m)$
To replace $\lit{broadcast}$ (resp. $\lit{send}$) calls, a process first signs the value $m$  with respect to string $T = \ms{ID} || \TAG || l_1 || \dots || l_k$ for some $k \geq 0$ (at line~\ref{line:tb:sign}), outputting $\sigma$.
We assume every $T$ in Algorithm~\ref{alg:tb} takes this form.
Then, $p$ invokes $\lit{anon\_broadcast}$ (resp. $\lit{anon\_send}$) with respect to $(\ms{ID}, \TAG, \lit{label}, m, \sigma)$ (line~\ref{line:tb:anonbcast}).

On receipt of a tuple $(\ms{ID}, \TAG, \lit{label}, m, \sigma)$ via anonymous channels, $p$ first checks that the signature $\sigma$ is well-formed (line~\ref{line:tb:isvalid}).
Given this, all message/signature pairs in the set $\ms{msg\_buf}[\ms{ID}, \TAG, \lit{label}]$ are compared to the received tuple via $\lit{Trace}$ calls (line~\ref{line:tb:trace}).
Irrespective of the outcome, the received tuple is stored in $\ms{msg\_buf}[]$ (line~\ref{line:tb:store}).
Provided that the tuple is independent with respect to the signer of all messages stored thus far in $\ms{msg\_buf}[\ms{ID}, \TAG, \lit{label}]$, $(\ms{ID}, \TAG, \lit{label}, m)$ is considered receipt over regular channels (line~\ref{line:tb:recv}).

We now prove that the protocol is correct and ensures anonymity under the assumptions of the model.

\begin{lemma}
	\sloppy{
	\label{anon-bcast}
	Suppose $p_i \in P$ is non-faulty and invokes $\lit{broadcast}$ $(\ms{ID}, \TAG, \lit{label}, m)$.
	Suppose $p_i$ has previously invoked $\lit{broadcast}$ an arbitrary number of times.
	Then, all non-faulty processes will eventually invoke $\lit{receive}$ $(\ms{ID}, \TAG, \lit{label}, m)$.
	Moreover, $p_i$ preserves anonymity which is modelled as follows, provided that $p_i$ does not invoke $\lit{broadcast}$ $(\ms{ID}, \TAG, \lit{label}, m)$ and $\lit{broadcast}$ $(\ms{ID}', \TAG', \lit{label}, m')$ where $(\ms{ID}, \TAG, \lit{label}) = (\ms{ID}', \TAG', \lit{label}')$.
	Suppose that an adversary $A$ is allowed to output a single guess, $g \in \{1, \dots, n\}$ as to the identity of $p_i$ a finite amount of time after invoking $\lit{receive}$ at a corrupted process of its choice.
	$A$ is then allowed to perform an arbitrary number of computations and make oracle calls described in Section~\ref{sec:model:model} before outputting its guess.
	Then for any $A$, $Pr(i = g) \leq \frac{1}{n - t}$.
}
\end{lemma}
\begin{proof}
	\sloppy{
Process $p_i$ invokes $\lit{Sign}(i, \ms{ID}||\TAG||l_1||\cdots||l_k, m)$, producing the signature $\sigma$ (line~\ref{line:tb:sign}).
Then, $p_i$ invokes ``$\lit{anon\_broadcast}$ $(i, \ms{ID}||\TAG||l_1||\cdots||l_k, m, \sigma)$'' (line~\ref{line:tb:anonbcast}).
By construction of the model, $A$ is unable to observe the existence or timing of $p_i$'s call to $\lit{anon\_broadcast}$.
By the anonymity assumption, $A$ is not able to produce a guess $g \neq \frac{1}{n - t}$ when unable to make oracle calls, since it is only aware that the $t$ processes it has corrupted are different from $p_i$.
By the reliability of the network, $A$ will eventually invoke $\lit{anon\_receive}$ with respect to $m$ at one of its $t$ corrupted processes, as will every non-faulty process.

By signature correctness, each process will deduce that $\lit{VerifySig}(\ms{ID}||\TAG||l_1||\cdots||l_k, m, \sigma) = 1$ at line~\ref{line:tb:isvalid}.
By definition of $\lit{Sign}$, $A$ cannot invoke $\lit{Sign}(i, \ms{ID}||\TAG||l_1||\cdots||l_k, m')$ nor produce a value $\sigma'$ such that $\lit{VerifySig}(\ms{ID}||\TAG||l_1||\cdots||l_k, m', \sigma') = 1$ for any $m'$.
This holds for any tuple of the form $(\ms{ID}, \TAG, l_1, \dots, l_k)$.
Thus, every non-faulty process will progress to line~\ref{line:tb:for}.
By traceability and entrapment-freeness, $A$ cannot produce $(m', \sigma')$ such that $\lit{Trace}(\ms{ID}||\TAG||l_1||\cdots||l_k, m, \sigma, m', \sigma') \neq$ ``indep''.
In particular, if $m' \neq m$, then $A$ cannot expose the identity of $p_i$ through a $\lit{Trace}$ call, as $\lit{Trace}$ will never output $i$.
Thus, every non-faulty process will reach line~\ref{line:tb:recv}.
By signature anonymity, and the fact that $p_i$ never signs with respect to two tuples such that $(\ms{ID}', \TAG', l_{1'}, \dots, l_{k'}) = (\ms{ID}, \TAG, l_1, \dots, l_k)$, no previous computation or messages can aid $A$'s guess to be such that $g \neq \frac{1}{n - t}$.
}
\end{proof}
\fi
Hereafter, we assume that calls to primitives $\lit{send}$ and $\lit{receive}$ are handled by the procedure presented in this subsection unless explicitly stated otherwise.

\subsection{Binary consensus}
Broadly, the binary consensus problem involves a set of processes reaching an agreement on a binary value $b \in \{0, 1\}$.
We first recall the definitions that define the binary Byzantine consensus (BBC) problem as stated in \cite{CGLR18}. 
In the following, we assume that every non-faulty process proposes a value, and remark that only the values in the set $\{0, 1\}$ can be decided by a non-faulty process.
\begin{enumerate}
	\item \textbf{BBC-Termination:} Every non-faulty process eventually decides on a value.
	\item \textbf{BBC-Agreement:} No two non-faulty processes decide on different values.
	\item \textbf{BBC-Validity:} If all non-faulty processes propose the same value, no other value can be decided.
\end{enumerate}
We present the safe, non-terminating variant of the binary consensus routine from~\cite{CGLR18} in Algorithm~\ref{alg:obc}.
As assumed in the model (Section~\ref{sec:model:model}), the terminating routine relies on partial synchrony between processes in $P$.
The protocols execute in asynchronous rounds.

\begin{algorithm}[h]
	% state in pseudo-code
	\caption{Safe binary consensus routine}
	\label{alg:obc}
	\begin{algorithmic}[1]
	{\small
		\Part{$\lit{bin\_propose}(v)$} {
		\State $est \leftarrow v$; $r \leftarrow 0$;
		\Repeat
		\State $r \leftarrow r + 1$;
		\State $\lit{BV-broadcast}($\EST$, r, est)$ \label{line:obc:bvbcast}
		\WUntil $(\ms{bin\_values}[r] \neq \emptyset)$ \label{line:obc:bvalsnonempty} \EndWUntil
		\State $\lit{broadcast}$ $($\AUX$, r, \ms{bin\_values}[r])$ \label{line:obc:bcastaux}
		\WUntil{$(|\ms{values_r}| \geq n - t) \land (val \in \ms{bin\_values}[r]$ for all $val \in \ms{values_r})$}
	\label{line:obc:conds}
%		\WUntil received $n - t$ messages from distinct processes \Statex
%		\ \ \ \ \ \ \ \ $($\AUX$, r, \ms{b\_val}_{p(1)}$), $\cdots$, $($\AUX$, r, \ms{b\_val}_{p(n - t)})$ \Statex
%		\ \ \ \ \ \ \ \ $values = \cup_{1 \leq x \leq n - t}\ms{b\_val_{p(x)}}$ s.t. $\emptyset \neq values \subseteq \ms{bin\_values[r]}$
%		messages $($\AUX$, r, \ms{b\_val_p}(1))$, $\cdots$, $($\AUX$, r, \ms{b\_val_p}(n - t))$ have been received from $(n - t)$ processes $p(x), 1 \leq x \leq n - t$ with contents s.t. $\exists$ $values \neq \emptyset$ where the following two conditions hold: 
%		\begin{itemize}
%		\item $values = \cup_{1 \leq x \leq n - t}\ms{b\_val_{p(x)}}$
%		\item $values \subseteq \ms{bin\_values[r]}$ 
%		\end{itemize}
		\EndWUntil
		\State $b \leftarrow r \pmod{2}$ \label{line:obc:brmod2}
		\If{$val = w$ for all $val \in \ms{values_r}$ where $w \in \{0, 1\}$}
		\State $est \leftarrow w$; \label{line:obc:esttov}
		\If{$w = b$} \label{line:obc:isvb}
		\State $\lit{decide}(v)$ if not yet invoked $\lit{decide}()$ \label{line:obc:decide}
		\EndIf
		\Else 
		\State $est \leftarrow b$ \label{line:obc:esttob}
		\EndIf
		\EndRepeat
		}\EndPart
		\Upon{intial receipt of $(\AUX, r, b)$ for some $b \in \{0, 1\}$ from process $p_j$}
		\State $\ms{values_r}.\lit{append}(b)$
		\EndUpon
		\Statex
		\Part{$\lit{BV-broadcast}($\EST$,r,v_i)$} {
			\State $\lit{broadcast}$ $($\EST$,r,v_i)$ \label{line:obc:bcastinbv}
		}\EndPart
		\Upon{receipt of $($\EST$,r,v)$}
		\If{$($\EST$,r,v)$ received from $(t + 1)$ processes and not yet broadcast} \label{line:obc:tplusone}
		\State $\lit{broadcast}$ $($\EST$,r,v)$		\EndIf
		\If{$($\EST$,r,v)$ received from $(2t + 1)$ processes} \label{line:obc:twotonecond}
		\State $\ms{bin\_values}[r] \leftarrow \ms{bin\_values}[r] \cup \{v\}$ \label{line:obc:bvdeliver}
		\EndIf
		\EndUpon
		}
	\end{algorithmic}
\end{algorithm}

\paragraph{State.} A process keeps track of a binary value $est \in \{0, 1\}$, corresponding to a process' current estimate of the decided value, arrays $bin\_values[1..]$, a round number $r$ (initialised to 0), an auxiliary binary value $b$, and lists of (binary) values $\ms{values_r}$, $r = 1, 2, \dots$, each of which are initially empty. 

\paragraph{Messages.} Messages of the form $($\EST$, r, b)$ and $($\AUX$, r, b)$, where $r \geq 1$ and $b \in \{0, 1\}$, are sent and processed by non-faulty processes. Note that we have omitted the dependency on a label $label$ and identifier $\ms{ID}$ for simplicity of exposition.
%In the original paper~\cite{CGLR18}, these messages take the form \EST$[r](b)$ and \AUX$[r](b)$.

\paragraph{BV-broadcast.}
To exchange \EST \ messages, the protocol relies on an all-to-all communication abstraction, BV-broadcast \cite{Mostefaoui2015},
\ifconference which is presented in Algorithm~\ref{alg:obc}.
\else which satisfies the following properties for a given round $r \geq 1$ that a process is executing:
\begin{itemize}
\item {\bf BV-Obligation:} If at least $(t + 1)$ non-faulty processes BV-broadcast the same value $v$, then $v$ will eventually be added to the set $\ms{bin\_values}$ of each non-faulty process.
\item {\bf BV-Justification:} If $p$ is non-faulty and $v \in \ms{bin\_values}[r]$, then $v$ must have been BV-broadcast by a non-faulty process.
\item {\bf BV-Uniformity:} If $v$ is added to a non-faulty process $p$'s $\ms{bin\_values}[r]$ set, then eventually $v \in \ms{bin\_values}[r]$ at every non-faulty process.
\item {\bf BV-Termination:} Eventually, $\ms{bin\_values}[r]$ becomes non-empty for every non-faulty process.
\end{itemize}

Primarily, BV-broadcast serves to filter values that are only proposed by faulty processes (BV-Obligation and BV-Justification), to ensure progress (BV-Termination), and to work towards agreement (BV-Uniformity).
\fi
When a process adds a value $v \in \{0, 1\}$ to its array $\ms{bin\_values}[r]$ for some $r \geq 1$, we say that $v$ was BV-delivered.

\paragraph{Functions.}
Let $b \in \{0, 1\}$.
In addition to $\lit{BV-broadcast}$ and the communication primitives in our model (Section~\ref{sec:model:model}), a process can invoke $\lit{bin\_propose}(b)$ to begin executing an instance of binary consensus with input $b$, and $\lit{decide}(b)$ to decide the value $b$.
In a given instance of binary consensus, these two functions may be called exactly once.
In addition, the function $list.\lit{append}(v)$ appends the value $v$ to the list $list$.

\ifconference
	To summarise Algorithm~\ref{alg:obc}, the BV-broadcast component ensures that only a value proposed by a correct process may be decided, and the auxiliary broadcast component ensures that enough processes have received a potentially decidable value to ensure agreement.
	The interested reader can verify the correctness of the protocol and read a thorough description of how it operates in~\cite{CGLR18}, where the details of the corresponding terminating protocol also reside.
\else
\paragraph{Protocol description.}
Upon the invocation bin\_propose($b$), where $b \in \{0, 1\}$, a process will enter a sequence of asynchronous rounds.

In a given round, a process will invoke BV-broadcast (line \ref{line:obc:bvbcast}), broadcasting its current estimate (line \ref{line:obc:bcastinbv}) for the round $r$, which is $est$.
	In an instance of BV-broadcast, after receiving a binary value from $t + 1$ different processes, a process will broadcast the value if not yet done (line \ref{line:obc:tplusone}).
	Eventually, a process will BV-deliver a value $v$ upon reception from $2t + 1$ different processes, fulfilling the condition at line \ref{line:obc:twotonecond}.
	In doing so, a process appends $v$ to its set $\ms{bin\_values}[r]$ (line \ref{line:obc:bvdeliver}), and we note that $\ms{bin\_values}[r]$ is not necessarily in its final form at this point in time.

	Since $\ms{bin\_values}[r]$ is now a singleton set (after line \ref{line:obc:bvalsnonempty}), non-faulty processes will broadcast the value $v$ contained in $\ms{bin\_values}[r]$ (line \ref{line:obc:bcastaux}).
	Then, processes wait until they can form a list $\ms{values_r}$ such that $val \in \ms{bin\_values}[r]$ for all $val \in \ms{values_r}$ and $\ms{values_r}$ is formed from at least $n - t$ $($\AUX$, r, b)$ messages sent by distinct processes that they have received (line \ref{line:obc:conds}).
	This ensures that enough processes have sent messages to be able to (potentially) decide, and that only values BV-broadcast previously are considered candidates for being decided.

Then, processes attempt to decide a value via local computation. $b$ is set to $r \pmod{2}$ (line \ref{line:obc:brmod2}), and then each process checks the following:
\begin{itemize}
	\item If each element of $\ms{values_r}$ is $w$, the estimate for the next round is set to $w$ (line \ref{line:obc:esttov}).
	\item Given the above, $w$ is then decided if $r \pmod{2} = w$, (lines \ref{line:obc:isvb} and \ref{line:obc:decide}).
		If this does not hold, the value can be decided in the next round provided that $\ms{values_{r + 1}}$ is uniform at line \ref{line:obc:esttov}.
	\item Else, $\ms{values_r}$ contains both 0's and 1's (and $\ms{bin\_values}[r] = \{0, 1\}$), and $est$ is set to be $r \pmod{2}$ (line \ref{line:obc:esttob}), as a process cannot decide at this point.
\end{itemize}

Note that, upon invocation of $\lit{decide}()$, processes still participate in the protocol, enabling other processes who may not have decided to decide.
The interested reader may verify the correctness of this protocol, and the corresponding terminating, partially-synchronous protocol in~\cite{CGLR18}.
\fi
%We also present the relevant components of the terminating variant that require handling in AVCP, and describe how to ensure termination.

\subsection{Anonymity-preserving all-to-all reliable broadcast}
%In this section, we introduce the problem of anonymised all-to-all reliable broadcast, which extends the definition of reliable broadcast \cite{Bracha1987, Raynal2018} to hide the identities of designated senders.
%As such, recipients can 
To 
reach eventual agreement in the presence of Byzantine processes without revealing 
%the broadcaster's identity, 
who proposes what, 
we introduce the
%Then, we describe a protocol, 
\textit{anonymity-preserving all-to-all reliable broadcast} problem that preserves the anonymity of each honest sender reliably broadcasting.
% CC: be more explicit on what it achieves
In this primitive, all processes are assumed to (anonymously) broadcast a message, and
all processes deliver messages over time. It ensures that all honest processes
always receive the same message from one specific sender while
hiding the identity of any non-faulty sender.

Let $\ms{ID} \in \{0, 1\}^{*}$ be an \textit{identifier}, a string that uniquely identifies an instance of AARB-broadcast.
Let $m$ be a message, and $\sigma$ be the output of the call $\lit{Sign}(i, T, m)$ for some $i \in \{1, \dots, n\}$, where $T = f(ID, \lit{label})$ is as in Algorithm~\ref{alg:obc}.
Each process is equipped with two operations, ``$\lit{AARBP}$'' and ``$\lit{AARB-deliver}$''.
$\lit{AARBP}[\ms{ID}](m)$ is invoked once with respect to $\ms{ID}$ and any message $m$, denoting the beginning of a process' execution of AARBP with respect to $\ms{ID}$.
$\lit{AARB-deliver}[ID](m, \sigma)$ is invoked between $n - t$ and $n$ times over the protocol's execution.
When a process invokes $\lit{AARB-deliver}[\ms{ID}](m, \sigma)$, they are said to ``AARB-deliver'' $(m, \sigma)$ with respect to $\ms{ID}$.
Let $T = f(ID, \lit{label})$ be as in Algorithm~\ref{alg:obc}.
Then, given $t < \frac{n}{3}$, we define a protocol that implements anonymity-preserving all-to-all reliable broadcast (AARB-broadcast) with respect to $\ms{ID}$ as satisfying the following \textit{six} properties:
%\begin{enumerate}[noitemsep,topsep=5pt,parsep=2pt,partopsep=5pt,itemindent=1em,listparindent=1em,leftmargin=1em]
\begin{enumerate}
	\item {\bf AARB-Signing:}
		If a non-faulty process $p_i$ AARB-delivers a message with respect to $\ms{ID}$, then it must be of the form $(m, \sigma)$, where a process $p_i \in P$ invoked $\lit{Sign}(i, T, m)$ and obtained $\sigma$ as output.
	\item {\bf AARB-Validity:}
%CC: integrity (authenticity)
		\sloppy{
		Suppose that a non-faulty process AARB-delivers $(m, \sigma)$ with respect to $\ms{ID}$.
		Let $i = \lit{FindIndex}(T, m, \sigma)$ denote the output of an idealised call to $\lit{FindIndex}$.
		Then if $p_i$ is non-faulty, $p_i$ must have anonymously broadcast $(m, \sigma)$.}
	\item {\bf AARB-Unicity:}
%CC: what does this say? that each message in the output is from a distinct sender...? not clear to me yet; also, FindIndex() is from the implementation of ring sigs, the ARB primitive alone should not refer to the implementation; last but not least, this should consider $delivered$ from one process (p) or not from many ("they")
		\sloppy{Consider any point in time in which a non-faulty process $p$ has AARB-delivered more than one tuple with respect to $\ms{ID}$. Let $delivered = \{(m_1, \sigma_1), \dots, (m_l, \sigma_l)\}$, where $|delivered| = l$, denote the set of these tuples. For each $i \in \{1, \dots, l\}$, let $out_i = \lit{FindIndex}(T, m_i, \sigma_i)$ denote the output of an idealised call to $\lit{FindIndex}$. Then for all distinct pairs of tuples $\{(m_i, \sigma_i), (m_j, \sigma_j)\}$, $out_i \neq out_j$.}
		%to the set of ARB-delivered messages A non-faulty process ARB-delivers at most one message from a particular process.
	\item {\bf AARB-Termination-1:} 
%CC: validity
		If a process $p_i$ is non-faulty and invokes $\lit{AARBP}[ID](m)$, all the non-faulty processes eventually AARB-deliver $(m, \sigma)$ with respect to $\ms{ID}$, where $\sigma$ is the output of the call $\lit{Sign}(i, T, m)$.
	\item {\bf AARB-Termination-2:} 
%CC: goes towards agreement (= consistency + totatlity), but needs the instance identifier
		If a non-faulty process AARB-delivers $(m, \sigma)$ with respect to $\ms{ID}$, then all the non-faulty processes eventually AARB-deliver $(m, \sigma)$ with respect to $\ms{ID}$.

%	\item {\bf ARB-Anonymity}. Let $0 \leq f \leq t$ be the number of \textit{actual} faulty processes, i.e. processes corrupted by the adversary, over the course of protocol execution. Then, consider the message $m$ ARB-broadcast by a non-faulty process $p$. Then, the probability $Pr$ of the adversary correctly determining the identity of $p$ satisfies:
%		Let $M = \{m_{f(1)}, \cdots, m_{f(n - f)}\}$ denote these messages, each of which has been ARB-delivered, where $f(i) \in \{1..n\}$ corresponds to the unique, secret index of a process in $P$ who produced $m_{f(i)}$.
%		Given a message $m \in M$, it is impossible to determine the value of $f(i)$ s.t. $m = m_{f(i)}$ with probability $p$ s.t. $p > \frac{1}{n - f}$.
%		\[Pr < \frac{1}{n - f} + \epsilon(k).\]
	% TODO: define more specifically in terms of Sign/etc operations as per Fujisaki's papers, then state somewhere that the definition of anonymity can depend on the signature scheme (i.e. LRS can avoid Trace() function calls)
\end{enumerate}
We require AARB-Signing to ensure that the other properties are meaningful.
Since messages are anonymously broadcast, properties refer to the index of the signing process determined by an idealised call to $\lit{FindIndex}$.
In spirit, AARB-Validity ensures if a non-faulty process AARB-delivers a message that was signed by a non-faulty process $p_i$, then $p_i$ must have invoked AARBP.
Similarly, AARB-Unicity ensures that a non-faulty process will AARB-deliver at most one message signed by each process.
We note that AARB-Termination-2 is
critical
for
consensus, as without it, different processes may AARB-deliver different messages produced by the \textit{same} process, as in the two-step algorithm implementing no-duplicity broadcast \cite{Bracha1983, Raynal2018}.
Finally, we state the anonymity property:
%\begin{enumerate}[noitemsep,topsep=5pt,parsep=2pt,partopsep=5pt,itemindent=1em,listparindent=1em,leftmargin=1em]
\begin{enumerate}
\setcounter{enumi}{5}
%	\item {\bf AARB-Anonymity: (old)}
%		\sloppy{
%		Let $f$ be the actual number of faulty processes, where $0 \leq f \leq n - 2$. % TODO - change?
%		Let $(m, \sigma)$ be a tuple anonymously broadcast by a non-faulty process $p_i$.
%		Let $out = \lit{FindIndex}(sk_1, \cdots, sk_n, L, m, \sigma)$ denote the output of an idealised call to $\lit{FindIndex}$.
%		Then, it is impossible for the adversary to determine the value of $i$ with probability $Pr > \frac{1}{n - t}$.}
	\item {\bf AARB-Anonymity:}
		\sloppy{
		Suppose that non-faulty process $p_i$ invokes $\lit{AARBP}[\ms{ID}](m)$ for some $m$ and given $\ms{ID}$, and has previously invoked an arbitrary number of $\lit{AARBP}[\ms{ID_j}](m_j)$ calls where $\ms{ID} \neq \ms{ID_j}$ for all such $j$.
		Suppose that an adversary $A$ is required to output a guess $g \in \{1, \dots, n\}$, corresponding to the identity of $p_i$ after performing an arbitrary number of computations, allowed oracle calls and invocations of networking primitives.
		Then for any $A$, $Pr(i = g) \leq \frac{1}{n - t}$.
}
\end{enumerate}
Informally, AARB-Anonymity guarantees that the source of an anonymously broadcast message by a non-faulty process is unknown to the adversary, in that it is indistinguishable from $n - t$ (non-faulty) processes.
AARBP can be implemented by composing $n$ instances of Bracha's reliable broadcast algorithm\ifconference
, which we describe and prove correct in the full paper~\cite{avcpaper}.
\else
, which we describe and prove correct in Appendix~\ref{sec:arbproof}.
\fi

% TODO: include Sign() in ARB-broadcast algorithm, define state explicitly (i.e. pass in L = (issue, m)), assume issue is agreed upon a priori
%\subsection{Protocol description}
%
%We present Anonymous Reliable Broadcast Protocol (AARBP), an adaption of Bracha's RB-broadcast protocol \cite{Bracha1987} that ensures the properties of the ARB-Broadcast problem.

\section{Anonymity-preserving vector consensus} \label{sec:avc:avc}
In this section, we introduce the \textit{anonymity-preserving vector consensus} problem and present and discuss the protocol Anonymised Vector Consensus Protocol (AVCP) that solves it.
\ifconference
\else We defer its proof to Appendix~\ref{sec:avcproof}.
\fi
%This section introduces the problem of \textit{anonymised vector consensus}, an extension of vector consensus. We then describe our protocol, Anonymised Vector Consensus Protocol (AVCP), that solves anonymised vector consensus, deferring a correctness proof to Appendix \ref{sec:avcproof}. As we did in the previous section, we analyse the protocol's complexity, and provide some remarks and optimisations.
%
%\subsection{Preliminaries} \label{sec:avc:avcprelim}
%Reliable broadcast can be useful when consensus is not needed, like when broadcasting information to a designated group. 
%In some contexts, however, anonymised all-to-all reliable broadcast is not sufficient, as processes are only guaranteed that they will \textit{eventually} deliver messages sent by non-faulty processes (ARB-Termination-1), or messages ARB-delivered by non-faulty processes (ARB-Termination-2). That is, there is no mechanism to decide when to stop accepting messages from processes, and if processes only wait for $n - t$ messages in all-to-all ARB-broadcast, they are not guaranteed to reach agreement.
The anonymity-preserving vector consensus problem brings anonymity to the vector consensus problem~\cite{Doudou1997} where non-faulty processes reach an agreement upon a vector containing at least $n-t$ proposed values.
%
%Vector consensus~\cite{DS09} is an extension of consensus and a variation on interactive consistency for when synchrony is not assumed. It provides the means for processes to reach agreement on a set of valid values proposed by processes from some predetermined group, guaranteeing eventual termination.
More precisely, anonymised vector consensus asserts that the identity of a process who proposes must be indistinguishable from that of all non-faulty processes.
%In electronic voting, ensuring the properties of consensus and anonymity provides a natural mechanism to prevent double-voting, enforce election eligibility requirements and mitigate .
As in AARBP, instances of AVCP are identified uniquely by a given identifier $\ms{ID}$.
Each process is equipped with two operations. 
Firstly, ``$\lit{AVCP}[\ms{ID}](m)$'' begins execution of an instance of AVCP with respect to $\ms{ID}$ and proposal $m$.
Secondly, ``$\lit{AVC-decide}[\ms{ID}](V)$'' signals the output of $V$ from the instance of AVCP identified by $\ms{ID}$, and is invoked exactly once per identifier.
We define a protocol that solves anonymity-preserving vector consensus with respect to these operations as satisfying the following four properties:
%
%As such, we define \textit{anonymised vector consensus}, which ensures not only the usual properties of safety and termination as in vector consensus, but also the \textit{anonymity} of non-faulty processes.
%In the case of distributed electronic voting, ensuring that both consensus is reached \textit{and} anonymity is preserved can dissociate a voter from their ballot, and provides a natural mechanism to prevent double-voting and enforce election eligibility requirements.
%
%We define a protocol that implements anonymised vector consensus as satisfying the following properties:
%\begin{enumerate}[noitemsep,topsep=5pt,parsep=2pt,partopsep=5pt,itemindent=1em,listparindent=1em,leftmargin=1em]
\begin{enumerate}
	\item {\bf AVC-Anonymity:}
		Suppose that non-faulty process $p_i$ invokes $\lit{AVCP}[\ms{ID}](m)$ for some $m$ and given $\ms{ID}$, and has previously invoked an arbitrary number of $\lit{AVCP}[\ms{ID_j}](m_j)$ calls where $\ms{ID} \neq \ms{ID_j}$ for all such $j$.
		Suppose that an adversary $A$ is required to output a guess $g \in \{1, \dots, n\}$, corresponding to the identity of $p_i$ after performing an arbitrary number of computations, allowed oracle calls and invocations of networking primitives.
		Then for any $A$, $Pr(i = g) \leq \frac{1}{n - t}$.
%		Suppose that $m$ is proposed by a non-faulty process $p_i$.
%		Then, it is impossible for the adversary to determine the value of $i$ with probability $Pr > \frac{1}{n - t}$.
\end{enumerate}
It also requires the original agreement and termination properties of vector consensus to be ensured:
\sloppy{
%\begin{enumerate}[noitemsep,topsep=5pt,parsep=2pt,partopsep=5pt,itemindent=1em,listparindent=1em,leftmargin=1em]
\begin{enumerate}
	\setcounter{enumi}{1}
\item {\bf AVC-Agreement:} All non-faulty processes that invoke $\lit{AVC-decide}[\ms{ID}](V)$ do so with respect to the same vector $V$ for a given $\ms{ID}$.
\item {\bf AVC-Termination:} Every non-faulty process eventually invokes $\lit{AVC-decide}[\ms{ID}](V)$ for some vector $V$ and a given $\ms{ID}$.
\end{enumerate}
}
It also requires a validity property that depends on a pre-determined, deterministic validity predicate $\lit{valid}()$~\cite{CKPS01, CGLR18} which we assume is common to all processes. We assume that all non-faulty processes propose a value that satisfies $\lit{valid}()$.
%As done in DBFT \cite{CGLR18}, we define AVC-Validity with respect to a validity predicate $\lit{valid}()$, 
%which behaves as follows:
%that is identical across processes and deterministic (e.g., in tracking a blockchain $\lit{valid}()$ must assert that a proposal contains the hash of the previously-decided block):
% TODO - contrast with previous work (deciding 2t + 1?)
%%\begin{itemize}
%%	\item 
%	All processes access the same function $\lit{valid}()$.
%	%\item 
%	$\lit{valid}()$ must be deterministic.
%	%\item 
%	valid() can be adapted as per context - in a blockchain, for instance, $\lit{valid}()$ might assert that a proposal must contain the hash of the previously-decided block.
%	%\item 
%	Proposals that fulfil valid() that are produced by Byzantine processes may still be decided.
%%\end{itemize}
%\begin{enumerate}[noitemsep,topsep=5pt,parsep=2pt,partopsep=5pt,itemindent=1em,listparindent=1em,leftmargin=1em]
\begin{enumerate}
\setcounter{enumi}{3}
\item {\bf AVC-Validity:} Consider each non-faulty process that invokes $\lit{AVC-decide}[\ms{ID}](V)$ for some $V$ and a given $\ms{ID}$.
	Each value $v \in V$ must satisfy $\lit{valid}()$, and $|V| \geq n - t$.
	Further, at least $|V| - t$ values correspond to the proposals of distinct non-faulty processes.
\end{enumerate}

%Our definition of AVC-Anonymity is the same as in ARB-broadcast, ensuring non-faulty processes cannot be de-anonymised by up to $f < n - 1$ colluding Byzantine processes.
%As we shall see, we rely on partial synchrony implicitly in the terminating variant of the binary consensus routine invoked in our protocol, and so we do not explicitly reason with respect to partial synchrony.

%\vspace{0.5em}\noindent{\bf AVCP.}
\subsection{AVCP}%\paragraph{AVCP.}
We present a reduction to binary consensus which may converge in four message steps, as in the reduction to binary consensus in DBFT~\cite{CGLR18}, at least one of which must be performed over anonymous channels.
\ifconference
	We present the proof of correctness in the full paper~\cite{avcpaper}.
\else
\fi
We note that comparable problems~\cite{Diamantopoulos2015}, including agreement on a core set over an asynchronous network~\cite{Ben-Or1994}, rely on such a reduction.
The protocol is divided into two components.
Firstly, the reduction component (Algorithm~\ref{alg:avc}) reduces anonymity-preserving vector consensus to binary consensus.
Here, one instance of AARBP and $n$ instances of binary consensus are executed.
But, since proposals are made anonymously, processes cannot associate proposals with binary consensus instances a priori.
Consequently, processes start with $n$ unlabelled binary consensus instances, and label them over time with the hash digest of proposals they deliver (of the form $h \in \{0, 1\}^{*}$).
To cope with messages sent and received in unlabelled binary consensus instances, we require a handler component (Algorithm~\ref{alg:bchandler}) that replaces function calls made in binary consensus instances.
%
%Firstly, the \textit{reduction} component comprises the reduction from anonymised vector consensus to binary consensus. Then, the \textit{binary consensus} component involves a slight augmentation of an existing binary consensus algorithm~\cite{CGLR18}. In particular, the binary consensus routine invokes BV-broadcast, an all-to-all communication abstraction defined initially in~\cite{Mostefaoui2015} for filtering binary values. Finally, the \textit{handler} component contains the function signatures corresponding to the augmentations of the binary consensus routine, and handles the sending and receiving of all messages in the protocol.

%\vspace{0.5em}\noindent{\bf Functions.}
\paragraph{Functions.}
In addition to the communication primitives detailed in Section~\ref{sec:model:model} and the two primitives ``$\lit{AVCP}$'' and ``$\lit{AVC-decide}$'', the following primitives may be called:
	``$inst.\lit{bin\_propose}(v)$'', where $inst$ is an instance of binary consensus and $v \in \{0, 1\}$, begins execution of $inst$ with initial value $v$,
	``$\lit{AARBP}$'' and ``$\lit{AARB-deliver}$'', as in Section~\ref{sec:arb:arb},
	``$\lit{valid}()$'' as described above,
%	``$\lit{sort}(l)$'', which takes a set of labels $l$ and returns a corresponding sorted list,
	``$m.\lit{keys}()$'' (resp. ``$m.\lit{values}()$''), which returns the keys (resp. values) of a map $m$,
	``$item.\lit{key}()$'', which returns the key of $item$ in a map $m$,
	``$s.\lit{pop}()$'', which removes and returns a value from set $s$, and
	``$H(v)$'', a collision-resistant hash function which returns $h \in \{0, 1\}^{*}$ based on $v \in \{0, 1\}^{*}$.

%\vspace{0.5em}\noindent{\bf State.}
\paragraph{State.}
Each process tracks the following variables:
% TODO: maybe State, Messages, Functions could be reduced in size by not using bullet points
%\begin{itemize}
%[noitemsep,topsep=5pt,parsep=2pt,partopsep=5pt,itemindent=1em,listparindent=1em,leftmargin=1em]
	$\ms{ID} \in \{0, 1\}^{*}$, a common identifier for a given instance of AVCP.
	$\ms{proposals}[]$, which maps labels of the form $l \in \{0, 1\}^{*}$ to AARB-delivered messages of the form $(m, \sigma) \in (\{0, 1\}^{*}, \{0, 1\}^{*})$ that may be decided, and is initially empty.
	$\ms{decision\_count}$, tracking the number of binary consensus instances for which a decision has been reached, initialised to 0.
	$\ms{decided\_ones}$, the set of proposals for which 1 was decided in the corresponding binary consensus instance, initialised to $\emptyset$.
	$\ms{labelled}[]$, which maps labels, which are the hash digest $h \in \{0, 1\}^{*}$ of AARB-delivered proposals, to binary consensus instances, and is initially empty. %TODO: clean up
	$\ms{unlabelled}$, a set of binary consensus instances (initially of cardinality $n$) with no current label.
	$\ms{ones}[][]$, which maps two keys, \EST \ and \AUX, to maps with integer keys $r \geq 1$ which map to a set of labels, all of which are initially empty.
	$\ms{counts}[][]$, which maps two keys, \EST \ and \AUX, to maps with integer keys $r \geq 1$ which map to an integer $n \in \{0, \dots, n\}$, all of which are initialised to 0.
%\end{itemize}

%Each process tracks $pk_N$, $sk_i$, $\ms{issue}$ and $L$ as in AARBP in Section \ref{sec:arb:arb}.
%Each process stores $pk_N$, the set of $n$ public keys identifying each process in $P$.
%In a consortium blockchain, this information can be encoded in the genesis block. % TODO: include this?
%The array $proposals[1..n]$ stores valid proposals which are locally indexed by an iterating counter \ms{curr\_index}.
%Proposals are associated with binary consensus instances, and so $bin\_decisions[1..n]$ tracks the corresponding decision for all processes.
%These instances are tracked locally in \ms{BIN\_CONS}[$1..n, l_1..l_n$].
%In our protocol, processes do not know a priori which process corresponds to which proposal, and so cannot agree on indices for binary consensus instances a priori.
%As such, each process labels binary consensus instances with the hash of a message that is ARB-delivered (values $l_i$, where $i \in \{1..n\}$).
%Further, each process has four key-value stores, which we define in describing the alteration of the binary consensus algorithm.
%\subsubsection{Messages}

%\vspace{0.5em}\noindent{\bf Messages.}
\paragraph{Messages.}
In addition to messages propagated in AARBP, non-faulty processes process messages of the form
$(\ms{ID}, $\TAG$, r, \ms{label}, b)$, where \TAG \ $\in \{$\EST $,$ \AUX$\}$, $r \geq 1$, $\ms{label} \in \{0, 1\}^{*}$ and $b \in \{0, 1\}$. A process buffers a message $(\ms{ID}, $\TAG$, r, label, b)$ until $\ms{label}$ labels an instance of binary consensus $inst$, at which point the message is considered receipt in $inst$.
The handler, described below, ensures that all messages sent by non-faulty processes eventually correspond to a label in their set of labelled consensus instances (i.e. contained in $\ms{labelled}.\lit{keys}()$).
Similarly, a non-faulty process can only broadcast such a message after labelling the corresponding instance of binary consensus.
Processes also process messages of the form $(\ms{ID}, $\TAG$, r, ones)$, where \TAG \ $\in \{$\ESTONES $,$ \AUXONES$\}$, $r \geq 1$, and $ones$ is a set of strings corresponding to binary consensus instance labels.
% talk about buffering and broadcast when label exists, and say for non-faulty process ARB-Termination-2 implies they will process it
%	\item \EST$[r, label](b)$, where $b \in \{0, 1\}$, denotes messages propagated via BV-broadcast \cite{Mostefaoui2015} in the binary consensus. These messages correspond to processes attempting to disseminate binary values broadcast by non-faulty processes through BV-broadcast.
%	\item \AUX$[r, label](b)$, where $b = 1$, corresponds to potential decision values in the binary consensus instance \ms{BIN\_CONS}$[i, label]$ for some $i$. Messages of the form \AUX$[r, label](b)$, where $b = 0$, are handled differently.
%	\item $($\MSG$, 0, r, \neg S = \{s_1, \cdots \})$, where \MSG \ $\in \{$\EST$, $\AUX$\}$. Simply put, these messages correspond to the sending of \MSG$[r, label](0)$ in all instances of binary consensus that a process wishes not to send \MSG$[r, label](1)$ in.

%\subsubsection{Reduction}
\begin{algorithm}[t]
	% state in pseudo-code
\caption{AVCP (1 of 2): Reduction to binary consensus}
	{\footnotesize
	\label{alg:avc}
	\begin{algorithmic}[1]
		\Part{$\lit{AVCP}[\ms{ID}](m')$} {
%		\State Spawn $n$ instances of binary consensus, each labelled as $\ms{BIN\_CONS}[\perp]$
		\State $\lit{AARBP}[\ms{ID}](m')$ \label{line:avc:arbbcast} \Comment{anonymised reliable broadcast of proposal}
		\WUntil{$|\ms{decided\_ones}| \geq n - t$} \Comment{wait until $n - t$ instances terminate with 1} \EndWUntil \label{line:avc:wuntilnminust} 
%		\Repeat %		\If{$\exists$ $(m, \sigma)\in \ms{proposals} : \ms{labelled}[H(m \mid \mid \sigma)].\lit{bin\_propose}()$ not invoked} \label{line:avc:condition}
%		\State Invoke $\ms{labelled}[H(m \mid \mid \sigma)].\lit{bin\_propose}(1)$ \label{line:avc:binpropone}
%		\EndIf
%		\EndRepeat
%		\Until{$|\ms{decided\_ones}| \geq n - t$} \label{line:avc:condition} \Comment{waiting for $n - t$ 1's to be decided} %,  which is optimal under asynchrony}
%		\EndUntil
%		\State Deposit buffered messages into different consensus instances
		\For{\textbf{each} $inst \in \ms{unlabelled} \cup \ms{labelled}.\lit{values}()$ such that \\ \ \ \ $inst.\lit{bin\_propose}()$ not yet invoked}
		\State Invoke $inst.\lit{bin\_propose}(0)$ \label{line:avc:binpropzeros} \Comment{propose 0 in all binary consensus not yet invoked}
		\EndFor
		\WUntil {$\ms{decision\_count} = n$} \label{line:avc:nbindecs} \Comment{wait until all $n$ instances of binary consensus terminate} \EndWUntil
%		\State $I$ = $\{k \mid \ms{bin\_decisions[k]} = 1\}$ \label{line:avc:definei} \Comment{taking the bitmask}
%		\WUntil{$(\nexists k \in I \mid \ms{proposals[k] = \perp})$}  \label{line:avc:waitforreturn} \Comment{waiting for the ARB-delivery of values in $I$} \EndWUntil
		\State $\lit{AVC-decide}[ID](\ms{decided\_ones})$ \label{line:avc:returnprops}
	}\EndPart
%	\Upon{the receipt of 
%	\Upon{the receipt of $(b, r, i)$}
%	\State{Wait until $H^{-1}(i)$ is ARB-delivered}
%	\If{$(i \notin J) \land (b = 1)$}
%	\State $J \leftarrow J \cup i$
%	\EndIf
%	\State Label a $\ms{BIN\_CONS}$ instance $i$
%	\State{Input $(b, r, i)$ into $\ms{BIN\_CONS[i]}$}
%	\EndUpon
	\Statex
	\Upon{invocation of $\lit{AARB-deliver}[\ms{ID}](m, \sigma)$}
	\State $\ms{labelled}[H(m \mid \mid \sigma)] \leftarrow \ms{unlabelled}.\lit{pop}()$ \label{line:avc:labelinst}
	\If{$\lit{valid}(m, \sigma)$} \Comment{deterministic, common validity function} \label{line:avc:ifvalid}
	\State $proposals[H(m \mid \mid \sigma)] \leftarrow (m, \sigma)$ \label{line:avc:addtoprops}
	\State Invoke $\ms{labelled}[H(m \mid \mid \sigma)].\lit{bin\_propose}(1)$ if not yet invoked \label{line:avc:invokebpone}
	\EndIf
	\EndUpon
	\Statex
	\Upon{$inst$ deciding a value $v \in \{0, 1\}$, where $inst \in \ms{labelled}.\lit{values}() \cup \ms{unlabelled}$} \label{line:avc:bindecision}
	\If{$v = 1$} \Comment{store proposals for which 1 was decided in the corresponding binary consensus}
	\State $\ms{decided\_ones} \leftarrow \ms{decided\_ones} \cup \{ \ms{proposals}[inst.\lit{key}()]\}$ \label{line:avc:addone}
	\EndIf
	\State $\ms{decision\_count} \leftarrow \ms{decision\_count} + 1$ \label{line:avc:incdeccount}
%	\If{$v = 1$ $\land$ $|1$'s in \ms{bin\_decisions}$| < n - t$} $J \leftarrow J \cup H(k)$ \EndIf
%	\If{$inst \in \ms{labelled}.values()$} $\ms{bin\_decisions[i]} \leftarrow v$ \EndIf \label{line:avc:depositbdprop}
	\EndUpon
%	\Upon{ARB-delivery of $msg$}
%	\State $v \leftarrow$ $valid(v)$ $?$ $1 : 0$
%	\State Label an instance $\ms{BIN\_CONS}[H(msg)]$
%	\EndUpon
	\algstore{alg:avc}
	\end{algorithmic}}
\end{algorithm}

%\vspace{0.5em}\noindent{\bf Reduction.}
\paragraph{Reduction.}
%relying on $n$ instances of binary consensus (that can execute concurrently).
%Structurally, the reduction is very similar to the reduction from multivalued consensus in DBFT \cite{CGLR18}, which in turn was initially proposed by Ben-Or \cite{Ben-Or1994}, relying on $n$ instances of binary consensus (that can execute concurrently). 
%Each process is equipped with the operation $\lit{bin\_propose}(v)$, where $v \in \{0, 1\}$, corresponding to invoking $\ms{BIN\_CONS}$ with input $v$. Each process is also equipped with operations, $\lit{ARB-broadcast}$ and $\lit{ARB-deliver}$, defined in Section~\ref{sec:arb:arb}
%All processes have access to the same deterministic, local operation 
%and $\lit{valid}()$ for disseminating the validity of ARB-delivered proposals.
%At the protocol's outset,
In the reduction, $n$ (initially unlabelled) instances of binary consensus are used, each corresponding to a value that one process in $P$ may propose.
Each (non-faulty) process invokes AARBP with respect to $\ms{ID}$ and their value $m'$ (line~\ref{line:avc:arbbcast}), anonymously broadcasting $(m', \sigma')$ therein.
On AARB-delivery of some message $(m, \sigma)$, an unlabelled instance of binary consensus is deposited into $\ms{labelled}$, whose key (label) is set to $H(m \mid \mid \sigma)$ (line \ref{line:avc:labelinst}).
Proposals that fulfil $\lit{valid}()$ are stored in $\ms{proposals}$ (line~\ref{line:avc:addtoprops}), and $inst.\lit{bin\_propose}(1)$ is invoked with respect to the newly labelled instance $inst = \ms{labelled}[H(m \mid \mid \sigma)]$ if not yet done (line \ref{line:avc:invokebpone}).
Upon termination of each instance (line \ref{line:avc:bindecision}), provided 1 was decided, the corresponding proposal is added to $\ms{decided\_ones}$ (line~\ref{line:avc:addone}).
For either decision value, $\ms{decision\_count}$ is incremented (line~\ref{line:avc:incdeccount}).
Once 1 has been decided in $n - t$ instances of binary consensus, processes will propose 0 in all instances that they have not yet proposed in (line~\ref{line:avc:binpropzeros}).
Note that upon AARB-delivery of valid messages after this point, $\lit{bin\_propose}(1)$ is not invoked at line \ref{line:avc:invokebpone}.
Upon the termination of all $n$ instances of binary consensus (after line~\ref{line:avc:nbindecs}), all non-faulty processes decide their set of values for which 1 was decided in the corresponding instance of binary consensus (line~\ref{line:avc:returnprops}).

\begin{algorithm}[t]
	% state in pseudo-code
	{\footnotesize
%	\begin{multicols}{2}
	\caption{AVCP (2 of 2): Handler of Algorithm~\ref{alg:obc}}
	\label{alg:bchandler}
	\begin{algorithmic}
	\algrestore{alg:avc}
%	\Upon{invoking ``$\lit{broadcast}$ $(\ms{ID}, $\EST$, r, \ms{label}, b)$'' (line~\ref{line:obc:bcastinbv} in Algorithm~\ref{alg:obc}) in \Statex $\ms{inst} \in \ms{labelled}.\lit{values}() \cup \ms{unlabelled}$} \label{line:avc:uponbcastest}
	\Upon{``$\lit{broadcast}$ $(\ms{ID}, $\EST$, r, \ms{label}, b)$'' in $\ms{inst} \in \ms{labelled}.\lit{values}() \cup \ms{unlabelled}$} \label{line:avc:uponbcastest}
%	\If{$inst \in \ms{labelled}$} \label{line:avc:instlabelledest}
%	\EndIf
	\If{b = 1} \label{line:avc:ifboneest}
	\State $\lit{broadcast}$ $(\ms{ID}, $\EST$, r, \ms{label}, b)$ \label{line:avc:bcastest}
	\State $\ms{ones}[$\EST$][r] \leftarrow \ms{ones}[$\EST$][r] \cup \{\ms{inst}.\lit{key}()\}$ \label{line:avc:appendonesest}
	\EndIf
	\State $\ms{counts}[$\EST$][r] \leftarrow \ms{counts}[$\EST$][r] + 1$ \label{line:avc:inccountest}
	\If{$\ms{counts}[$\EST$][r] = n \land |ones[$\EST $][r]| < n$} \label{line:avc:ifnest}
	\State $\lit{broadcast}$ $(\ms{ID}, $\ESTONES$, r, ones[$\EST $][r])$ \label{line:avc:bcastestones}
	\EndIf
	\EndUpon
%	\Statex TODO: include this Statex?
%	\Upon{invoking ``$\lit{broadcast}$ $(\ms{ID}, $\AUX$, r, \ms{label}, b)$'' (line~\ref{line:obc:bcastaux} in Algorithm~\ref{alg:obc}) in \Statex $\ms{inst} \in \ms{labelled}.\lit{values}() \cup \ms{unlabelled}$} \label{line:avc:uponbcastaux}
	\Upon{``$\lit{broadcast}$ $(\ms{ID}, $\AUX$, r, \ms{label}, b)$'' in $\ms{inst} \in \ms{labelled}.\lit{values}() \cup \ms{unlabelled}$} \label{line:avc:uponbcastaux}
%	\If{$inst \in \ms{labelled}$}
%	\EndIf
	\If{b = 1}
	\State $\lit{broadcast}$ $(\ms{ID}, $\AUX$, r, \ms{label}, b)$ \label{line:avc:bcastaux}
	\State $\ms{ones}[$\AUX$][r] \leftarrow \ms{ones}[$\AUX$][r] \cup \{\ms{inst}.\lit{key}()\}$

	\EndIf
	\State $\ms{counts}[$\AUX$][r] \leftarrow \ms{counts}[$\AUX$][r] + 1$
	\If{$\ms{counts}[$\AUX$][r] = n \land |ones[$\AUX $][r]| < n$}
	\State $\lit{broadcast}$ $(\ms{ID}, $\AUXONES$, r, ones[$\AUX $][r])$
	\EndIf
	\EndUpon
	\Upon{receipt of $(\ms{ID}, $\TAG$, r, ones)$ s.t. \TAG \ $\in \{$\ESTONES $,$ \AUXONES $\}$}
	\WUntil $one \in \ms{labelled}.\lit{keys}()$ $\forall one \in ones$  \EndWUntil \label{line:avc:wuntillabel}
	\If{\TAG \ $=$ \ESTONES}
	\State \TEMP \ $\leftarrow$ \EST
	\Else \ \TEMP \ $\leftarrow$ \AUX
	\EndIf
	\For{\textbf{each} $l \in \ms{labelled}.\lit{keys}()$ such that $l \notin ones$}
	\State deliver $(\ms{ID}, $\TEMP$, r, l, 0)$ in $\ms{labelled}[l]$ \label{line:avc:delivertemp1}
	\EndFor
	\For{\textbf{each} $inst \in \ms{unlabelled}$} \label{line:avc:forlooptemp2}
	\State deliver $(\ms{ID}, $\TEMP$, r, \perp, 0)$ in $inst$ \label{line:avc:delivertemp2}
	\EndFor
	\EndUpon
%	\Upon{receipt of $(${\scriptsize $\lit{MSG}$}$, 0, r, \neg (S = \{l_1, \dots, l_k\}))$, where {\scriptsize $\lit{MSG}$} $\in$ $\{${\scriptsize $\lit{EST}$}, {\scriptsize $\lit{AUX}$}$\}$} \label{line:bchandler:receiptzeros}
%	\If{$|S| \geq n - t$} needed?
%	\WUntil for all $l \in S$, for some $i \in 1..n$, $\exists$ $\ms{BIN\_CONS}[i, label]$ s.t. $label = l$ \label{line:bchandler:waitzeroarb} \EndWUntil
%	\For{each $i \in 1..n$ s.t. $\ms{BIN\_CONS}[i, label]$ is s.t. $label \notin (S \cup \{\perp\})$} \Comment{each $label \in S^c$} \label{line:bchandler:zerodep}
%	\State Consider {\scriptsize $\lit{MSG}$}$[r, label](0)$ as receipt in $\ms{BIN\_CONS[i, label]}$
%	\EndFor
%	\State Atomically consider $MSG[r,label](0)$ as receipt in all $\ms{BIN\_CONS}[i, label]$ where $label = \perp$ \label{line:bchandler:zeroperpdep}
%	\EndUpon
	\end{algorithmic}
%	\end{multicols}
}
\end{algorithm}

%\vspace{0.5em}\noindent{\bf Handler.}
\paragraph{Handler.}
%\subsubsection{Handler}
As proposals are anonymously broadcast, binary consensus instances cannot be associated with process identifiers a priori, and so are labelled by AARB-delivered messages.
Thus, we require the handler, which overrides two of the three $\lit{broadcast}$ calls in the non-terminating variant of the binary consensus of \cite{CGLR18} (Algorithm~\ref{alg:obc}).
%As the third $\lit{broadcast}$ is only made in labelled binary consensus instances, it does not require handling.
\ifconference
\else We defer the reader to Appendix \ref{sec:bvbincons} for a description of the non-terminating algorithm and the terminating variant that requires handling.
\fi

%The handler serves to transform the sending and receiving of values in the binary consensus algorithm -- Algorithm~\ref{alg:obc} -- such that processes may progress in instances without having ARB-delivered a particular process' message.
%By assumption of $t < \frac{n}{3}$, where for non-trivial $n$ we may have $t > 0$, a process may be faulty and thus never ARB-broadcast their proposal.
%Similarly, a proposal may not be ARB-broadcast or ARB-delivered in a timely fashion.
%TODO: clean up later

We now describe the handler (Algorithm~\ref{alg:bchandler}).
Let $inst$ be an instance of binary consensus.
On calling $inst.\lit{bin\_propose}(b)$ ($b \in \{0, 1\}$) (and at the beginning of each round $r \geq 1$), processes invoke BV-broadcast (line~\ref{line:obc:bvbcast} of Algorithm~\ref{alg:obc}), immediately calling ``broadcast $(\ms{ID}, $\EST$, r, \ms{label}, b)$'' (line~\ref{line:obc:bcastinbv} of Algorithm~\ref{alg:obc}).
%In particular, messages where $\ms{label} = \perp$ are not sent by non-faulty processes.
%We note that messages of the form $(\ms{ID}, $\EST$, r, \ms{label}, b)$, where \TAG \ $\in \{$\EST$,$ \AUX$\}$ are buffered by processes until $\ms{label}$ labels an instance of binary consensus (after ARB-delivery of the corresponding message).
If $b = 1$, $(\ms{ID}, $\EST$, r, \ms{label}, 1)$ is broadcast, and $\ms{label}$ is added to the set $\ms{ones}[$\EST$][r]$ (line \ref{line:avc:appendonesest}).
Note that, given AARB-Termination-2, all messages sent by non-faulty processes of the form $(\ms{ID}, $\EST$, r, \ms{label}, 1)$ will be deposited in an instance $\ms{inst}$ labelled by $\ms{label}$.
Then, as the binary consensus routine terminates when all non-faulty processes propose the same value, all processes will decide the value 1 in $n - t$ instances of binary consensus (i.e. will pass line~\ref{line:avc:wuntilnminust}), after which they execute $\lit{bin\_propose}(0)$ in the remaining instances of binary consensus.

%$If $b = 1$ (line \ref{line:avc:ifboneest}), non-faulty processes broadcast $(\ms{ID}, $\EST$, r, \ms{label}, 1)$ and add \ms{label} %Note that when $b = 1$, $\ms{label} \neq \perp$, as non-faulty processes only invoke $\lit{bin\_propose}(1)$ and can propose values $b = 1$ once $inst \in \ms{labelled}.\lit{values}()$.
%For each $b \in \{0, 1\}$, $\ms{counts}[$\EST$][r]$ -- corresponding to the number of ``$\lit{broadcast}$ $(\ms{ID}, $\EST$, r, \ms{label}, b)$'' calls at line~\ref{line:obc:bcastinbv} of Algorithm~\ref{alg:obc} made in round $r$ over all instances -- is incremented (line \ref{line:avc:inccountest}).
Since these instances may not be labelled when a process wishes to broadcast a value of the form $(\ms{ID}, $\EST$, r, label, 0)$, we defer their broadcast until ``$\lit{broadcast}$ $(\ms{ID}, $\EST$, r, label, b)$'' is called in all $n$ instances of binary consensus.
%At this point, the set of instances that they wish to propose zero in is defined by the complement of all instances that they proposed one in.
At this point (line \ref{line:avc:ifnest}), $(\ms{ID}, $\ESTONES$, r, \ms{ones}[EST][r])$ is broadcast (line \ref{line:avc:bcastestones}).
A message of the form $(\ms{ID}, $\ESTONES$, r, \ms{ones})$ is interpreted as the receipt of zeros in all instances not labelled by elements in $\ms{ones}$ (at lines \ref{line:avc:delivertemp1} and \ref{line:avc:delivertemp2}).
This can only be done once all elements of $\ms{ones}$ label instances of binary consensus (i.e., after line~\ref{line:avc:wuntillabel}).
Note that if $|\ms{ones}[$\EST$][r] = n|$, then there are no zeroes to be processed by receiving processes, and so the broadcast at line \ref{line:avc:bcastestones} can be skipped.

Handling ``$\lit{broadcast}$ $(\ms{ID}, $\AUX$, r, \ms{label}, b)$'' calls (line~\ref{line:obc:bcastaux} of Algorithm~\ref{alg:obc}) is identical to the handling of initial ``$\lit{broadcast}$ $(\ms{ID}, $\EST$, r, \ms{label}, b)$'' calls.
Note that the third broadcast in the original algorithm, where $(\ms{ID}, $\EST$, r, \ms{label}, b)$ is broadcast upon receipt from $t + 1$ processes if not yet done before (line~\ref{line:obc:tplusone} of Algorithm~\ref{alg:obc} (BV-Broadcast)), can only occur once the corresponding instance of binary consensus is labelled.
Thus, it does not need to be handled. From here, we can see that messages in the handler are processed as if $n$ instances of the original binary consensus algorithm were executing.
\begin{table}[ht]
\centering
%{\footnotesize
\setlength{\tabcolsep}{18pt}
\caption{Comparing the complexity of AVCP and DBFT~\cite{CGLR18} after GST~\cite{Dwork1988}}
\vspace{5pt}
\begin{tabular}{|l|c|c|}
\hline
\textbf{Complexity} & \textbf{AVCP} & \textbf{DBFT} \\ \hline
Fault-free message complexity & $O(n^3)$    & $O(n^3)$    \\
Worst-case message complexity & $O(tn^3)$    & $O(tn^3)$    \\
Fault-free bit complexity &  $O((S + c)n^3)$ %$O(cn^3 + kn^3)$    
& $O(n^3)$    \\
Worst-case bit complexity & $O((S + c)tn^3)$ %$O(kn^4 + ctn^3)$    
& $O(tn^3)$    \\
\hline
\end{tabular}
%}
\end{table}
%\vspace{0.5em}\noindent{\bf Complexity and optimizations.}
%\paragraph{Complexity and optimizations.}
\subsection{Complexity and optimizations}
Let $k$ be a security parameter, $S$ the size of a signature and $c$ the size of a message digest.
We compare the message and communication complexity of AVCP with DBFT \cite{CGLR18}, which, as written, can be easily altered to solve vector consensus. We assume that AVCP is invoking the terminating variant of the binary consensus of \cite{CGLR18}. When considering complexity, we only count messages in the binary consensus routines once the global stabilisation time (GST) has been reached~\cite{Dwork1988}.
Both fault-free and worst-case message complexity are identical between the two protocols.
We remark that there exist runs of AVCP where processes are faulty with $O(n^3)$ message complexity, such as when a process has crashed.
AVCP mainly incurs greater communication complexity proportional to the size of the signatures, which can vary from size $O(k)$~\cite{Tsang2005a, gu2019efficient} to $O(kn)$~\cite{Fujisaki2011}.
If processes make a single anonymous broadcast per run, the fault-free and worst-case bit complexities of AVCP are lowered to $O(Sn^2 + cn^3)$ and $O(Sn^2 + ctn^3)$.
As is done in DBFT \cite{CGLR18}, we can combine the anonymity-preserving all-to-all reliable broadcast of a message $m$ and the proposal of the binary value 1 in the first round of a binary consensus instance.
To this end, a process may skip the BV-broadcast step in round 1, which may allow AVCP to converge in four message steps, at least one of which must be anonymous.
It may be useful to invoke ``$\lit{broadcast}$ \TAG$[r](b)$'', where \TAG \ $\in \{$\EST$,$\AUX$\}$ (lines \ref{line:avc:bcastest} and \ref{line:avc:bcastaux}) when the instance of binary consensus is labelled, rather than simply when $b = 1$ (i.e., the condition preceding these calls).
%Suppose that a (non-faulty) process $p$ invokes AVC\_BV\_broadcast (or indeed AVC\_AUX\_broadcast) in round $r$ with respect to 0, and that they have labelled the respective binary consensus instance with $label$. As written, zero values are only broadcast when AVC\_BV\_broadcast or AVC\_AUX\_broadcast has been invoked $n$ times in round $r$ (after lines \ref{line:bchandler:estcountn} and \ref{line:bchandler:auxcountn}). But, since $label$ is known to $p$, $p$ has the capacity to broadcast a message of the form \MSG$[r, label](0)$.
Since it may take some time for all $n$ instances of binary consensus to synchronise, doing this may speed up convergence in the ``faster'' instances.
%However, messages of the form \TAG$[r](ones)$, where \TAG \ $\in \{$\ESTONES $,$ \AUXONES$\}$, would then often contain a larger $ones$ set.

\ifconference
\else
\section{Combining regular and anonymous channels} \label{sec:anon}
If processes only use anonymous channels to communicate, it is clear that their anonymity is preserved provided that processes do not double-sign with ring signatures for each message type.
For performance and to prevent long-term correlation attacks on certain anonymous networks~\cite{Mathewson2004}, it may be of interest to use anonymous message passing to propose a value, and then to use regular channels for the rest of a given protocol execution.
% TODO: There are two conditions in which processes can be de-anonymised (i) (ii)
In this setting, the adversary can de-anonymise a non-faulty process by observing message transmission time~\cite{Back2001} and the order in which messages are sent and received~\cite{Raymond2001}.
For example, a single non-faulty process may be relatively slow, and so the adversary may deduce that messages it delivers late over anonymous channels were sent by that process.
\fi
\ifconference
\else  To cope, we make the following assumption about message passing: \\
%To cope, we posit \textit{two} assumptions:% adversarial behaviour that can de-anonymise a non-faulty process in our protocols:
%\begin{enumerate}
%[noitemsep,topsep=5pt,parsep=2pt,partopsep=5pt,itemindent=1em,listparindent=1em,leftmargin=1em]
%\item There exists a sufficiently large window of synchrony at the beginning of any execution of AARBP or AVCP in which all non-faulty processes broadcast and deliver their anonymous proposals.
%\item Non-faulty processes do not send any messages apart from anonymously broadcasting their signed proposal (at line~\ref{line:ab:initbcast} of Algorithm~\ref{alg:arb}) until the window of synchrony has elapsed. This implies that a previous instance of AVCP or AARBP must have terminated prior to the start of the synchrony window.
%\item (old) The order in which a process delivers a message anonymously broadcast is random among non-faulty processes. This mitigates attacks as in ``\textit{(i)}''.

%\item %Suppose that all non-faulty processes each invoke $\lit{anon\_broadcast}$ once in a given protocol's execution.
\noindent \textbf{Independence assumption:}
Consider processes in $P$ who participate in $k \geq 1$ instances of some distributed protocol where messages are exchanged over both regular and anonymous channels.
	   Let $X_i$ be a random variable that maps all invocations of $\lit{send}$ or $\lit{receive}$ by every process to the time that the operation was called in instance $i$.
	   Analogously, let $Y_i$ be a random variable, defined as above, but with respect to invocations of $\lit{anon\_send}$ and $\lit{anon\_receive}$.
	   Then $X_1, \dots, X_k, Y_1, \dots, Y_k$ are mutually independent. \\

	   \noindent This ensures that the adversary cannot correlate the behaviour of a process over regular channels with their behaviour over anonymous channels, and thus the adversary cannot de-anonymise them.
\fi
\ifconference
\else We show that AVCP and our protocol solving anonymity-preserving all-to-all reliable broadcast satisfy their respective definitions of anonymity under our assumption in Appendices \ref{sec:avcproof} and \ref{sec:arbproof}.
\fi
\ifconference
\else Our assumption is quite general, and so achieving it in practice depends on the latency guarantees of the anonymous channels, the speed of each process, and the latency guarantees of the regular channels.
One possible strategy could be to use public networks like Tor~\cite{Dingledine2004} where message transmission time through the network can be measured\footnote{\url{https://metrics.torproject.org/}}.
Then, based on the behaviour of the anonymous channels, processes can vary the timing of their own messages by introducing random message delays~\cite{Mathewson2004} to minimise the correlation between random variables.
It may also be useful for processes to synchronise between protocol executions.
This prevents a process from being de-anonymised when they, for example, invoke $\lit{anon\_send}$ in some instance when all other processes are executing in a different instance.
\fi

\section{Experiments} \label{sec:exp}
Benchmarks of distributed protocols were performed using Amazon EC2 instances.
We refer to each EC2 instance used as a \textit{node}, corresponding to a process in protocol descriptions.
For each value of $n$ (the number of nodes) chosen, we ran experiments with an equal number of nodes from \textit{four} regions: Oregon (us-west-2), Ohio (us-east-2), Singapore (ap-southeast-1) and Frankfurt (eu-central-1).
The type of instance chosen was c4.xlarge, which provide 7.5GiB of memory, and 4 vCPUs, i.e. 4 cores of an Intel Xeon E5-2666 v3 processor.
We performed between 50 and 60 iterations for each value of $n$ and $t$ benchmarked.
We vary $n$ incrementally, and vary $t$ both with respect to the maximum fault-tolerance (i.e. $t = \lfloor \frac{n - 1}{3} \rfloor$), and also fix $t = 6$ for values of $n = 20, 40, \dots$
%As such, we evaluate the effect of varying both $n$ and $t$.
All networking code, and the application logic, was written in Python (2.7).
As we have implemented our cryptosystems in golang, we call our libraries from Python using ctypes\footnote{\url{https://docs.python.org/2/library/ctypes.html}}.
To simulate reliable channels, nodes communicate over TCP.
Nodes begin timing once all connections have been established (i.e. after handshaking).

Our protocol, Anonymised Vector Consensus Protocol (AVCP), was implemented on top of the existing DBFT \cite{CGLR18} codebase, as was the case with \ifconference our implementation of AARBP.
\else AARBP.
\fi
We do not use the fast-path optimisation described in Section~\ref{sec:avc:avc}, but we hash messages during reliable broadcast to reduce bandwidth consumption.
%With respect to optimisations, we use of the aforementioned hashing optimisation in both AARBP and reliable broadcast~\cite{Bracha1987}, or RB-broadcast (which are primitives in AVCP and DBFT respectively), but do not make use of the fast-path optimisation described in Section~\ref{sec:avc:avc}.
We use the most conservative choice of ring signatures, $O(kn)$-sized traceable ring signatures~\cite{Fujisaki2007}, which require $O(n)$ operations for each signing and verification call, and $O(n^2)$ work for tracing overall.
Each process makes use of a single anonymous broadcast in each run of the algorithm.
To simulate the increased latency afforded by using publicly-deployed anonymous networks, processes invoke a local timeout for 750ms before invoking $\lit{anon\_broadcast}$, which is a regular broadcast in our experiments.

\begin{figure}[htbp]
\begin{subfigure}{0.5\textwidth}
    \centering
	\includegraphics[width=1.0\textwidth]{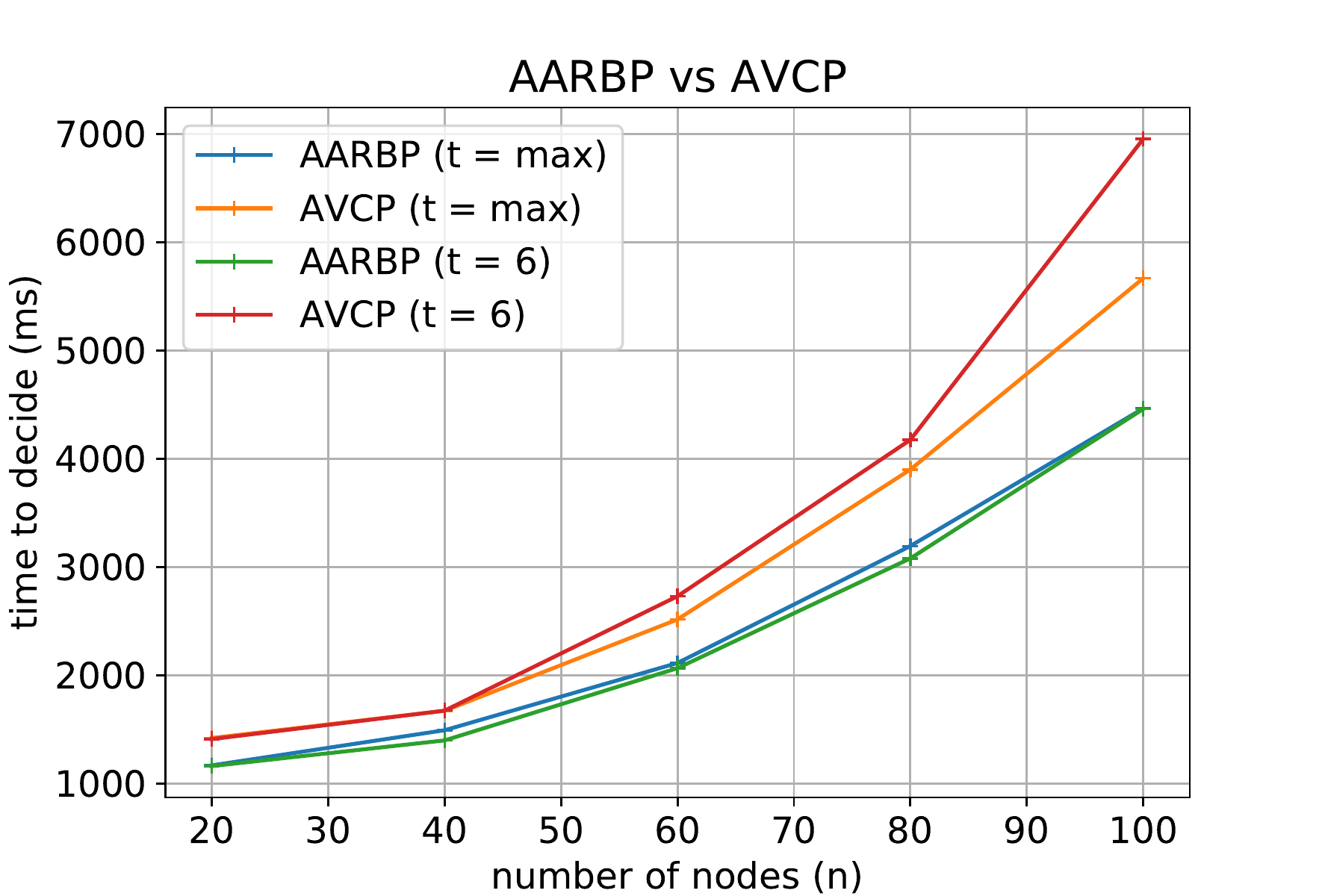}
	\caption{Comparing AARBP and AVCP}
    \label{fig:arb_vs_avc}
\end{subfigure}
\begin{subfigure}{0.5\textwidth}
    \centering
	\includegraphics[width=1.0\textwidth]{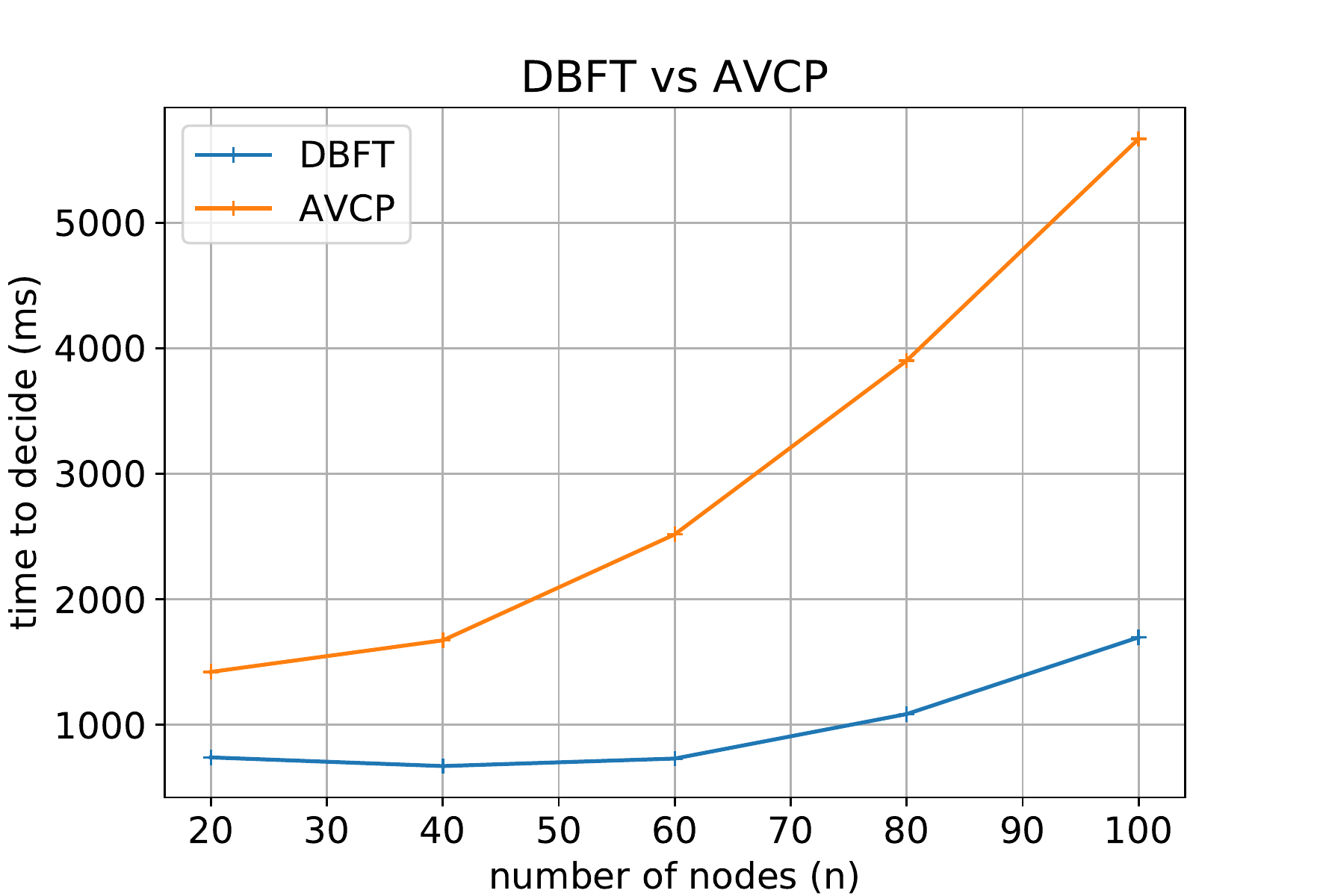}
	\caption{Comparing DBFT and AVCP}
	\label{fig:psync_avc}
\end{subfigure}
\caption{Evaluating the cost of the anonymous broadcast (AARBP) in our solution (AVCP) and the performance of our solution (AVCP) compared to an efficient Byzantine consensus baseline (DBFT) without anonymity preservation}
\end{figure}

Figure \ref{fig:arb_vs_avc} compares the performance of AARBP with that of AVCP. In general, convergence time for AVCP is higher as we need at least three more message steps for a process to decide. Given that the fast-path optimisation is used, requiring 1 additional message step over AARBP in the good case, the difference in performance between AVCP and AARBP would indeed be smaller.

Comparing AVCP with $t = max$ and $t = 6$, we see that when $t = 6$, convergence is slower. Indeed, AVC-Validity states at least $n - t$ values fulfilling valid() are included in a process' vector given that they decide. Consequently, as $t$ is smaller, $n - t$ is larger, and so nodes will process and decide more values. Although AARB-delivery may be faster for some messages, nodes generally have to perform more TRS verification/tracing operations. As nodes decide 1 in more instances of binary consensus, messages of the form $($\MSG$, 0, r, S = \{s_1, \cdots\})$ are propagated where $|S|$ is generally larger, slowing down decision time primarily due to size of the message.
\ifconference
\else We conjecture that nodes having to AARB-deliver all values in $S$ before processing such a message does not slow down performance, as all nodes are non-faulty in our experiments.
\fi

Figure \ref{fig:psync_avc} compares the performance of DBFT as a vector consensus vector routine with AVCP. Indeed, the difference in performance between AVCP and DBFT when $n = 20$ and $n = 40$ is primarily due to AVCP's 750ms timeout. As expected when scaling $n$, cryptographic operations result in worse scaling characteristics for AVCP.
\ifconference
\else As can be seen, DBFT performs relatively well. However, DBFT does not leverage anonymous channels, nor relies on ring signatures, and so AVCP's comparatively slow performance was expected. It is reassuring that AVCP's performance does not differ by an order of magnitude from that of DBFT, given AVCP provides anonymity guarantees.

Overall, AVCP performs reasonably well. Interestingly, AVCP performs better when $t$ is set as the maximum possible value, and so is best used in practice when maximising fault tolerance. Nevertheless, converging when $n = 100$ takes between 5 and 7 seconds, depending on $t$, which is practically reasonable. 
\fi
\ifconference

	Overall, AVCP performs reasonably well, reaching convergence when $n = 100$ between 5 and 7 seconds depending on $t$, which is practically reasonable.
\else \fi

\ifconference
\else
\section{Conclusion} \label{sec:conclusion}
In this paper, we have presented modular and efficient distributed protocols which allow identified processes to propose values anonymously.
Importantly, anonymity-preserving vector consensus ensures that the proposals of non-faulty processes are not tied to their identity, which has applications in electronic voting.
We proposed definitions of anonymity in the context of Byzantine fault-tolerant computing and corresponding assumptions to mitigate the de-anonymisation of identified processes.
In terms of future work, it is of interest to evaluate anonymity in different formal models~\cite{Serjantov2002, Halpern2005} and with respect to various practical attack vectors~\cite{Raymond2001}.
It will be useful also to formalise anonymity under more practical assumptions so that the timing of anonymous and regular message passing do not correlate highly.
In addition, a reduction to a randomized~\cite{Cachin2005} binary consensus algorithm would remove the dependency on the weak coordinator, a form of leader used in each round of the binary consensus algorithm we rely on~\cite{CGLR18}.
\fi

\bibliographystyle{plain}
{
%\setstretch{1.25}
%\cleardoublepage
\phantomsection
\bibliography{db}

\begin{thebibliography}{10}

\bibitem{Ristretto}
The ristretto group.
\newblock \url{https://ristretto.group/}, 2018.
\newblock Accessed: 2018-11-03.

\bibitem{Adida2008}
Ben Adida.
\newblock Helios: Web-based open-audit voting.
\newblock In {\em Proceedings of the 17th Conference on Security Symposium},
  SS'08, pages 335--348, Berkeley, CA, USA, 2008. USENIX Association.

\bibitem{Back2001}
Adam Back, Ulf M{\"o}ller, and Anton Stiglic.
\newblock Traffic analysis attacks and trade-offs in anonymity providing
  systems.
\newblock In {\em International Workshop on Information Hiding}, pages
  245--257. Springer, 2001.

\bibitem{Bellare1993}
Mihir Bellare and Phillip Rogaway.
\newblock Random oracles are practical: A paradigm for designing efficient
  protocols.
\newblock In {\em Proceedings of the 1st ACM conference on Computer and
  communications security}, pages 62--73. ACM, 1993.

\bibitem{Ben-Or1994}
Michael Ben-Or, Boaz Kelmer, and Tal Rabin.
\newblock Asynchronous secure computations with optimal resilience (extended
  abstract).
\newblock In {\em Proceedings of the Thirteenth Annual ACM Symposium on
  Principles of Distributed Computing}, PODC '94, pages 183--192, New York, NY,
  USA, 1994. ACM.

\bibitem{Bernstein2006}
Daniel~J Bernstein.
\newblock Curve25519: new diffie-hellman speed records.
\newblock In {\em International Workshop on Public Key Cryptography}, pages
  207--228. Springer, 2006.

\bibitem{Bracha1987}
Gabriel Bracha.
\newblock Asynchronous byzantine agreement protocols.
\newblock {\em Information and Computation}, 75(2):130 -- 143, 1987.

\bibitem{Bracha1983}
Gabriel Bracha and Sam Toueg.
\newblock Resilient consensus protocols.
\newblock In {\em Proceedings of the second annual ACM symposium on Principles
  of distributed computing}, pages 12--26. ACM, 1983.

\bibitem{Bracha1985}
Gabriel Bracha and Sam Toueg.
\newblock Asynchronous consensus and broadcast protocols.
\newblock {\em Journal of the ACM (JACM)}, 32(4):824--840, 1985.

\bibitem{brakerski2014efficient}
Zvika Brakerski and Vinod Vaikuntanathan.
\newblock Efficient fully homomorphic encryption from (standard) lwe.
\newblock {\em SIAM Journal on Computing}, 43(2):831--871, 2014.

\bibitem{Brier2002}
Eric Brier and Marc Joye.
\newblock Weierstra{\ss} elliptic curves and side-channel attacks.
\newblock In {\em International Workshop on Public Key Cryptography}, pages
  335--345. Springer, 2002.

\bibitem{CKPS01}
Christian Cachin, Klaus Kursawe, Frank Petzold, and Victor Shoup.
\newblock Secure and efficient asynchronous broadcast protocols.
\newblock In {\em Proceeding of the 21st Annual International Cryptology
  Conference on Advances in Cryptology (CRYPTO)}, pages 524--541, 2001.

\bibitem{Cachin2005}
Christian Cachin, Klaus Kursawe, and Victor Shoup.
\newblock Random oracles in {Constantinople}: Practical asynchronous
  {Byzantine} agreement using cryptography.
\newblock {\em Journal of Cryptology}, 18(3):219--246, 2005.

\bibitem{Camp1996}
Jean Camp, Michael Harkavy, J.~D. Tygar, and Bennet Yee.
\newblock Anonymous atomic transactions.
\newblock In {\em In Proceedings of the 2nd USENIX Workshop on Electronic
  Commerce (Nov.), USENIX Assoc}, 1996.

\bibitem{Castro1999}
Miguel Castro and Barbara Liskov.
\newblock Practical byzantine fault tolerance.
\newblock In {\em OSDI}, volume~99, pages 173--186, 1999.

\bibitem{Chaum1983}
David Chaum.
\newblock Blind signatures for untraceable payments.
\newblock In {\em Advances in cryptology}, pages 199--203. Springer, 1983.

\bibitem{Chaum1988}
David Chaum.
\newblock The dining cryptographers problem: Unconditional sender and recipient
  untraceability.
\newblock {\em Journal of cryptology}, 1(1):65--75, 1988.

\bibitem{Chaum1981}
David~L. Chaum.
\newblock Untraceable electronic mail, return addresses, and digital
  pseudonyms.
\newblock {\em Commun. ACM}, 24(2):84--90, February 1981.

\bibitem{Correia2006}
Miguel Correia, Nuno~Ferreira Neves, and Paulo Ver{\'\i}ssimo.
\newblock From consensus to atomic broadcast: Time-free byzantine-resistant
  protocols without signatures.
\newblock {\em The Computer Journal}, 49(1):82--96, 2006.

\bibitem{CGLR18}
Tyler Crain, Vincent Gramoli, Mikel Larrea, and Michel Raynal.
\newblock {DBFT:} efficient leaderless {Byzantine} consensus and its
  application to blockchains.
\newblock In {\em Proceedings of the 17th {IEEE} International Symposium on
  Network Computing and Applications, {NCA} 2018}, pages 1--8, 2018.

\bibitem{Cramer1996}
Ronald Cramer, Matthew Franklin, Berry Schoenmakers, and Moti Yung.
\newblock Multi-authority secret-ballot elections with linear work.
\newblock In {\em International Conference on the Theory and Applications of
  Cryptographic Techniques}, pages 72--83. Springer, 1996.

\bibitem{Danezis2008}
George Danezis and Claudia Diaz.
\newblock A survey of anonymous communication channels.
\newblock Technical report, Technical Report MSR-TR-2008-35, Microsoft
  Research, 2008.

\bibitem{Desmedt1992}
Yvo Desmedt.
\newblock Threshold cryptosystems.
\newblock In {\em International Workshop on the Theory and Application of
  Cryptographic Techniques}, pages 1--14. Springer, 1992.

\bibitem{Diamantopoulos2015}
Panos Diamantopoulos, Stathis Maneas, Christos Patsonakis, Nikos Chondros, and
  Mema Roussopoulos.
\newblock Interactive consistency in practical, mostly-asynchronous systems.
\newblock In {\em Parallel and Distributed Systems (ICPADS), 2015 IEEE 21st
  International Conference on}, pages 752--759. IEEE, 2015.

\bibitem{Dingledine2004}
Roger Dingledine, Nick Mathewson, and Paul Syverson.
\newblock Tor: The second-generation onion router.
\newblock In {\em Proceedings of the 13th Conference on USENIX Security
  Symposium - Volume 13}, SSYM'04, pages 21--21, Berkeley, CA, USA, 2004.
  USENIX Association.

\bibitem{Doudou1997}
Assia Doudou and Andr{\'e} Schiper.
\newblock Muteness failure detectors for consensus with byzantine processes.
\newblock Technical report, in Proceedings of the 17th ACM Symposium on
  Principle of Distributed Computing, (Puerto), 1997.

\bibitem{Dwork1988}
Cynthia Dwork, Nancy Lynch, and Larry Stockmeyer.
\newblock Consensus in the presence of partial synchrony.
\newblock {\em Journal of the ACM (JACM)}, 35(2):288--323, 1988.

\bibitem{elgamal1985public}
Taher ElGamal.
\newblock A public key cryptosystem and a signature scheme based on discrete
  logarithms.
\newblock {\em IEEE transactions on information theory}, 31(4):469--472, 1985.

\bibitem{Fischer1985}
Michael~J. Fischer, Nancy~A. Lynch, and Michael~S. Paterson.
\newblock Impossibility of distributed consensus with one faulty process.
\newblock {\em J. ACM}, 32(2):374--382, April 1985.

\bibitem{freedman2004efficient}
Michael~J Freedman, Kobbi Nissim, and Benny Pinkas.
\newblock Efficient private matching and set intersection.
\newblock In {\em International conference on the theory and applications of
  cryptographic techniques}, pages 1--19. Springer, 2004.

\bibitem{Fujisaki2011}
Eiichiro Fujisaki.
\newblock Sub-linear size traceable ring signatures without random oracles.
\newblock In {\em IEICE Transactions on Fundamentals of Electronics
  Communications and Computer Sciences}, volume E95.A, pages 393--415, 04 2011.

\bibitem{Fujisaki2007}
Eiichiro Fujisaki and Koutarou Suzuki.
\newblock Traceable ring signature.
\newblock In Tatsuaki Okamoto and Xiaoyun Wang, editors, {\em Public Key
  Cryptography -- PKC 2007}, pages 181--200, Berlin, Heidelberg, 2007. Springer
  Berlin Heidelberg.

\bibitem{gentry2009fully}
Craig Gentry.
\newblock Fully homomorphic encryption using ideal lattices.
\newblock In {\em Proceedings of the forty-first annual ACM symposium on Theory
  of computing}, pages 169--178, 2009.

\bibitem{Gilad2012}
Yossi Gilad and Amir Herzberg.
\newblock Spying in the dark: Tcp and tor traffic analysis.
\newblock In Simone Fischer-H{\"u}bner and Matthew Wright, editors, {\em
  Privacy Enhancing Technologies}, pages 100--119, Berlin, Heidelberg, 2012.
  Springer Berlin Heidelberg.

\bibitem{Golle2004}
Philippe Golle, Markus Jakobsson, Ari Juels, and Paul Syverson.
\newblock Universal re-encryption for mixnets.
\newblock In Tatsuaki Okamoto, editor, {\em Topics in Cryptology -- CT-RSA
  2004}, pages 163--178, Berlin, Heidelberg, 2004. Springer Berlin Heidelberg.

\bibitem{Golle2004a}
Philippe Golle and Ari Juels.
\newblock Dining cryptographers revisited.
\newblock In {\em International Conference on the Theory and Applications of
  Cryptographic Techniques}, pages 456--473. Springer, 2004.

\bibitem{Groth2004}
Jens Groth.
\newblock Efficient maximal privacy in boardroom voting and anonymous
  broadcast.
\newblock In Ari Juels, editor, {\em Financial Cryptography}, pages 90--104,
  Berlin, Heidelberg, 2004. Springer Berlin Heidelberg.

\bibitem{gu2019efficient}
Ke~Gu, Xinying Dong, and Linyu Wang.
\newblock Efficient traceable ring signature scheme without pairings.
\newblock {\em Advances in Mathematics of Communications}, page~0, 2019.

\bibitem{Halpern2005}
Joseph~Y Halpern and Kevin~R O'Neill.
\newblock Anonymity and information hiding in multiagent systems.
\newblock {\em Journal of Computer Security}, 13(3):483--514, 2005.

\bibitem{Hamburg2015}
Mike Hamburg.
\newblock Decaf: Eliminating cofactors through point compression.
\newblock In {\em Annual Cryptology Conference}, pages 705--723. Springer,
  2015.

\bibitem{Jakobsson1998}
Markus Jakobsson.
\newblock A practical mix.
\newblock In {\em International Conference on the Theory and Applications of
  Cryptographic Techniques}, pages 448--461. Springer, 1998.

\bibitem{Juels2010}
Ari Juels, Dario Catalano, and Markus Jakobsson.
\newblock {\em Coercion-Resistant Electronic Elections}, pages 37--63.
\newblock Springer Berlin Heidelberg, Berlin, Heidelberg, 2010.

\bibitem{Kulyk2014}
O.~Kulyk, S.~Neumann, M.~Volkamer, C.~Feier, and T.~Koster.
\newblock Electronic voting with fully distributed trust and maximized
  flexibility regarding ballot design.
\newblock In {\em 2014 6th International Conference on Electronic Voting:
  Verifying the Vote (EVOTE)}, pages 1--10, Oct 2014.

\bibitem{Lamport1982}
Leslie Lamport, Robert Shostak, and Marshall Pease.
\newblock The {Byzantine} generals problem.
\newblock {\em ACM Transactions on Programming Languages and Systems (TOPLAS)},
  4(3):382--401, 1982.

\bibitem{Liu2004}
Joseph~K. Liu, Victor~K. Wei, and Duncan~S. Wong.
\newblock Linkable spontaneous anonymous group signature for ad hoc groups.
\newblock In Huaxiong Wang, Josef Pieprzyk, and Vijay Varadharajan, editors,
  {\em Information Security and Privacy}, pages 325--335, Berlin, Heidelberg,
  2004. Springer Berlin Heidelberg.

\bibitem{Mathewson2004}
Nick Mathewson and Roger Dingledine.
\newblock Practical traffic analysis: Extending and resisting statistical
  disclosure.
\newblock In {\em International Workshop on Privacy Enhancing Technologies},
  pages 17--34. Springer, 2004.

\bibitem{Mostefaoui2015}
Achour Most{\'e}faoui, Hamouma Moumen, and Michel Raynal.
\newblock Signature-free asynchronous binary {Byzantine} consensus with $t<
  n/3$, ${O}(n^2)$ messages, and ${O}(1)$ expected time.
\newblock {\em Journal of the ACM (JACM)}, 62(4):31, 2015.

\bibitem{Murdoch2005}
Steven~J. Murdoch and George Danezis.
\newblock Low-cost traffic analysis of {Tor}.
\newblock In {\em Proceedings of the 2005 IEEE Symposium on Security and
  Privacy}, SP '05, pages 183--195, Washington, DC, USA, 2005. IEEE Computer
  Society.

\bibitem{neff2001verifiable}
C~Andrew Neff.
\newblock A verifiable secret shuffle and its application to e-voting.
\newblock In {\em Proceedings of the 8th ACM conference on Computer and
  Communications Security}, pages 116--125, 2001.

\bibitem{Neves2005}
Nuno~Ferreira Neves, Miguel Correia, and Paulo Verissimo.
\newblock Solving vector consensus with a wormhole.
\newblock {\em IEEE Transactions on Parallel and Distributed Systems},
  16(12):1120--1131, 2005.

\bibitem{paillier1999public}
Pascal Paillier.
\newblock Public-key cryptosystems based on composite degree residuosity
  classes.
\newblock In {\em International conference on the theory and applications of
  cryptographic techniques}, pages 223--238. Springer, 1999.

\bibitem{Raymond2001}
Jean-Fran{\c{c}}ois Raymond.
\newblock Traffic analysis: Protocols, attacks, design issues, and open
  problems.
\newblock In {\em Designing Privacy Enhancing Technologies}, pages 10--29.
  Springer, 2001.

\bibitem{Raynal2018}
Michel Raynal.
\newblock {\em Reliable Broadcast in the Presence of Byzantine Processes},
  pages 61--73.
\newblock Springer International Publishing, Cham, 2018.

\bibitem{Rivest2001}
Ronald~L. Rivest, Adi Shamir, and Yael Tauman.
\newblock How to leak a secret.
\newblock In Colin Boyd, editor, {\em Advances in Cryptology --- ASIACRYPT
  2001}, pages 552--565, Berlin, Heidelberg, 2001. Springer Berlin Heidelberg.

\bibitem{Serjantov2002}
Andrei Serjantov and George Danezis.
\newblock Towards an information theoretic metric for anonymity.
\newblock In {\em International Workshop on Privacy Enhancing Technologies},
  pages 41--53. Springer, 2002.

\bibitem{Shamir1979}
Adi Shamir.
\newblock How to share a secret.
\newblock {\em Commun. ACM}, 22(11):612--613, November 1979.

\bibitem{Shoup2002}
Victor Shoup and Rosario Gennaro.
\newblock Securing threshold cryptosystems against chosen ciphertext attack.
\newblock {\em Journal of Cryptology}, 15(2):75--96, Jan 2002.

\bibitem{Tsang2005a}
Patrick~P. Tsang and Victor~K. Wei.
\newblock Short linkable ring signatures for e-voting, e-cash and attestation.
\newblock In Robert~H. Deng, Feng Bao, HweeHwa Pang, and Jianying Zhou,
  editors, {\em Information Security Practice and Experience}, pages 48--60,
  Berlin, Heidelberg, 2005. Springer Berlin Heidelberg.

\bibitem{Tsang2005}
Patrick~P. Tsang, Victor~K. Wei, Tony~K. Chan, Man~Ho Au, Joseph~K. Liu, and
  Duncan~S. Wong.
\newblock Separable linkable threshold ring signatures.
\newblock In Anne Canteaut and Kapaleeswaran Viswanathan, editors, {\em
  Progress in Cryptology - INDOCRYPT 2004}, pages 384--398, Berlin, Heidelberg,
  2005. Springer Berlin Heidelberg.

\bibitem{yao1986generate}
Andrew Chi-Chih Yao.
\newblock How to generate and exchange secrets.
\newblock In {\em 27th Annual Symposium on Foundations of Computer Science
  (sfcs 1986)}, pages 162--167. IEEE, 1986.

\bibitem{ye2008distributed}
Qingsong Ye, Huaxiong Wang, and Josef Pieprzyk.
\newblock Distributed private matching and set operations.
\newblock In {\em International Conference on Information security practice and
  experience}, pages 347--360. Springer, 2008.

\bibitem{Yuen2013}
Tsz~Hon Yuen, Joseph~K Liu, Man~Ho Au, Willy Susilo, and Jianying Zhou.
\newblock Efficient linkable and/or threshold ring signature without random
  oracles.
\newblock {\em The Computer Journal}, 56(4):407--421, 2013.

\bibitem{Zantout2011}
Bassam Zantout and Ramzi Haraty.
\newblock I2p data communication system.
\newblock In {\em Proceedings of ICN}, pages 401--409. Citeseer, 2011.

\end{thebibliography}
}

\ifconference
\else
\appendix

\section{Analysis of AVCP} \label{sec:avcproof}

In this appendix, we prove that the properties of \textit{anonymity-preserving vector consensus} (AVC), presented in Section~\ref{sec:avc:avc}, are satisfied by our protocol Anonymised Vector Consensus Protocol (AVCP) (Algorithms \ref{alg:avc} and \ref{alg:bchandler} in Section~\ref{sec:avc:avc}).
%At a high level, our proof strategy is as follows.
%Firstly, we prove that the labelling mechanism is correct (Lemma \ref{avc-labelling}), which is implicitly relied upon in the rest of our proofs.
%Then, we prove that non-faulty processes send and receive messages in a manner that preserves the properties of $n$ concurrent instances of binary consensus.
%Firstly, we consider messages of the form \MSG$[r, i](1)$  which allows us to prove that all non-faulty processes eventually decide 1 in $n - t$ instances of binary consensus (Lemma \ref{avc-reachline}), whose corresponding proposals fulfil valid() (Lemma \ref{avc-val1}), after proving a few lemmas.
%Then, we consider the propagation of zero-messages (Lemma \ref{avc-equivall}).
%Having effectively considered all components of the reduction, we proceed to prove each property of AVC individually (Theorems \ref{avc-val} to \ref{avc-anon}), completing the proof.

%We note that a process' local vector $V$ is denoted by $\ms{proposals}$ in the protocol. Further, we assume that a necessary condition for a process to be considered correct is in the AARB-broadcast of a value that satisfies valid().

\begin{lemma}
\label{avc-labelling}
In AVCP, for each identifier $\ms{ID}$, a non-faulty process' data structure ``$\ms{labelled}[]$'' is such that $|\ms{labelled}.\lit{values}()| \leq n$.
Further, no two binary consensus instances are labelled by the same value $l$.

%At most $n$ binary consensus instances will be labelled, and no two labels will conflict insofar as the collision resistance of $H$ and the pre-image resistance of hash functions $H_1$ and $H_2$ invoked in the traceable ring signature scheme \cite{Fujisaki2007} holds.
\end{lemma}

\begin{proof}
	For the first part of the lemma, recall that an instance $inst$ is only moved to $\ms{labelled}$ at line~\ref{line:avc:labelinst} upon AARB-delivery of $(m, \sigma)$.
	By AARB-Unicity, a non-faulty process will AARB-deliver at most one tuple $(m, \sigma)$ for every process $p_i$ s.t. $\sigma \leftarrow \lit{Sign}(i, \ms{ID}, m)$.
	Since $|P| = n$, at most $n$ instances $inst$ will thus be moved to $\ms{labelled}$.

	For the second part of the lemma, recall that we assume signature uniqueness.
	Thus, every tuple $(m, \sigma)$ anonymously broadcast by a non-faulty process is unique.
	Moreover, any duplicate tuple $(m, \sigma)$, even if broadcast by a faulty process, cannot be AARB-delivered twice (checked at line~\ref{line:ab:ifnotin}).
	Thus, by the collision-resistance of $H$, $H(m \mid \mid \sigma) \neq H(m' \mid \mid \sigma')$ for any two distinct tuples $(m, \sigma)$ and $(m', \sigma')$ that are AARB-delivered.
	Since each instance is labelled by $H(m \mid \mid \sigma')$ for some tuple $(m, \sigma)$, it follows that no two instances will conflict in label.
\end{proof}
% defer part about A_1 to TRS section and refer to?

\begin{lemma}
\label{avc-sameinfo}
Consider AVCP. Let $p$ be a non-faulty process. Then, given that the ``for'' loop on lines \ref{line:avc:forlooptemp2} and \ref{line:avc:delivertemp2} is executed as an atomic operation, all instances in $\ms{unlabelled}.\lit{values}()$ will always contain the same set of messages.
\end{lemma}

\begin{proof}
	At the protocol's outset, all instances of binary consensus are in $\ms{labelled}$, where no messages have been received.
	The only place that messages are deposited to unlabelled instances is at lines \ref{line:avc:forlooptemp2} and \ref{line:avc:delivertemp2}, where messages are deposited into \textit{all} members of $\ms{unlabelled}$.
	Finally, note that instances can only moved out of $\ms{unlabelled}$, which occurs at line~\ref{line:avc:labelinst}.
%	Denote the set of instances labelled with $\perp$ by $L$.
%	At the protocol's outset, no values are receipt in any instance in $L$, and $|L| = n$.
%	Then, an instance $\ms{BIN\_CONS}[i, \perp]$ is potentially populated only in the block beginning at line \ref{line:bchandler:receiptzeros} upon receipt of a message of the form $($\MSG$, 0, r, \neg S=\{l_1, \cdots, l_k\})$.
%	For such a message to be processed, labels must eventually be assigned to instances of the form $\ms{BIN\_CONS}[i, label]$ s.t. $label \in S$ for all $label \in S$ (line \ref{line:bchandler:label}).
%	Given this holds, the same message \MSG$[r, \perp](0)$ is deposited into all instances in $L$ (line \ref{line:bchandler:zeroperpdep}), the only time where $L$ is populated.
%	This holds for the lifetime of the protocol as $L$ only shrinks in size from the protocol's outset.
\end{proof}

\begin{lemma}
\label{avc-labelsend}
In AVCP, a non-faulty process $p$ sends a message of the form $(\ms{ID}, $\TAG$, r, \ms{label}, b)$, then there locally exists an instance of binary consensus $inst$ such that $\ms{labelled}[\ms{label}] = inst$.
\end{lemma}

\begin{proof}
Suppose $p$ sends $(\ms{ID}, $\EST$, r, \ms{label}, b)$.
If $b = 1$, then one of three scenarios holds:
\textit{(i)} $r = 1$ and $p$ has invoked $inst.\lit{bin\_propose}(1)$ where $\ms{labelled}[\ms{label}] = inst$,
\textit{(ii)} $r \geq 1$, $p$ has received $t + 1$ messages of the form  $(\ms{ID}, $\EST$, r, \ms{label}, b)$ and has not yet broadcast $(\ms{ID}, $\EST$, r, \ms{label}, b)$, or
\textit{(iii)} their value $est$ was set to $b = 1$ at the end of the previous round of binary consensus.
In case (i), note that $\lit{bin\_propose}(1)$ (at line~\ref{line:avc:invokebpone}) is executed after $inst$ is labelled (at line~\ref{line:avc:labelinst}).
Note that in both cases (ii) and (iii), $p$ must have processed messages of the form $(\ms{ID}, $\EST$, r, \ms{label}, 1)$.
As described in the text, these messages are only processed \textit{after} the AARB-delivery of the corresponding message $(m, \sigma)$ where $\ms{label} = H(m \mid \mid \sigma)$ and some instance $inst$ is labelled by $\ms{label}$ (at line~\ref{line:avc:labelinst}).
If $b = 0$, only case (ii) needs to be considered, since the other $\lit{broadcast}$ call is handled by the handler (Algorithm~\ref{alg:bchandler}).
Here, $b = 1$ must have been initially broadcast, which is described by cases (i) and (iii).
Suppose $p$ sends $(\ms{ID}, $\AUX$, r, \ms{label}, b)$.
Note that $b \neq 0$, since the handler handles this case. %TODO: elaborate slightly
Thus, we consider $b = 1$.
$b$ must have been BV-delivered (line \ref{line:obc:bvalsnonempty} of Algorithm~\ref{alg:obc}), requiring $p$ to have processed messages of the form $(\ms{ID}, $\EST$, r, \ms{label}, b)$, which requires some instance to be labelled by $\ms{label}$ to do.
This exhausts the possibilities.
\end{proof}

\begin{lemma}
\label{avc-deposit}
%Suppose a non-faulty process AARB-broadcasts $m$. Then, all messages sent by any non-faulty process relating to $m$ will eventually be deposited into some instance $\ms{BIN\_CONS}[i, H(m)]$ by all non-faulty processes.
Consider AVCP. Let $r \geq 1$ be an integer. Then, all messages sent by non-faulty processes of the form $(\ms{ID}, $\TAG$, r, \ms{label}, b)$ where \TAG \ $\in$ $\{$\EST $,$ \AUX $\}$ and $b \in \{0, 1\}$, and $(\ms{ID}, $\TAG$, r, ones)$ where \TAG \ $\in$ $\{$\ESTONES $,$ \AUXONES $\}$ and $b \in \{0, 1\}$, are eventually deposited into the corresponding binary consensus instances of a non-faulty process.
\end{lemma}

\begin{proof}
	Let $p$ be the recipient of a message in the above forms.
	We first consider messages of the form $(\ms{ID}, $\TAG$, r, \ms{label}, b)$ where \TAG \ $\in$ $\{$\EST $,$ \AUX $\}$ and $b \in \{0, 1\}$.
	By Lemma~\ref{avc-labelsend}, the non-faulty sender of such a message must have AARB-delivered some message $(m, \sigma)$ such that $\ms{label} = H(m \mid \mid \sigma)$.
	By AARB-Termination-2, $p$ will eventually AARB-deliver $(m, \sigma)$ and subsequently label an instance of binary consensus which is uniquely defined by $\ms{label}$ (Lemma~\ref{avc-labelling}).
	Thus, $p$ will eventually process $(\ms{ID}, $\TAG$, r, \ms{label}, b)$ in instance $inst = \ms{labelled}[\ms{label}]$.
	Consider messages of the form $(\ms{ID}, $\TAG$, r, ones)$ where \TAG \ $\in$ $\{$\ESTONES $,$ \AUXONES $\}$ and $b \in \{0, 1\}$.
	Since the sender is non-faulty, every element of $ones$ (i.e. every label) must correspond to a message that was AARB-delivered by the sender.
	By AARB-Termination-2, $p$ will eventually AARB-deliver these messages, label instances of binary consensus (at line \ref{line:avc:labelinst}) and thus progress beyond line~\ref{line:avc:wuntillabel}.
	Thus, $p$ will deposit 0 in all instances \textit{not} labelled by elements of $ones$ (at lines \ref{line:avc:delivertemp1} and \ref{line:avc:delivertemp2}), which is the prescribed behaviour.
\end{proof}

\begin{lemma}
	\label{avc-reachline}
	In AVCP, every non-faulty process reaches line \ref{line:avc:binpropzeros}, i.e. decides 1 in $n - t$ instances of binary consensus.
\end{lemma}

\begin{proof}
	Given at most $t$ faulty processes, at least $n - t$ non-faulty processes invoke AARBP at line \ref{line:avc:arbbcast}, broadcasting a value that satisfies $\lit{valid}()$ by assumption.
	By AARB-Termination-1, every message anonymously broadcast by a non-faulty process is eventually AARB-delivered by all non-faulty processes.
	Suppose a non-faulty process $p$ labels an instance of consensus $\ms{inst}$ by one of these messages.
	By Lemma~\ref{avc-deposit}, all messages sent by non-faulty processes associated with $\ms{inst}$ are eventually processed by $p$.
	Then, all non-faulty processes will invoke $\lit{bin\_propose}(1)$ in at least $n - t$ instances of binary consensus (i.e. while blocked at line~\ref{line:avc:wuntilnminust}).
	Consider the first $n - t$ instances of consensus for which a decision is made at.
	By the intrusion tolerance of binary consensus, no coalition of faulty processes can force 0 to be a valid value in these consensus instances, since all non-faulty processes must invoke $\lit{bin\_propose}(1)$ in them. %TODO: reword %TODO: reword %TODO: reword %TODO: reword
	Then, by BBC-Agreement, BBC-Validity and BBC-Termination, all non-faulty processes will eventually decide 1 in $n - t$ $\ms{BIN\_CONS}$ instances.
	Thus, all non-faulty processes eventually reach line \ref{line:avc:binpropzeros}.
\end{proof}

\noindent We now prove that AVC-Validity holds.

\begin{theorem}
\label{avc-val}
AVCP satisfies AVC-Validity, which is stated as follows.
Consider each non-faulty process that invokes $\lit{AVC-decide}[\ms{ID}](V)$ for some $V$ and a given $\ms{ID}$.
Each value $v \in V$ must satisfy $\lit{valid}()$, and $|V| \geq n - t$.
Further, at least $|V| - t$ values correspond to the proposals of distinct non-faulty processes.

\end{theorem}

\begin{proof}
	Let $p$ be a non-faulty process.
	By construction of the binary consensus algorithm, $p$ must process messages of the form $(\ms{ID}, $\TAG$, r, \ms{label}, b)$ where \TAG \ $\in$ $\{$\EST $,$ \AUX $\}$ and $b \in \{0, 1\}$ to reach a decision.
	By Lemma~\ref{avc-deposit}, $p$ must have AARB-delivered the corresponding message $(m, \sigma)$ such that $\ms{label} = H(m \mid \mid \sigma)$ before deciding.
	That is, $\ms{proposals}$ must have been populated with the corresponding message $(m, \sigma)$ at line~\ref{line:avc:addtoprops}, and $\lit{valid}()$ must be true for $(m, \sigma)$, checked previously at line~\ref{line:avc:ifvalid}.
	By Lemma \ref{avc-reachline}, $p$ eventually decides 1 in $n - t$ instances of binary consensus.
	Upon each decision, the corresponding message is added to $V = \ms{decided\_ones}$ (at line~\ref{line:avc:addone}).
	For the latter component of the definition, we note that ARB-Unicity ensures that each value in $V$ contains a signature produced by a different process.
	So, given $V$ is decided, at most $t$ of the corresponding signatures could have been produced by faulty processes.
\end{proof}

% all unlabelled are identical, so will terminate at same time
\begin{lemma}
\label{avc-equivall}
%Let $t < \frac{n}{3}$. Then, fixing $r$, if AVC\_BV\_broadcast \EST$[r,label](b_{label})$ or broadcast \AUX$[r,label](b_{label})$ is invoked in all $n$ instances of $\ms{BIN\_CONS}$ by a non-faulty process, then all non-faulty processes eventually process \EST$[r,label](b_{label})$ or \AUX$[r,label](b_{label})$ (respectively) in all $n$ instances.
Consider AVCP. Fix $r \geq 1$. Given that a non-faulty process $p$ invokes ``$\lit{broadcast}$ $(\ms{ID}, $\EST$, r, \ms{label}, b)$'' or ``$\lit{broadcast}$ $(\ms{ID}, $\AUX$, r, \ms{label}, b)$'' in all $n$ instances of binary consensus, then all non-faulty processes interpret the corresponding messages sent by $p$ as if $p$ were executing the original binary consensus algorithm.
\end{lemma}

\begin{proof}
	Without loss of generality, suppose $p$ invokes ``$\lit{broadcast}$ $(\ms{ID}, $\EST$, r, \ms{label}, b)$'' in all $n$ instances of binary consensus.
	For every instance $inst$ such that $b = 1$, $p$ broadcasts $(\ms{ID}, $\EST$, r, \ms{label}, b)$ (at line~\ref{line:avc:bcastest}) which is identical to the behaviour in the original algorithm.
	Note that $p$ does not broadcast $(\ms{ID}, $\EST$, r, \ms{label}, b)$ here when $b = 0$.
	By Lemma~\ref{avc-labelsend}, $p$ must have labelled $inst$ with $\ms{label}$. Thus, $p$ adds $\ms{label}$ to $\ms{ones}[$\EST$][r]$ (at line \ref{line:avc:appendonesest}).
	Once $n$ invocations of ``$\lit{broadcast}$ $(\ms{ID}, $\EST$, r, \ms{label}, b)$'' have been executed, $\ms{ones}[$\EST$][r]$ is sent to all non-faulty processes.
	By Lemma~\ref{avc-deposit}, all non-faulty processes deposit $0$ in all instances not labelled by elements of $\ms{ones}[$\EST$][r]$, corresponding to all broadcasts \textit{not} performed by $p$ previously.
	In this sense, this behaviour is equivalent to having broadcast all values $(\ms{ID}, $\EST$, r, \ms{label}, b)$ where $b = 0$.
%	Consider an arbitrary round $r$.
	
%	Firstly, consider messages of the form \EST$[r,label](b_{label})$. By Lemma \ref{avc-equivones}, all messages of the form \EST$[r,label](1)$ are handled. So, we consider messages of the form \EST$[r,label](0)$. To propose 0 in some instance $\ms{BIN\_CONS}$, $p$ must have incremented $est\_count[r]$ to $n$ (on reaching line \ref{line:bchandler:incestcount} $n$ times, having stored all proposed 1's in $\ms{est\_ones}[r]$ at line \ref{line:bchandler:popestones} prior). Then, $p$ broadcasts $($\EST$, 0, r, \neg \ms{est\_ones}[r])$, corresponding to proposing zeroes in all other instances not labelled by elements in $\ms{est\_ones}[r]$. Then, at line \ref{line:bchandler:waitzeroarb}, by $p$'s correctness and AARB-Termination-2, all values in $\ms{est\_ones}[r]$ will be eventually known by all non-faulty processes, and so will consider \EST$[r,i](0)$ as receipt for all $i \notin \ms{est\_ones}[r]$. Moreover, every round will be considered, as $p$ does not terminate upon deciding.

%	Now, consider messages of the form \AUX$[r,label](b_{label})$. By Lemma \ref{avc-equivones}, all proposals of 1 are handled. Then, let $p$ denote a non-faulty process. Consider the proposal of \AUX$[r,H(m)](0)$. By the same reasoning as above, $p$ will eventually reach line \ref{line:bchandler:bcastauxzeros}, broadcasting $($\AUX$, 0, r ,\neg \ms{aux\_ones[r]})$, and since $p$ is non-faulty, all other non-faulty processes will eventually reach line \ref{line:bchandler:zerodep}, depositing all zero values into relevant instances of $\ms{BIN\_CONS}$.
\end{proof}

\begin{theorem}
\label{avc-term}
AVCP ensures AVC-Termination. That is, every non-faulty process eventually invokes $\lit{AVC-decide}[\ms{ID}](V)$ for some vector $V$ and a given $\ms{ID}$.
\end{theorem}

\begin{proof}
	Consider a non-faulty process $p$.
	By Lemma \ref{avc-reachline}, $p$ reaches line \ref{line:avc:binpropzeros} with $n - t$ values in their local array $\ms{decided\_ones}$.
	At this point, $p$ invokes $\lit{bin\_propose}(0)$ in all other instances of binary consensus.
	Since all $n$ instances are thus executed, $p$ will eventually invoke ``$\lit{broadcast}$ $(\ms{ID}, $\EST$, r, \ms{label}, b)$'' in all $n$ instances of binary consensus.
	By Lemma~\ref{avc-equivall}, non-faulty processes interpret $p$'s corresponding messages as if $p$ were executing the original binary consensus algorithm.
	Similarly, $p$ will eventually invoke ``$\lit{broadcast}$ $(\ms{ID}, $\AUX$, r, \ms{label}, b)$'' in all $n$ instances.

	Note that DBFT~\cite{CGLR18} is guaranteed termination with $n$ instances of binary consensus as each instance is guaranteed to terminate due to BBC-Termination.
	It remains to show that AVCP preserves BBC-Termination for all $n$ instances of consensus.
	Lemma 11 in \cite{CGLR18} states that, at some point after which the global stabilistation time (GST) is reached, any execution of the original binary consensus routine eventually executes in synchronous steps.
	Note that in our construction, the slowest binary consensus instance is no slower than if $n$ instances of the original binary consensus routine were executing.
	For example, in AVCP, a message of the form $(\ESTONES, r, ones)$ is sent after $n$ invocations of ``$\lit{broadcast}$ $(\ms{ID}, $\EST$, r, \ms{label}, b)$'' are performed for a given $r$.
	Then, Lemma~\ref{avc-equivall} implies that messages sent by non-faulty processes remain unchanged as in an execution of $n$ instances of the original binary consensus routine.
	Thus, Lemma 11 holds for the slowest instance of consensus.
	Since all other instances are faster, they too are synchronised, and so Lemma 11 holds for them also.
	By Lemma 9 in \cite{CGLR18}, which asserts that BBC-Termination holds for an instance of binary consensus, it follows that each non-faulty process will decide a value in the $n$ instances of consensus in AVCP.
	At this point, $p$ immediately invokes $\lit{AVC-decide}$ with respect to $\ms{ID}$ and $\ms{decided\_ones}$ (at line \ref{line:avc:returnprops}).
	Since $p$ was arbitrary (but non-faulty), it follows that AVC-Termination holds.
\end{proof}

\begin{theorem}
\label{avc-agree}
AVCP ensures AVC-Agreement. That is, all non-faulty processes that invoke $\lit{AVC-decide}[\ms{ID}](V)$ do so with respect to the same vector $V$ for a given $\ms{ID}$.
\end{theorem}

\begin{proof}
	Consider two non-faulty processes, $p_i$ and $p_j$ that have decided (AVC-Termination).
	Suppose that $p_i$ has decided $\ms{decided\_ones}$.
	Each value in $\ms{decided\_ones}$ must have been AARB-delivered to $p_i$ in order to decide 1 in the corresponding instances of binary consensus (Lemma~\ref{avc-labelsend}). % at line x?
	By AARB-Unicity and AARB-Termination-2, $p_j$ will AARB-deliver all values in $\ms{decided\_ones}$.
	By Lemma~\ref{avc-reachline}, all non-faulty processes eventually decide 1 in $n - t$ instances of binary consensus.
	Consequently, $p_i$ and $p_j$ participate in all $n$ instances of binary consensus which by Lemma~\ref{avc-equivall} is equivalent to participating in $n$ instances of binary consensus as per the original algorithm.
	Then, by BBC-Agreement, $p_j$ decides 1 in an instance of binary consensus if and only if $p_i$ decides 1 in that instance.
	Consequently, each corresponding value, which have all been AARB-delivered prior to deciding 1 in each instance of binary consensus (Lemma~\ref{avc-labelsend}) will be added to $p_j$'s local array $\ms{decided\_ones}'$.
	Thus, $\ms{decided\_ones}' = \ms{decided\_ones}$.

%	By Lemma \ref{avc-reachline}, at least $n - t$ 1's are decided, and the corresponding proposals are deposited in $p_i$ and $p_j$'s $\ms{decided\_ones}$ arrays.
%	Consider a particular instance $inst$ that decides 1 at this point.
%	By the intrusion tolerance of the binary consensus routine, a non-faulty process must have invoked $inst.\lit{bin\_propose}(1)$.
%	Since this can only be issued at line \ref{line:avc:invokebpone}, after the AARB-delivery of $(m, \sigma)$, it follows from AARB-Unicity and AARB-Termination-2 that all non-faulty processes AARB-deliver $(m, \sigma)$
%	In particular, given $p_i$ labels an instance $H(m \mid \mid \sigma)$, $p_j$ eventually will too.
%	Moreover, by Lemma \ref{avc-deposit}, all messages sent by non-faulty processes relating to \ms{inst} will eventually be processed by $p_i$ and $p_j$.
%	Then, by BBC-Agreement, for all instances of binary consensus for which $p_i$ decides 1 in, so too does $p_j$.
%	Then, all labels of instances for which 1 is decided in are added to the respective set $\ms{decided\_ones}$.
%	By AVC-Termination, $p_i$ and $p_j$ eventually return $\ms{decided\_ones}_i = \ms{decided\_ones}_j$.
%	Since $p_i$ and $p_j$ were arbitrary non-faulty processes, it it follows that AVC-Agreement holds.
\end{proof}
%We restate the definition of AVC-Anonymity, then prove it holds under arbitrary composition in a restricted setting and sequential composition in a general setting under different assumptions:

\begin{theorem}
\label{avc-anon}
AVCP ensures AVC-Anonymity under the assumptions of our model (Section~\ref{sec:model:model}) and the three conditions described in Section~\ref{sec:anon}.
AVC-Anonymity is stated as follows:
Suppose that non-faulty process $p_i$ invokes $\lit{AVCP}[\ms{ID}](m)$ for some $m$ and given $\ms{ID}$, and has previously invoked an arbitrary number of $\lit{AVCP}[\ms{ID_j}](m_j)$ calls where $\ms{ID} \neq \ms{ID_j}$ for all such $j$.
Suppose that the adversary $A$ is required to output a guess $g \in \{1, \dots, n\}$, corresponding to the identity of $p_i$ after performing an arbitrary number of computations, allowed oracle calls and invocations of networking primitives.
Then $Pr(i = g) = \frac{1}{n - t}$.
%		Suppose that $m$ is proposed by a non-faulty process $p_i$.
%		Then, it is impossible for the adversary to determine the value of $i$ with probability $Pr > \frac{1}{n - t}$.
\end{theorem}

\begin{proof}
	We note that each process invokes AARBP exactly once at the beginning of any execution of AVCP.
	In the case all communication is performed over anonymous channels, it suffices to observe that no process sends more than one message of a given type across many executions of AVCP.
	Then, from the proof of correctness of traceable broadcast (Lemma~\ref{anon-bcast}), it follows that $A$ is unable to output a guess $g$ s.t. $g \neq \frac{1}{n - t}$.
	In the case that all communication except for the initial broadcasting is performed over reliable channels, the proof of AVC-Anonymity follows from the proof of AARB-Anonymity (Lemma~\ref{arb-anon}).
\end{proof}

\section{Ensuring termination in binary consensus} \label{sec:bvbincons}
The terminating algorithm (Figure 2 in~\cite{CGLR18}) uses in each round an additional $\lit{broadcast}$, which is performed by a rotating coordinator.
This message contains the header \COORDVALUE.
Note that $i \in \{1, \dots, n\}$ denotes the index of the process $p_i$ who is locally executing instructions, and that these instructions are executed after $\ms{bin\_values}[r] \neq \emptyset$ (i.e. after line~\ref{line:obc:bvalsnonempty} of Algorithm~\ref{alg:obc}).

\begin{algorithm}
\begin{algorithmic}[1]
	\caption{Additional broadcast in Figure 2 of \cite{CGLR18}}
\State $coord \leftarrow (r - 1) \pmod{n} + 1$
\If{$i = coord$}
\State ${w} = \ms{bin\_values}[r]$
\State $\lit{broadcast}$ $($\COORDVALUE$, r, w)$
\EndIf
\end{algorithmic}
\end{algorithm}

This broadcast call can be handled exactly as in the logic beginning at lines \ref{line:avc:uponbcastest} and \ref{line:avc:uponbcastaux} of the handler (Algorithm~\ref{alg:bchandler}), provided that for a given round $r$, the coordinator is common across all $n$ instances of consensus.
No other communication steps are added in the terminating algorithm.

In the non-terminating binary consensus algorithm, a process executes indefinitely after invoking $\lit{decide}()$.
The terminating variant, by contrast, imposes certain conditions upon termination, which are checked at the end of each round $r$.
In the context of the following algorithm excerpt, a process that invokes the instruction $\lit{halt}$ discontinues executing instructions in the binary consensus instance that $\lit{halt}$ was called in, and drops all related messages.

\begin{algorithm}
\begin{algorithmic}[1]
	\caption{Termination conditions in Figure 2 of \cite{CGLR18}}
	\If{$\lit{decide}()$ invoked in round $r$}
\WUntil $\ms{bin\_values[r]} = \{0, 1\}$ \EndWUntil
\Else
\If{$\lit{decide}()$ invoked in round $r - 2$}
$\lit{halt}$
\EndIf
\EndIf
\end{algorithmic}
\end{algorithm}

It is shown in \cite{CGLR18} that, given that some non-faulty process decides in round $r$, all non-faulty processes will decide by round $r + 2$.
Note that a process may invoke $\lit{decide}()$ in different rounds in different instances of binary consensus.
Thus, a process may invoke $\lit{halt}$ in some, but not all, instances of binary consensus.
Suppose that a process has invoked $\lit{halt}$ in $k$ instances of binary consensus, where $1 < k < n$, in some round $r$.
Then, broadcasting a message with header \ESTONES \ in the handler (Algorithm~\ref{alg:bchandler}), say, will be impossible, since $counts[$\EST$][r]$ will never be incremented to $n$.

To cope with this problem, we define a new termination condition.
Let $r_{max}$ be the largest round number of the $n$ instances of binary consensus that a given process decides in.
Then, that process can invoke $\lit{halt}$ at the end of round $r_{max} + 2$.

%\newpage

\section{AARBP} \label{sec:arbproof}
In this appendix, we first present Anonymised All-to-all Reliable Broadcast Protocol.
Notably, the protocol as presented terminates in after single anonymous message step and two regular message steps.
We then show that it satisfies each property of the anonymity-preserving all-to-all reliable broadcast problem described in Section~\ref{sec:arb:arb}.
\subsection{Protocol}

\vspace{0.5em}\noindent{\bf State and messages.}
%\paragraph{State and messages.}
%\subsubsection{State} 
Each process tracks $\ms{ID}$, which identifies an instance of AARB-broadcast.
For each $\ms{ID}$, each process tracks two buffers: $\ms{m\_buffer}$, corresponding to messages that they may AARB-deliver, and $\ms{m\_delivered}$, corresponding to all messages that they have AARB-delivered; both are initially set to $\emptyset$.
%\paragraph{Message types.}
%\subsubsection{Message types} 
For a given instance of AARBP identified by $\ms{ID}$, all messages sent by non-faulty processes must contain $\ms{ID}$.
Similarly, messages must contain one of three headers:
%\begin{itemize}
%	\item 
	\INIT, corresponding to the initial broadcast of a process' proposal,
%	\item 
	\ECHO, corresponding to an acknowledgement of the \INIT \ message, or
%	\item 
	\READY, corresponding to an acknowledgement that enough processes have received the message to ensure eventual, safe AARB-delivery.
%\end{itemize}

\vspace{0.5em}\noindent{\bf Protocol.}
%\paragraph{Protocol.}
We now present AARBP (Algorithm~\ref{alg:arb}).
\begin{algorithm}[ht]
	% state in pseudo-code
%	\begin{multicols}{2}
	\caption{AARBP} %as invoked by $p'$}
	{\footnotesize
	\label{alg:arb}
	\begin{algorithmic}[1]
%	\Part{Initial State}
%	\State $\ms{m\_buffer} \leftarrow \emptyset$ \Comment{Buffer of messages}%\Comment{Valid, non-duplicate messages, where $p'$ received and processed \INIT$(m)$ for all $m \in \ms{m, sigmas\_buffer}$}
%	\State $\ms{m\_delivered} \leftarrow \emptyset$  \Comment{Delivered messages}%\Comment{Messages that have been ARB-delivered}
%	\State $\ms{ID}$ \Comment{Identifier for the instance of AARB-broadcast}
%	\State $pk_N$, $sk_i$, $\ms{issue}$ \Comment{Distributed a priori}%\Comment{Set of $n$ public keys, whoser contents and order is agreed upon a priori}
%	\State $L = (\ms{issue}, pk_N)$ \Comment{Tag derived from $\ms{issue}$ and $pk_N$}
%	\EndPart
%	\Statex
	\Part{$\lit{AARBP}[\ms{ID}](m')$}
	\State $\sigma' \leftarrow \lit{Sign}(i, \ms{ID}, m')$ \label{line:arb:sign}
	\State $\lit{anon\_broadcast}$ $(\ms{ID},$\INIT$, m', \sigma')$ \label{line:ab:initbcast}
	\EndPart
	\Statex
	\Upon{initial receipt of $(\ms{ID},$\INIT$, m, \sigma)$} \label{line:ab:init}
	\State $\ms{valid} \leftarrow (\lit{VerifySig}(\ms{ID}, m, \sigma) = 1)$ \label{line:ab:isvalid}
	%\Comment{Confirming $s(v)$ exists and is well-formed}
	\If{$\ms{valid}$}
	\For{\textbf{each} $(m^{*}, \sigma^{*}) \in \ms{m\_buffer}$ $\cup$ $\ms{m\_delivered}$} \label{line:ab:beginfor}
	\If{$\lit{Trace}(\ms{ID}, m, \sigma, m^{*}, \sigma^{*}) \neq$ ``indep''}
	\State $\ms{valid} \leftarrow \lit{false}$ \Comment{Double-signing detected}
	\State \textbf{break}
	\EndIf
	\EndFor \label{line:ab:endfor}
	\State $\ms{m\_buffer} \leftarrow \ms{m\_buffer} \cup \{(\ms{m, \sigma})\}$ \label{line:ab:deposit}
	\If{$\ms{valid}$}
	\State $\lit{broadcast}$ $(\ms{ID},$\ECHO$, m, \sigma)$ \label{line:ab:bcastecho}
	\EndIf
	\EndIf
%	\If{\ms{s(v)} is not unique}
%	\State Consider the process with $pk = pk_v$ deduced from the LRS to be faulty
%	\EndIf
	\EndUpon
%	\newpage
\Statex	
%	}{
	\Upon{receipt $(\ms{ID},$\ECHO$, m, \sigma)$ from $\lfloor \frac{n + t}{2} \rfloor + 1$ proc.}\label{line:ab:recvechos}
	\If{$(\ms{ID},$\READY$, m, \sigma)$ not yet broadcast}
	\State $\lit{broadcast}$ $(\ms{ID},$\READY$, m, \sigma)$ \label{line:ab:echosbcastready}
	\EndIf
	\EndUpon

	\Statex
	\Upon{receipt $(\ms{ID},$\READY$, m, \sigma)$ from $(t + 1)$ proc.} \label{line:ab:toneif}
	\If{$(\ms{ID},$\READY$, m, \sigma)$ not yet broadcast}
	\State $\lit{broadcast}$ $(\ms{ID},$\READY$, m, \sigma)$ \label{line:ab:tone}
	\EndIf
	\EndUpon
	\Statex
	\Upon{receipt $(\ms{ID},$\READY$, m, \sigma)$ from $(2t + 1)$ proc.}\label{line:ab:twotone}
%	\State $\ms{deliver} \leftarrow true$
%	\For{each $\ms{m, sigma}^* \in \ms{m, sigmas\_delivered}$}
%	\If{$\ms{Trace(m, sigma, m, sigma^{*})} \neq ``indep''$}
%	\State $\ms{deliver} \leftarrow \ms{false}$
%	\State break
%	\EndIf
%	\EndFor
%	\If{\ms{deliver}}
	\If{$(m, \sigma) \notin \ms{m\_delivered}$} \label{line:ab:ifnotin}
	\State $\ms{m\_delivered}$ $\leftarrow$ $\ms{m\_delivered}$ $\cup$ $\{(m, \sigma)\}$
	\State $\ms{m\_buffer}$ $\leftarrow$ $\ms{m\_buffer} \setminus \{(\ms{m, \sigma)}\}$
	\State $\lit{AARB-deliver}[\ms{ID}](m, \sigma)$ \label{line:ab:arbdeliver}
	\EndIf
%	\EndIf
%	\State $\ms{ids_i}$ $\leftarrow$ $\ms{ids_i}$ $\cup$ $H(v||s(v'))$
	\EndUpon
	\end{algorithmic}
	}
%	\end{multicols}
\end{algorithm}

%\subsubsection{Protocol}
\sloppy{
To begin, each process $p_i$ broadcasts $(\ms{ID},$\INIT$, m', \sigma')$ over anonymous channels (line \ref{line:ab:initbcast}), where $\sigma' \leftarrow \lit{Sign}(i, \ms{ID}, m')$.
%Importantly, the signer sends \INIT$(msg)$ to itself over anonymous channels. 
%(Note that anonymous channels could be approximated with Tor, for example, where 
%an exit node would perform
%a TCP handshake with the sending process and transmit 
%\INIT$(msg)$.) % relevant?
Upon first receipt of each message of the form $(\ms{ID},$\INIT$, m, \sigma)$, a process checks the following (line \ref{line:ab:init}):}
\begin{itemize}[noitemsep,topsep=5pt,parsep=2pt,partopsep=5pt,itemindent=1em,listparindent=1em,leftmargin=1em]
	\item Whether the signature $\sigma$ is well-formed as per $\lit{VerifySig}$, verifying the signer's membership in $P$ (line \ref{line:ab:isvalid}).
	\item Whether any message in $\ms{m\_buffer} \cup \ms{m\_delivered}$ is not independent from $m$ via $\lit{Trace}$, ensuring AARB-Unicity as $m$ is discarded if double-signing is detected (lines \ref{line:ab:beginfor} to \ref{line:ab:endfor}).
\end{itemize}
Then, given that $(m, \sigma)$ passes the above checks, $(\ms{ID},$\ECHO$, m, \sigma)$ is broadcast (line \ref{line:ab:bcastecho}).
Note that this broadcast, and all subsequent broadcasts with respect to $(m, \sigma)$, are performed over regular (reliable, non-anonymous) channels of communication.
To mitigate de-anonymisation, the signer of $m$ performs the same message processing and message propagation as all other non-faulty processes.
%
%Now, suppose that $p$ produced $msg = (v, s(v))$ and was mandated to \textit{not} broadcast \ECHO$(msg)$ by the protocol. Assuming all processes bar $p$ were to broadcast \ECHO, which is a conceivable message pattern, an interested process could observe the identities of such $n - 1$ senders and infer $p$'s identity. Moreover, to 
%To  mitigate timing analysis, $p$ should both send \INIT$(msg)$ via anonymous links and perform verification/tracing of $msg$ such that its behaviour becomes effectively indistinguishable from that of any other (non-faulty) process.

The rest of the protocol proceeds as per Bracha's reliable broadcast~\cite{Bracha1987}:
%\begin{itemize}
%	\item 
If processes receive $(\ms{ID},$\ECHO$, m, \sigma)$ from more than $\frac{n + t}{2}$ different processes, they broadcast $(\ms{ID},$\READY$, m, \sigma)$ (lines~\ref{line:ab:recvechos}--\ref{line:ab:echosbcastready}).
%	\item 
	If processes receive $(\ms{ID},$\READY$, m, \sigma)$ from $t + 1$ different processes, they broadcast $(\ms{ID},$\READY$, m, \sigma)$ if not yet done (lines~\ref{line:ab:toneif}--\ref{line:ab:tone}). This ensures convergence and implies that at least one non-faulty process must have sent $(\ms{ID},$\READY$, m, \sigma)$ to the receiving process.
%	\item 
	Once processes have received $(\ms{ID},$\READY$, m, \sigma)$ from $2t + 1$ different processes, they AARB-deliver $(m, \sigma)$ with respect to $\ms{ID}$ if not yet done (line \ref{line:ab:arbdeliver}). At this point, at least $t + 1$ non-faulty processes must have broadcast $(\ms{ID},$\READY$, m, \sigma)$. Thus, all non-faulty processes will eventually broadcast $(\ms{ID},$\READY$, m, \sigma)$ at line \ref{line:ab:tone} if not yet done.
%\end{itemize}

%We note, in the absence of the propagation of \READY \ messages, that one can achieve all properties bar ARB-Termination-2 with $t < \frac{n}{3}$ in \textit{two} message steps \cite{Raynal2018}. 
%however, proves to be 
%in contexts where agreement is paramount, like 
%
%We defer a proof of the protocol's correctness to Appendix \ref{sec:arbproof}.

\vspace{0.5em}\noindent{\bf Complexity.}
%\paragraph{Complexity.} 
%\subsection{Complexity} 
Whilst Bracha's protocol \cite{Bracha1987} requires three message steps to converge, AARBP requires an initial, anonymous message step and two (regular) message steps.
%Over Tor, we would spend additional time establishing a circuit before propagating a message over three hops (the default amount). 
Both AARBP and $n$ invocations of Bracha's protocol have a message complexity of $O(n^3)$.
If we consider unsigned messages to be of size $O(1)$, then each message propagated in AARBP is of size $O(kn)$, the size of a TRS as in \cite{Fujisaki2007} where $k$ is a security parameter. %, where $k$ is a security parameter %(of length $\sim$256 bits as per current convention in ECC\vincent{definition or reference for ECC}).
Thus, the bit complexity of all-to-all AARBP is $O(kn^4)$ (as opposed to $O(n^3)$ in all-to-all reliable broadcast~\cite{Bracha1987}).
To reduce the bit complexity when some messages are delivered in order, we can hash \ECHO \ and \READY \ messages~\cite{CKPS01}.
Consequently, communication complexity can be reduced to $O((c + k)n^3)$.

We consider cryptographic overhead as if Fujisaki and Suzuki's TRS scheme~\cite{Fujisaki2007} were used.
In AARBP, each process signs one message ($O(n)$ work), verifies up to $n$ messages ($O(n^2)$ work), and perform tracing upon receipt of each \INIT \ message (up to $n$).
Tracing requires $O(n^3)$ work naively and $O(n^2)$ expected work if signatures are stored in and accessed using hash tables.
\subsection{Proof of correctness}
We now prove that AARBP satisfies each property of the anonymity-preserving all-to-all reliable broadcast problem.
\begin{remark}
	Let $m, n, t$ be positive integers s.t. $n > 3t$. We have: \[m > \frac{n - t}{2} \Leftrightarrow m \geq \lfloor \frac{n - t}{2} + 1 \rfloor.\]
\end{remark}
Thus, we refer to ``$\lfloor \frac{n - t}{2} \rfloor + 1$ processes'' and ``more than $\frac{n - t}{2}$ processes'' interchangeably.

\begin{lemma}
\label{arb-signing}
	AARBP ensures AARB-Signing.
	That is, if a non-faulty process $p_i$ AARB-delivers a message with respect to $\ms{ID}$, then it must be of the form $(m, \sigma)$, where a process $p_i \in P$ invoked $\lit{Sign}(i, \ms{ID}, m)$ and obtained $\sigma$ as output.
\end{lemma}

\begin{proof}
	Suppose $p$ is non-faulty and AARB-delivers a message $msg$.
	We aim to show that $msg = (m, \sigma)$, where $\sigma \leftarrow \lit{Sign}(i, \ms{ID}, m)$ was called by some process $p_i$.
	For a non-faulty process to broadcast $(\ms{ID}, $\READY$, msg)$, they must have either received $(\ms{ID}, $\READY$, msg)$ from $t + 1$ different processes or $(\ms{ID}, $\ECHO$, msg)$ from $\lfloor \frac{n + t}{2} \rfloor + 1$ different processes.
	Since $t + 1 > t$, there exists a non-faulty process $p'$ that must have broadcast $(\ms{ID}, $\READY$, msg)$ on receipt of $(\ms{ID}, $\ECHO$, msg)$ from $\lfloor \frac{n + t}{2} \rfloor + 1$ different processes to ensure that $msg$ was AARB-delivered to any process.
	Similarly, since $\lfloor \frac{n + t}{2} \rfloor + 1 > t$, there exists a non-faulty process $p''$ that must have broadcast $(\ms{ID}, $\ECHO$, msg)$ to ensure $(\ms{ID}, $\READY$, msg)$ was broadcast by a non-faulty process.
	Process $p''$ must have checked that $msg = (m, \sigma)$ (implicitly) and that $\lit{VerifySig}(\ms{ID}, m, \sigma) = 1$ holds (at line~\ref{line:ab:isvalid}).
	By signature unforgeability, $\lit{VerifySig}(\ms{ID}, m, \sigma) = 1$ implies that $\sigma \leftarrow \lit{Sign}(i, \ms{ID}, m)$ was called by some process $p_i$.
	Thus, the property holds.
\end{proof}

\begin{lemma}
\label{arb-anon}
AARBP ensures AARB-Anonymity under the assumptions of the model (Section~\ref{sec:model:model}) and those made in Section~\ref{sec:anon}.
AARB-Anonymity is stated as follows.
Suppose that non-faulty process $p_i$ invokes $\lit{AARBP}[\ms{ID}](m)$ for some $m$ and given $\ms{ID}$, and has previously invoked an arbitrary number of $\lit{AARBP}[\ms{ID_j}](m_j)$ calls where $\ms{ID} \neq \ms{ID_j}$ for all such $j$.
Suppose that the adversary $A$ is required to output a guess $g \in \{1, \dots, n\}$, corresponding to the identity of $p_i$ after performing an arbitrary number of computations, allowed oracle calls and invocations of networking primitives.
Then $Pr(i = g) = \frac{1}{n - t}$.
\end{lemma}
\begin{proof}
	We proceed by contradiction.
	Let $p_i$ be a non-faulty process that participates in the scenario described in the definition of AARB-Anonymity.
	Suppose that $A$ is able to output $g$ such that $Pr(i = g) \neq \frac{1}{n - t}$.
	Now, in $\lit{AARBP}$, process $p_i$ invokes $\lit{Sign}(i, \ms{ID}, m)$, producing the signature $\sigma$ (line~\ref{line:arb:sign}).
	Then, $p_i$ invokes ``$\lit{anon\_broadcast}$ $(\ms{ID},$\INIT$, m, \sigma)$'' (line~\ref{line:ab:initbcast}).
By Lemma~\ref{anon-bcast} (Section~\ref{sec:arb:arb}) and the symmetry between Algorithm~\ref{alg:tb} and AARBP, $A$ must utilise some combination of regular communication, its message history and previous computation to output $g$ s.t. $Pr(i  = g) \neq \frac{1}{n - t}$.

%	Let $p_a$ be a process that $A$ has corrupted that has invoked $\lit{anon\_receive}$ with respect to $m$.
%	Suppose that $p_a$ has invoked $\lit{anon\_receive}$ with respect to the messages of $k > 0$ non-faulty processes that invoked $\lit{anon\_broadcast}$.
%	Let $K = \{f(i_1), \dots, f(i_k)\}$ denote the indices of the $k$ non-faulty processes such that $f(i_j)$ corresponds to the process that invoked $\lit{anon\_broadcast}$ that resulted in $p_a$'s $j$th invocation of $\lit{anon\_receive}$.
%	Suppose that $A$ outputs a guess $g$ of the value of $i$ based on the sets $K$ it attains at its $t$ corrupted processes.
%	By condition 1 in Section~\ref{sec:anon}, $Pr(f(\alpha) = i) = Pr(f(\beta) = i) = \frac{1}{n - t}$ for all values $\alpha, \beta$ in each set $K$.
%	Consequently, $A$ gains no additional information about $i$ from the ordering of messages anonymously sent. 

%	We now consider the full restrictions on $A$'s behaviour subjected by the definition of AARB-Anonymity.
%	Note that all non-faulty processes invoke $\lit{AARBP}$ with respect to $\ms{ID}$, and so condition 1 in Section~\ref{sec:anon} applies to $\lit{AARBP}$.
	Firstly, we remark that $A$ cannot view the local state or computations of non-faulty processes by assumption, and so cannot, for example, view their calls to $\lit{Sign}$.
	Note that (non-faulty) processes propagate messages over regular (reliable) channels after their initial invocation of $\lit{anon\_broadcast}$ (such as at line~\ref{line:ab:bcastecho}).
	By assumption, $A$ has no inherent knowledge of when $\lit{anon\_send}$ or $\lit{anon\_receive}$ calls are made by non-faulty processes.
	Indeed, $A$ may observe when non-faulty processes respond to $\lit{anon\_receive}$ calls via messages they send over regular channels, and when $A$'s corrupted processes invoke any message passing primitive.
	Thus, under the assumptions of Section~\ref{sec:model:model} alone, it is conceivable that an adversary could correlate timing information regarding message passing over regular channels and anonymous channels in a given protocol execution to de-anonymise $p_i$.
%	To mitigate this, condition 2 (Section~\ref{sec:anon}) ensures that the timing of messages sent in any previous instance of AARBP does not aid $A$'s strategy to de-anonymise $p_i$, since no messages are sent by non-faulty processes for any other instance identified by identifier $\ms{ID_x}$, where $\ms{ID_x} \neq \ms{ID}$, when messages relating to $\ms{ID}$ are sent by any non-faulty process.
	As the network is asynchronous, the timing and relative order in which messages are sent and received by processes in $\lit{AARBP}[\ms{ID}]$ could be arbitrary.
	The independence assumption (Section~\ref{sec:anon}) assumes, however, that the order and timing of message delivery over anonymous channels is statistically independent from that of regular channels in a particular protocol execution.
	Note that this independence holds for the timing and ordering of messages sent of any number of processes during any stage of the protocol.
	In particular, this accounts for the behaviour of non-faulty processes in response to messages received from Byzantine processes of whom $A$ is free to control.
	Consequently, $A$ gains no information about the identity of $p_i$ through correlation, and so $A$ could not have produced $g$ such that $Pr(i = g) \neq \frac{1}{n - t}$ using this strategy.
	Since the independence assumption also considers timing and message ordering \textit{between} instances of AARBP, it follows that $A$ cannot de-anonymise $p_i$ by considering many executions.
	This exhausts $A$'s conceivable strategies to de-anonymise $p_i$ under our model.
%	TODO: change condition 1 to ``distribution of timing/relative message order statistically indistinguishable between anon and normal''?
%	Condition 2 in Section~\ref{sec:anon} implies that $p_i$ will not invoke $\lit{AARBP}$ until after terminating any previous instance.

\end{proof}

\begin{lemma}
\label{arb-term-1}
AARBP ensure AARB-Termination-1.
That is, if a process $p_i$ is non-faulty and invokes $\lit{AARBP}[ID](m)$, all the non-faulty processes eventually AARB-deliver $(m, \sigma)$ with respect to $\ms{ID}$, where $\sigma$ is the output of the call $\lit{Sign}(i, \ms{ID}, m)$.
\end{lemma}

\begin{proof}
	\sloppy{
	Suppose that the non-faulty process $p_i$ invokes AARBP with respect to identifier $\ms{ID}$ and message $m$.
	So, $p_i$ anonymously broadcasts $(\ms{ID},$\INIT$, m, \sigma)$, where $\sigma = \lit{Sign}(i, \ms{ID}, m)$.
	Since the network is reliable, all non-faulty processes eventually receive $(\ms{ID},$\INIT$, m, \sigma)$ (and indeed all messages broadcast by non-faulty processes).
	At line \ref{line:ab:isvalid}, $\lit{VerifySig}$ is queried with respect to $(m, \sigma)$.
	By signature correctness, $\lit{VerifySig}(\ms{ID}, m, \sigma) = 1$ is guaranteed, since $p_i$ invoked $\sigma \leftarrow \lit{Sign}(i, \ms{ID}, m)$.
	So, all non-faulty processes proceed to line \ref{line:ab:beginfor}.
	By traceability and entrapment-freeness, no non-faulty process could have received a message $(m', \sigma')$ such that $\lit{Trace}(\ms{ID}, m, \sigma, m', \sigma') \neq$ ``indep''.
	Consequently, all non-faulty processes reach line~\ref{line:ab:deposit} with value $\ms{valid}$ s.t. $\ms{valid} = \lit{true}$.
	Thus, all non-faulty processes (of which there are at least $n - t$) reach line \ref{line:ab:bcastecho}, broadcasting $(\ms{ID},$\ECHO$, m, \sigma)$.
	Since $n - t > \frac{n + t}{2}$, every non-faulty process will eventually receive more than $\frac{n + t}{2}$ $(\ms{ID},$\ECHO$, m, \sigma)$ messages, fulfilling the predicate at line~\ref{line:ab:recvechos}.
	More than $\frac{n + t}{2} > t + 1$ non-faulty processes will not have broadcast $(\ms{ID},$\READY$, m, \sigma)$, and so will do so (at line \ref{line:ab:echosbcastready}).
	Thus, all honest processes will eventually receive $t + 1$ $(\ms{ID},$\READY$, m, \sigma)$ messages, broadcasting $(\ms{ID},$\READY$, m, \sigma)$ there if not yet done.
	Then, since $n - t \geq 2t + 1$, line \ref{line:ab:twotone} will eventually be fulfilled for each non-faulty process.
Thus, every non-faulty process will AARB-deliver $(m, \sigma)$ at line \ref{line:ab:arbdeliver}, as required.}
\end{proof}

%Indeed, processes are free to propagate \READY \ messages without having received either the corresponding \INIT \ message or the \ECHO \ message from more than $\frac{n + t}{2}$ processes. Depending on the structure of the network, this flexibility may speed up convergence for non-faulty processes.

\begin{lemma}
\label{arb-val}
AARBP ensures AARB-Validity, which is stated as follows.
Suppose that a non-faulty process AARB-delivers $(m, \sigma)$ with respect to $\ms{ID}$.
Let $i = \lit{FindIndex}(\ms{ID}, m, \sigma)$ denote the output of an idealised call to $\lit{FindIndex}$.
Then if $p_i$ is non-faulty, $p_i$ must have anonymously broadcast $(m, \sigma)$.

\end{lemma}
\begin{proof}
	Let $p_i$ be non-faulty.
	From the protocol specification, $p_i$ must have invoked $\sigma \leftarrow \lit{Sign}(i, \ms{ID}, m)$.
	By definition of $\lit{Sign}$, no other process could have produced such a value of $\sigma$.
	By traceability and entrapment-freeness, no other process can produce a tuple $(m', \sigma')$ such that $\lit{Trace}(\ms{ID}, m, \sigma, m', \sigma') \neq$ ``indep''.
	Consequently, no process can produce a message that prevents non-faulty processes from propagating $(\ms{ID}, $\ECHO$, m, \sigma)$ messages.
	Thus, since $p_i$ is non-faulty, and therefore followed the protocol, $p_i$ must have anonymously broadcast $(m, \sigma)$.
\end{proof}

\begin{remark}
\label{arb-one-echo}
No non-faulty process will broadcast more than one message of the form $(\ms{ID},$\ECHO$, m', \sigma')$ where $\sigma'$ is such that $\sigma' = \lit{Sign}(x, \ms{ID}, m')$ for any message $m'$.
\end{remark}

\begin{proof}
	This follows from the proof of Lemma~\ref{arb-term-1}: if non-faulty process $p$ receives another message $(m, \sigma)$ s.t. $\sigma \leftarrow \lit{Sign}(x, \ms{ID}, m)$ for any $m$, then $\lit{Trace}(\ms{ID}, m, \sigma, m', \sigma') \neq$ ``indep'', and so $p$ will not reach line~\ref{line:ab:bcastecho}.
\end{proof}

\begin{lemma}
\label{arb-term-2-lemma}
In AARBP, if $p_i$ AARB-delivers $(m, \sigma)$ with respect to $\ms{ID}$, where $\sigma \leftarrow \lit{Sign}(x, \ms{ID}, m)$, and $p_j$ AARB-delivers $(m', \sigma')$ with respect to $\ms{ID}$, where $\sigma' \leftarrow \lit{Sign}(x, \ms{ID}, m')$, then $(m, \sigma) = (m', \sigma)$.
\end{lemma}
\begin{proof}
	For $p_i$ (resp. $p_j$) to have AARB-delivered $(m, \sigma)$ (resp. $(m', \sigma')$), $p_i$ (resp. $p_j$) must have broadcast $(\ms{ID},$\READY$, m, \sigma)$ (resp. $(\ms{ID},$\READY$, m', \sigma')$).
	Then, one of two predicates (at lines \ref{line:ab:recvechos} and \ref{line:ab:toneif}) must have been true for each process. That is, $p_i$ (resp. $p_j$) must have received either:
\begin{enumerate}
	\item $(\ms{ID},$\ECHO$, m, \sigma)$ (resp. $(\ms{ID},$\ECHO$, m', \sigma')$) from $\lfloor \frac{n + t}{2} \rfloor$ different processes, or
	\item $(\ms{ID},$\READY$, m, \sigma)$ (resp. $(\ms{ID},$\READY$, m', \sigma')$) from $(t + 1)$ different processes.
\end{enumerate}
We first prove the following claim: if $(\ms{ID},$\READY$, m, \sigma)$ is broadcast by $p_i$, and $(\ms{ID},$\READY$, m', \sigma')$ is broadcast by $p_j$, then $(m, \sigma) = (m', \sigma')$.
Suppose that both processes fulfill condition (1), receiving \ECHO \ messages from sets of processes $P$ and $P'$ respectively, where $p_i$ broadcasts $(\ms{ID},$\READY$, m, \sigma)$ and $p_j$ broadcasts $(\ms{ID},$\READY$, m', \sigma')$.
Assume that $(m, \sigma) \neq (m', \sigma')$. Then:
\begin{align*}
	|P \cap P'| &= |P| + |P'| - |P \cup P'|, \\
	&\geq |P| + |P'| - n, \\
	&>2\left(\frac{n + t}{2}\right) - n = t, \\
	&\Rightarrow |P \cap P'| \geq t + 1.
\end{align*}
Therefore, $P \cap P'$ must contain at least one correct process, say $p_c$. By Remark~\ref{arb-one-echo}, $p_c$ must have sent the same \ECHO \ message for some message to both processes, and so $p_i$ and $p_j$ must have received the same message $(\ms{ID},$\ECHO$, m, \sigma)$ = $(\ms{ID},$\ECHO$, m', \sigma')$, contradicting the assumption that $(m, \sigma) \neq (m', \sigma')$. Thus, they must broadcast the same \READY \ message.
% TODO - signatures may be different?
% TODO - continue from here

Suppose now that at least one process broadcasts \READY \ due to condition (2) being fulfilled.
Without loss of generality, assume exactly one process $p$ broadcasts \READY \ on this basis.
Then, $p$ must have received a \READY \ message from at least one correct process (since out of $t + 1$ processes, at least 1 must be non-faulty), say $p_a$.
Either $p_a$ received \READY \ from at least one correct process, say $p_b$, or it satisfied condition (1). By continuing the logic (and since $|P|$ is finite), there must exist a process $p_{(1)}$ that fulfilled condition (1). Then, by the correctness of processes $p_{(1)}, \dots, p_b, p_a$, \READY \ messages sent by $p_i$ and $p_j$ must be the same, completing the proof.

Therefore, $p_i$ broadcasting $(\ms{ID},$\READY$, m, \sigma)$, and $p_j$ broadcasting $(\ms{ID},$\READY$, m', \sigma')$ implies $(m, \sigma) = (m', \sigma')$ given that $p_i$ and $p_j$ are non-faulty.

We now directly prove the lemma. If $p_i$ AARB-delivers $(m, \sigma)$, it received $(\ms{ID},$\READY$, m, \sigma)$ from $(2t + 1)$ different processes, and thus received $(\ms{ID},$\READY$, m, \sigma)$ from at least one non-faulty process.
Similarly, if $p_j$ AARB-delivers $(m', \sigma')$, it must have received $(\ms{ID},$\READY$, m', \sigma')$ from at least one non-faulty process.
It follows from the previous claim that all non-faulty processes broadcast the same \READY \ message.
Thus, $p_i$ and $p_j$ AARB-deliver the same tuple.
\end{proof}

\begin{lemma}
\label{arb-term-2}
AARBP ensures AARB-Termination-2.
That is, if a non-faulty process AARB-delivers $(m, \sigma)$ with respect to $\ms{ID}$, then all the non-faulty processes eventually AARB-deliver $(m, \sigma)$ with respect to $\ms{ID}$.
\end{lemma}

\begin{proof}
	By Lemma \ref{arb-term-2-lemma}, all non-faulty processes that AARB-deliver a message $(m', \sigma')$ s.t. $\sigma' \leftarrow \lit{Sign}(x, \ms{ID}, m')$ AARB-deliver $(m, \sigma)$.
	Then, $p_i$ must have received the message $(\ms{ID},$\READY$, m, \sigma)$ from $(2t + 1)$ processes (line \ref{line:ab:twotone}), at least $t + 1$ of which must be non-faulty.
	These $t + 1$ processes must have broadcast $(\ms{ID},$\READY$, m, \sigma)$ (at line \ref{line:ab:echosbcastready} or \ref{line:ab:tone}), and so every non-faulty process will eventually receive $(t + 1)$ $(\ms{ID},$\READY$, m, \sigma)$ messages, and thus broadcast $(\ms{ID},$\READY$, m, \sigma)$ at some point.
	Given there are $n - t \geq 2t + 1$ non-faulty processes, each non-faulty process eventually receives $(\ms{ID},$\READY$, m, \sigma)$ from at least $2t + 1$ processes. Therefore, every non-faulty process will AARB-deliver $(m, \sigma)$.
\end{proof}

\begin{lemma}
\label{arb-uni}
AARBP ensures AARB-Unicity, which is stated as follows.
Consider any point of time in which a non-faulty process $p$ has AARB-delivered more than one tuple with respect to $\ms{ID}$.
Let $delivered = \{(m_1, \sigma_1), \dots, (m_l, \sigma_l)\}$, where $|delivered| = l$, denote the set of these tuples.
For each $i \in \{1, \dots, l\}$, let $out_i = \lit{FindIndex}(\ms{ID}, m_i, \sigma_i)$ denote the output of an idealised call to $\lit{FindIndex}$.
Then for all distinct pairs of tuples $\{(m_i, \sigma_i), (m_j, \sigma_j)\}$, $out_i \neq out_j$.
\end{lemma}
\begin{proof}
	We proceed by contradiction.
	Without loss of generality, suppose that $p$ AARB-delivers two tuples $(m, \sigma)$ and $(m', \sigma')$ such that $\lit{FindIndex}(\ms{ID}, m, \sigma) = \lit{FindIndex}(\ms{ID}, m', \sigma')$.
	Then, AARB-Termination-2 implies that all non-faulty processes will eventually AARB-deliver $(m, \sigma)$ and $(m', \sigma')$.
	But, Lemma \ref{arb-term-2-lemma} asserts that $\lit{FindIndex}(\ms{ID}, m, \sigma) = \lit{FindIndex}(\ms{ID}, m', \sigma')$, a contradiction.
	Thus, for each $i \in \{1..n\}$, $p$ AARB-delivers at most one tuple $(m, \sigma)$ such that $\lit{FindIndex}(\ms{ID}, m, \sigma) = i$.
\end{proof}

\begin{theorem}
\label{arb-thm}
AARBP (Algorithm~\ref{alg:arb} in Section \ref{sec:arb:arb}) satisfies the properties of AARB-Broadcast.
\end{theorem}

\begin{proof}
	From Lemmas \ref{arb-signing}, \ref{arb-anon}, \ref{arb-val}, \ref{arb-uni}, \ref{arb-term-1} and \ref{arb-term-2}, it follows that all properties of AARB-Broadcast are satisfied.
\end{proof}

\section{Composing threshold encryption and AVCP} \label{sec:evoting:abe}

Threshold encryption~\cite{Desmedt1992, Shoup2002} typically involves a set of processes who have a common encryption key and an individual share of the decryption key who must collaborate to decrypt any message.
A $k$-out-of-$n$ threshold encryption scheme requires the joint collaboration of $k$ processes to decrypt any encrypted value.
Importantly, $k - 1$ or less processes are unable to determine any additional information about a given encrypted value in collaboration.
Generally, the keying material can be reused in the sense that many values can be decrypted without compromising the security of the scheme.

AVCP assumes that $t$ processes may be Byzantine faulty.
Thus, by setting the decryption threshold to $t + 1$, a coalition of $t$ malicious processes are unable to deduce the contents of encrypted values until a non-faulty process initiates threshold decryption.
Consequently, a protocol that ensures that the contents of all non-faulty process' proposals to an instance of AVCP is not revealed until after termination can be designed as follows:
\begin{enumerate}
	\item All processes encrypt their proposal under a pre-determined public key.
	\item All processes invoke AVCP with input as their encrypted value.
	\item Upon termination, all processes initiate threshold decryption.
\end{enumerate}
In this appendix, we realise and analyse this protocol.

\subsection{Preliminaries}
We assume that the assumptions made by AVCP, including those made in our model (Section~\ref{sec:model:model}), hold true. 
In particular, the dealer generates the initial state, including the values of $n$ and $t$.
As such, the decryption threshold for the instance of threshold encryption is set to $k = t + 1$.
As done in characterising traceable ring signatures in Section~\ref{sec:model:model}, we assume processes have access to a distributed oracle which handles cryptographic operations, rather than explicitly referring to keying material as is done in practice.
In addition to the queries described in Section~\ref{sec:model:model}, we assume that the distributed oracle handles queries of the following functions:
\begin{enumerate}
	\item $c \leftarrow \lit{Enc}(\ms{ID}, m)$, which takes identifier $\ms{ID} \in \{0, 1\}^{*}$ and message $m \in \{0, 1\}^{*}$ as input, and outputs the ciphertext $c \in \{0, 1\}^{*}$. All parties (even those not in $P$) may query $\lit{Enc}$.
	\item $b \leftarrow \lit{VerifyEnc}(\ms{ID}, c)$, which takes identifier $\ms{ID}$ and ciphertext $c$ as input, and outputs a bit $b \in \{0, 1\}$. All parties may query $\lit{VerifyEnc}$.
	\item $\sigma \leftarrow \lit{ShareGen}(i, \ms{ID}, c)$, which takes integer $i \in \{1..n\}$, identifier $\ms{ID}$ and ciphertext $c$ as input, and outputs the decryption share $\sigma \in \{0, 1\}^{*}$. We restrict the interface $\lit{ShareGen}$ such that only process $p_i \in P$ may invoke $\lit{ShareGen}$ with first argument~$i$.
	\item $b \leftarrow \lit{VerifyShare}(\ms{ID}, c, \sigma)$, which takes identifier $\ms{ID}$, ciphertext $c$ and decryption share $\sigma$ as input, and outputs a bit $b \in \{0, 1\}$. All parties may query $\lit{VerifyShare}$.
	\item $m \leftarrow \lit{Dec}(\ms{ID}, c, \Sigma)$, which takes the identifier $\ms{ID}$, ciphertext $c$ and set $\Sigma = \{\sigma_1, \dots, \sigma_k\}$ of $k$ decryption shares, and outputs the plaintext $m \in \{0, 1\}^{*}$. All parties may query $\lit{Dec}$.
\end{enumerate}

These definitions are inspired by those used in defining Shoup and Gennaro's non-interactive threshold encryption scheme~\cite{Shoup2002}.
The behaviour of calls to the above functions satisfies the following properties:
\begin{itemize}
	\item \textbf{Encryption correctness and non-malleability:} $\lit{VerifyEnc}(\ms{ID}, c) = 1$ $\Longleftrightarrow$ there exists a party which previously invoked $\lit{Enc}(\ms{ID}, m)$ and obtained $c$ as a response. 
	Non-malleability is captured by the ``$\Rightarrow$'' claim, which ensures that any correct encryption was produced by a call to $\lit{Enc}$. 
	In particular, non-malleability implies that combining any value and an encryption (e.g. via concatenation) will not produce another valid encryption.
\item \textbf{Share correctness and unforgeability:} $\lit{VerifyShare}(\ms{ID}, c, \sigma) = 1$ $\Longleftrightarrow$ there exists process $p_i \in P$ that previously invoked $\lit{ShareGen}(i, \ms{ID}, c)$ and obtained $\sigma$ as a result.
\item \textbf{Decryption correctness and decryption security:} $\lit{Dec}(\ms{ID}, c, \Sigma) = m$, where $m$ was the input to a previous call to $\lit{Enc}(\ms{ID}, m)$ $\Longleftrightarrow$ the following conditions are met:
\begin{enumerate}
	\item $c \leftarrow \lit{Enc}(\ms{ID}, m)$ was called by some party.
	\item $\Sigma = \{\sigma_1, \dots, \sigma_k\}$ ($|\Sigma| = k$), where for each distinct pair $\sigma_i, \sigma_j \in \Sigma$, distinct processes $p_i, p_j \in P$ respectively called $\sigma_i \leftarrow \lit{ShareGen}(i, \ms{ID}, c)$ and $\sigma_j \leftarrow \lit{ShareGen}(j, \ms{ID}, c)$.
\end{enumerate}
The ``$\Leftarrow$'' claim captures decryption correctness. The ``$\Rightarrow$'' claim captures security, ensuring that $k$ valid shares produced by distinct processes are required for decryption.
	\item \textbf{Encryption hiding:} Suppose $c \leftarrow \lit{Enc}(\ms{ID}, m)$ is called by a party $p$. Then, $m$ can only be obtained by either $p$ revealing $m$ or a valid call $m \leftarrow \lit{Dec}(\ms{ID}, c, \Sigma)$ being made.
\end{itemize}

\subsection{Arbitrary Ballot Election (ABE)}
At a high level, processes first perform AVCP with respect to their proposal which is encrypted under the threshold encryption scheme.
Then, processes perform threshold decryption, which requires an additional message delay for termination.
We call this algorithm an Arbitrary Ballot Election (ABE), since processes can propose arbitrary values and are ensured of certain properties that are defined in the context of electronic voting.
We describe and prove these hold in the next subsection.

\paragraph{Functions.}
Each process has access to the distributed oracle and in addition to all functions invoked AVCP (Section~\ref{sec:avc:avc}).
For a given tuple $(m, \sigma)$ we assume that the predicate $\lit{valid}()$ in AVCP returns $\lit{true}$ only if $\lit{VerifyEnc}(\ms{ID}, m) = 1$.
That is, $m$ was the result of a query to $\lit{Enc}(\ms{ID}, msg)$ for some $msg \in \{0, 1\}^{*}$.

\paragraph{State.}
Each process tracks the variables
\ms{ID}, a string uniquely identifying an instance of ABE,
\ms{has-broadcast}, a Boolean that is initially $\lit{false}$,
\ms{unique-encs} and \ms{plaintexts}, sets of messages $m \in \{0, 1\}^{*}$ that are initially empty, and
\ms{partial-decs}, which maps ciphertexts to sets of decryption shares, each of which are initially empty.

\paragraph{Messages.}
In addition to messages propagated in AVCP (which includes messages from AARBP), messages of the form $($\DECS$, \ms{ID}, \ms{encs})$ are propagated, where $\ms{encs}$ is a map, each value of which is a singleton set.

\paragraph{Protocol description.}
\begin{algorithm}[htp]
\label{alg:election-new}
{\small
	% state in pseudo-code
	\caption{Arbitrary ballot election (ABE)}
	\begin{algorithmic}[1]
%		\algrestore{alg:p0}
		\Part{$\lit{ABE}(m)$} {
			\State $\ms{c} \leftarrow \lit{Enc}(\ms{ID}, m)$ \label{line:abe:encrypt}
			\State $\ms{encs} \leftarrow \lit{AVCP}[\ms{ID}](m)$ \label{line:abe:avcp}
			\For{\textbf{each} (disjoint) set of values $M = \{(m_1, \sigma_1), \ldots, (m_k, \sigma_l)\} \subseteq \ms{encs}$ s.t. $m_1 = \dots = m_l$} \label{line:abe:fordiscard}
			\State $\ms{unique-encs} \leftarrow \ms{unique-encs} \cup \{m_1\}$ \label{line:abe:discard}
			\EndFor
			\Return $\lit{decryption}(\ms{unique-encs})$ \EndReturn \label{line:abe:dec1}
%	\algstore{alg:election}
		}\EndPart
		\Statex

		\Part{$\lit{decryption}(\ms{encs})$} {
			\For{each $c$ in $\ms{encs}$}
			\State $\sigma \leftarrow \lit{ShareGen}(i, \ms{ID}, c)$
			\State $\ms{partial-decs}[c] \leftarrow \{\sigma\}$
			\EndFor
			\State $\lit{broadcast}$ $($\DECS$, \ms{ID}, \ms{partial-decs})$ 
			\State $\ms{has-broadcast} \leftarrow \lit{true}$
			\Statex
		}\EndPart
		\Upon{receipt of $($\DECS$, \ms{ID}, \ms{decs'})$} \label{line:abe:uponshares} {
%			\State $valid \leftarrow (\ms{encs}.\lit{keys}() = \ms{encs'}.\lit{keys}())$
%			\State $valid \leftarrow (\ms{encs}.\lit{keys}() = \ms{encs'}.\lit{keys}())$
			\If{$(\ms{partial-decs}.\lit{keys}() = \ms{decs'}.\lit{keys}()) \land (\lit{VerifyShare}(\ms{ID}, c, \ms{decs'}[c]) = 1$ for all $c \in \ms{decs'}.\lit{keys}())$}
			\WUntil $\ms{has-broadcast} = \lit{true}$ \EndWUntil \label{line:abe:condition}
%			\If{$valid$}
%			\For{\textbf{each} $c \in \ms{encs}.\lit{keys}()$}
%			\If{$\lit{VerifyShare}(\ms{ID}, c, \ms{encs'}[c]) \neq 1$}
%			\State $valid \leftarrow \lit{false}$
%			\State \textbf{break}
%			\EndIf
%			\EndFor
%			\If{$valid$}
			\For{\textbf{each} $(c, \sigma) \in \ms{partial-decs'}$}
			\State $\ms{partial-decs}[c] \leftarrow \ms{partial-decs}[c] \cup \{\sigma\}$
			\EndFor
			\If{$|\ms{partial-decs}[c]| = k$ for all $c \in \ms{partial-decs}.\lit{keys}()$}
			\For{\textbf{each} $c \in \ms{partial-decs}.\lit{keys}()$}
			\State $\ms{plaintexts} \leftarrow \ms{plaintexts} \cup \{\lit{Dec}(\ms{ID}, c, \ms{partial-decs}[c])\}$ \label{line:abe:dec}
			\EndFor
			\Return $\ms{plaintexts}$ \EndReturn
			\EndIf
			\EndIf
%			\EndIf
		} \EndUpon
	\end{algorithmic}}
\end{algorithm}

Each process begins with the plaintext $m$ to propose to consensus.
Processes encrypt $m$ under $\lit{Enc}$, producing the ciphertext $c$ (line~\ref{line:abe:encrypt}).
	Processes propose their encrypted value to an instance of AVCP identified by $\ms{ID}$.
	Since $\lit{valid}()$ checks that $\lit{VerifyEnc}(\ms{ID}, m) = 1$ for a given tuple $(m, \sigma)$, $\ms{encs}$ will thus contain well-formed encryptions (AVC-Validity) of processes' proposals.

	Now, AVCP guarantees that signatures, rather than the contents of messages, are unique.
So, it is conceivable that a process will mount a replay attack, which aims to disrupt an election by mimicking the input of some process.
Since values are encrypted under a common instantiation of threshold encryption, it is desirable to prevent this attack.
Consequently, we require that non-faulty processes prepend a sufficiently large sequence of random bits to their plaintext.
For simplicity, we assume that each sequence that a non-faulty proposes derives is unique.
So, non-unique encrypted messages are discarded (lines \ref{line:abe:fordiscard} and \ref{line:abe:discard}), resulting in the set $\ms{unique-encs}$.
	At this point, threshold decryption is performed with respect to each element of $\ms{unique-encs}$ (line~\ref{line:abe:dec1}).
For each (encrypted) value that a process decided, a decryption share is produced.
Then, processes broadcast each share and the corresponding ciphertext used to produce it.

Processes do not process shares that they have received from other processes until they have broadcast their shares, ensuring termination for all non-faulty processes.
In practice, processes can process shares provided they delay termination until after they have broadcast.
Upon receipt of a (potential) set of shares, processes check that the ciphertexts received match theirs.
Given this holds, they check that all shares are well-formed.
At this point, processes store these shares in \ms{encs}.
Once a process has received $k$ shares for every ciphertext in $\ms{unique-encs}$, each ciphertext is decrypted, and the resulting plaintexts are returned.

\subsection{Analysis}
We prove a number of properties hold:

\begin{lemma}
	ABE satisfies \textit{termination}.
	That is, all non-faulty processes eventually complete protocol execution.
\end{lemma}

\begin{proof}
	At the protocol's outset, all non-faulty processes produce a value $c$ s.t. $c \leftarrow \lit{Enc}(\ms{ID}, m)$ was called for some $m \in \{0, 1\}$ (line \ref{line:abe:encrypt}).
	Then, processes execute AVCP with respect to the identifier $\ms{ID}$, where every non-faulty process proposes a valid ciphertext $c$ (line~\ref{line:abe:avcp}).
	On AARB-delivery of each value, a process' call to $\lit{valid}()$ will return $\lit{true}$ by encryption correctness, as it must be the case $\lit{VerifyEnc}(\ms{ID}, c) = 1$ for such a $c$.
	By AVC-Termination, AVC-Validity and AVC-Agreement, all non-faulty processes eventually terminate AVCP with the same set of values, each of which satisfies $\lit{valid}()$.
	By AVC-Agreement, all non-faulty processes will obtain the same set $\ms{unique-encs}$ (after line \ref{line:abe:discard}).
	Then, all non-faulty processes will produce decryption shares for each value in $\ms{unique-encs}$ via calls to $\lit{ShareGen}$, which are guaranteed to be broadcast as no non-faulty process can terminate until the condition at line \ref{line:abe:condition} is fulfilled.
	On receipt of a set of decryption shares (line~\ref{line:abe:uponshares}) (with identical corresponding ciphertexts) from a non-faulty process, by share correctness each share $\sigma$ will satisfy $\lit{VerifyShare}(\ms{ID}, c, \sigma) = 1$ for the corresponding ciphertext $c$.
	Thus, a non-faulty process will eventually receive $k = t + 1$ valid decryption shares for each unique value that was decided by AVCP.
	Thus, they can call $\lit{Dec}$ with respect to each set of shares (line~\ref{line:abe:dec}) which by decryption correctness will pass, and thus the process will return the corresponding plaintexts.

\end{proof}

\noindent Comparable definitions to anonymity-preserving vector consensus (Section~\ref{sec:avc:avc}) regarding \textit{agreement}, \textit{validity} and \textit{anonymity} follow straightforwardly from the fact that AVCP satisfies AVC-Agreement, AVC-Validity and AVC-Anonymity, respectively.

\begin{lemma}
	ABE satisfies \textit{public verifiability}.
	That is, any third party can obtain the result of the election after termination.
\end{lemma}

\begin{proof}
	By termination and agreement, all non-faulty processes eventually agree on the same set of plaintext values.
	Then, a third party can request these values from all processes.
	On receipt of $t + 1$ identical sets of values, the third party can deduce that the set of values corresponds to the election result.
\end{proof}

\begin{lemma}
	ABE satisfies \textit{weak privacy}.
	That is, a value proposed by a non-faulty process is only revealed after AVCP has terminated for a non-faulty process.
\end{lemma}

\begin{proof}
	Each non-faulty process $p$ proposes a value $c$ s.t. $c \leftarrow \lit{Enc}(\ms{ID}, m)$ was invoked for some $m \in \{0, 1\}^{*}$ (line~\ref{line:abe:encrypt}).
	Encryption hiding implies that $m$ can only be revealed if either $p$ reveals $m$, which $p$ does not in the protocol, or if a valid call to $\lit{Dec}$ is made.
	By decryption security, $\lit{Dec}$ will only return $m$ if provided $k = t + 1$ valid decryption shares.
	Now, no non-faulty process broadcasts a decryption share until after it terminates AVCP.
	Thus, a coalition of $t$ faulty processes cannot decrypt $m$ via $\lit{Dec}$, which is the only conceivable way for them to obtain $m$ in the model.
\end{proof}

\begin{lemma}
	ABE satisfies \textit{eligibility}.
	That is, only processes in $P$ may propose a ballot that is decided.
\end{lemma}

\begin{proof}
	By assumption, non-faulty processes only process messages sent over regular channels from processes in $P$.
	But, any party may anonymously broadcast a value $v$.
	AARBP satisfies AVC-Signing.
	That is, all AARB-delivered messages are of the form $(m, \sigma)$, where $\sigma \leftarrow \lit{Sign}(i, \ms{ID}, m)$.
	By assumption on the interface of $\lit{Sign}$, only a process in $P$ may call $\lit{Sign}$.
	By construction of AVCP, only messages that are AARB-delivered are decided by non-faulty processes.
	Thus, no non-faulty process will decide a value $v$ in AVCP from any party outside of $P$.
\end{proof}

\begin{lemma}
	ABE satisfies \textit{non-reusability}.
	That is, at most one encrypted ballot signed by a particular process can be decided in the election.
\end{lemma}

\begin{proof}
	By construction of AVCP, only messages that are AARB-delivered are decided by non-faulty processes.
	By AARB-Unicity, no non-faulty process will AARB-deliver more than one tuple $(m, \sigma)$ such that a process $p_i$ invoked $\sigma \leftarrow \lit{Sign}(i, \ms{ID}, m)$.
\end{proof}

% TODO: efficiency/optimisations?

%\input{evotingabfe/evoting-abfe.tex}
\section{Further experiments} \label{sec:bench:bench}
In this section, we present experimental results of our implementations of cryptographic and distributed protocols.
To determine their influence in cost over our distributed protocols, we benchmark \textit{two} cryptographic protocols which our distributed protocols rely on, namely Fujisaki's traceable ring signature (TRS) scheme \cite{Fujisaki2007} and Shoup and Gennaro's threshold encryption scheme \cite{Shoup2002}.
With this information, we then benchmark and evaluate our \textit{three} distributed protocols: Anonymised All-to-all Reliable Broadcast Protocol (AARBP), Anonymised Vector Consensus Protocol (AVCP), and our election scheme, Arbitrary Ballot Election (ABE).

\subsection{Cryptography}
Benchmarks of standalone cryptographic constructions were performed on a laptop with an Intel i5-7200U (quad-core) processor clocked at 3.1GHz and 8GB of memory.
The operating system used was Ubuntu 18.04. Each cryptosystem was implemented in golang.
All cryptographic schemes were implemented using Curve25519 \cite{Bernstein2006}.
To simulate a prime-order group we use the Ristretto technique \cite{Ristretto}, derived from the Decaf approach \cite{Hamburg2015}, via go-ristretto\footnote{\url{https://github.com/bwesterb/go-ristretto}}.
All cryptographic operations rely on constant-time arithmetic operations to prevent side-channel attacks \cite{Brier2002}.
Each data point represents a minimum of 100 and a maximum of 10000 iterations.

\subsubsection{Ring signatures} \label{sec:bench:trs}
To confirm that the implementation is competitive, we compare it to an implementation of Lui et al.'s linkable ring signature scheme \cite{Liu2004}. The particular implementation we compare to was produced by EPFL's DEDIS group as part of kyber\footnote{ \url{https://godoc.org/github.com/dedis/kyber/sign/anon}}. We note that the ring signature schemes naturally lend themselves towards parallelism. To this end, we provide an extension of our TRS implementation that takes advantage of concurrency.

We compared the two major operations.
Firstly, the Sign() operation (``$\lit{Sign}$'' in Section~\ref{sec:model:model}) is that which forms a ring signature.
Secondly, Verify() (``$\lit{VerifySig}$'' in Section~\ref{sec:model:model}) takes a ring signature as input, and verifies the well-formedness of the signature, and (in implementation) outputs a tag which is used for linking/tracing.
\begin{figure}[htbp]
\centering
\begin{subfigure}{.48\textwidth}
\centering
\includegraphics[width=1.0\textwidth]{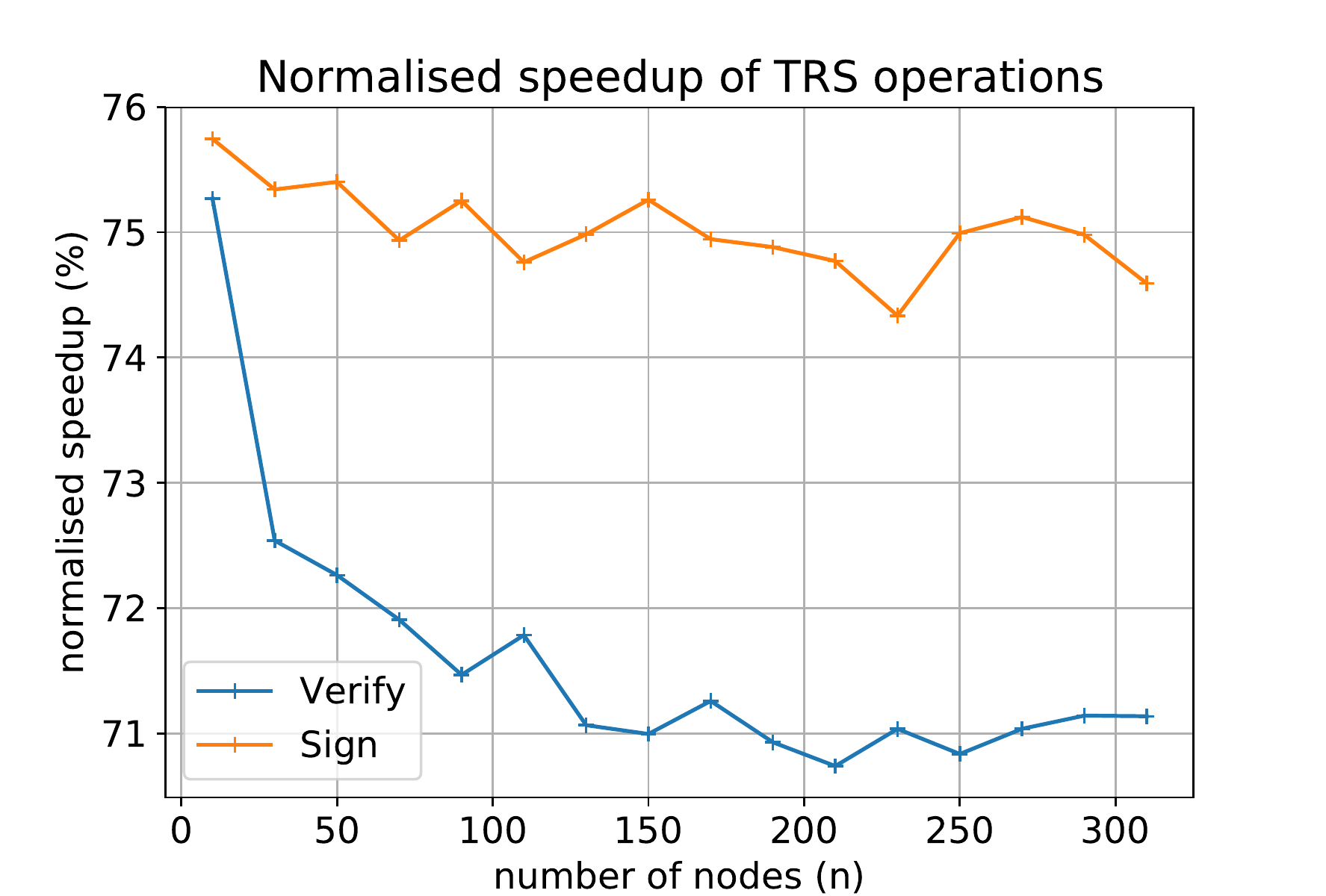}
\caption{Normalised speedup, as a percentage, of TRS operations over LRS operations}
\label{fig:1perfgain}
\end{subfigure}
\begin{subfigure}{.48\textwidth}
    \centering
	\includegraphics[width=1.0\textwidth]{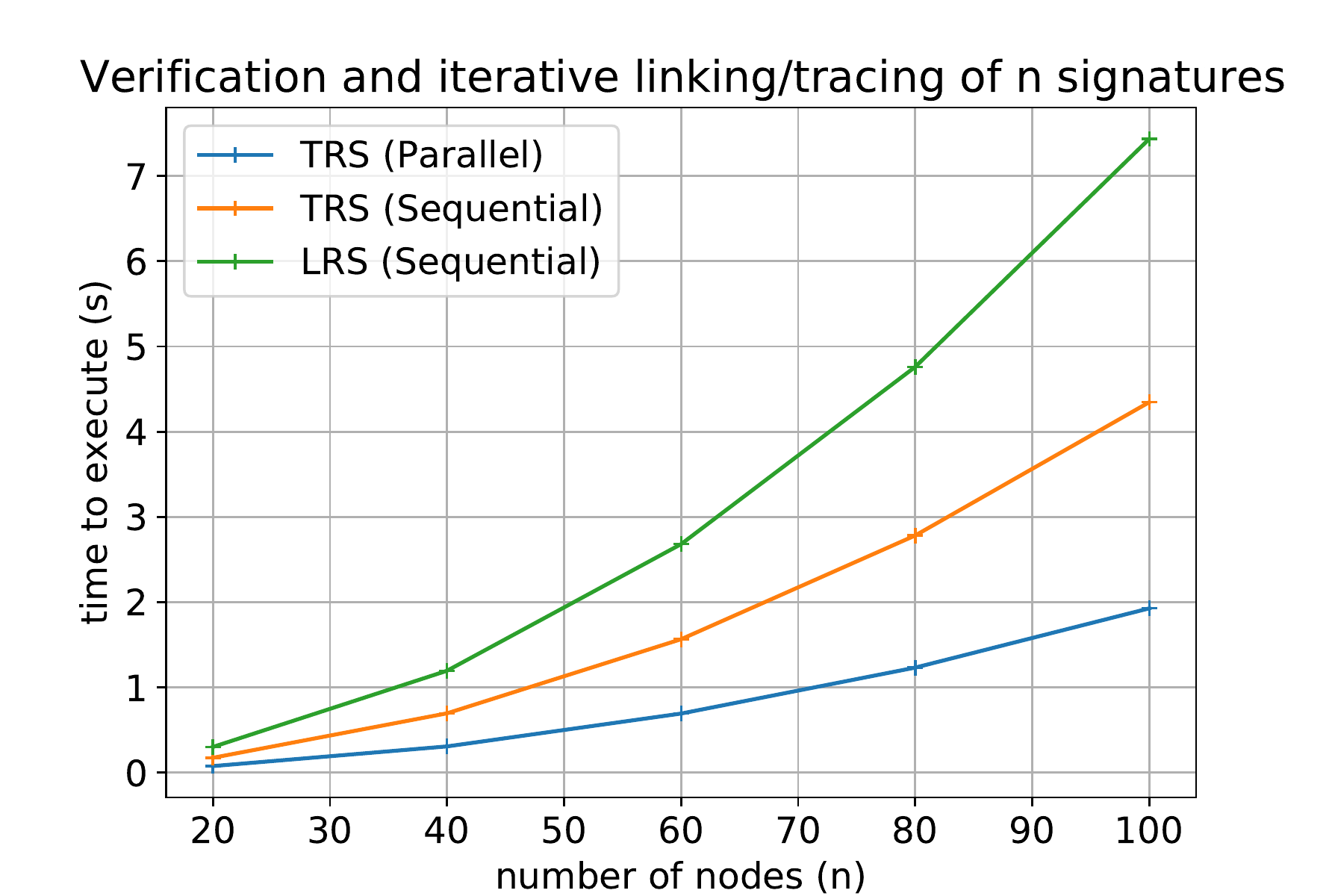}
	\caption{Comparing the performance of iterative verification and linking/tracing}
    \label{fig:iterativever}
\end{subfigure}
\caption{Comparing TRS and LRS operations}
\end{figure}

Figure \ref{fig:1perfgain} represents the normalised speedup, as a percentage, afforded to the Sign() and Verify() operations in the TRS scheme. We note that Figure \ref{fig:1perfgain} involves purely sequential implementations. We vary $n$, the size of the ring. Performance of Sign() and Verify() in both schemes scales linearly with respect to $n$.

Let $T_{LRS}$ be the time an operation (signing or verification) takes in the TRS implementation, and similarly define $T_{TRS}$. Then, normalised speedup is given by:
\[\text{normalised speedup} = \frac{T_{LRS} - T_{TRS}}{T_{TRS}}\]
Indeed, this value hovers above 70 percent for all values of $n$ tested. In the LRS implementation, the Sign() operation took between 7.5 and 230.4ms to execute (for values $n = 10$ and $n = 310$ respectively). Sign() in the TRS implementation took between 4.3 and 132.0ms to execute. Similarly, Verify() took between 7.7 and 230.3ms to execute in the LRS implementation, and took between 4.4 and 134.6ms to execute in the TRS implementation.

To understand the performance increase, we count cryptographic operations. We use the notation presented in \cite{Yuen2013}, adapted to represent elliptic curve operations. Let $E$ denote the cost of performing point multiplication, and $M$ denote the number of operations of the form $(aP + bQ)$, where $a$, $b$ are scalars, and $P, Q$ are points. Then, in the LRS scheme, Sign() requires $2(n - 1)M + 3E$ operations, and Verify() requires $2nM$ operations \cite{Yuen2013}. In the TRS scheme, Sign() requires the same $2(n -1)M$ and $3E$ operations, in addition to $nM$ other operations. Similarly, the TRS Verify() call requires $2nM + nM = 3nM$ operations. With the same elliptic curve implementation, we should see a performance increase from an efficient LRS implementation, as the TRS scheme requires ~1.5x $M$ operations by comparison to the LRS scheme. But, in the library we use for ECC, go-ristretto, base point multiplication is ~2.8x faster, point multiplication is ~2.3x faster, and point addition is ~2.6x faster. Thus, we see an overall increase in performance in our TRS implementation by comparison.

Recall that in AARBP (and thus AVCP), ring signature tracing is performed iteratively. Suppose that some process has $x$ correct ring signatures that they have verified at some point in time. Then, upon receipt of another signature, they first execute Verify(), and then perform the Trace() operation between the new signature and each of the $x$ stored signatures. To model this behaviour, we perform a benchmark where signatures are processed one-by-one with a ring of size $n$ that is varied from 20 to 100 in increments of 20. As the linkable ring signature scheme admits very similar functionality, except it performs a \textit{linking} operation, rather than a \textit{tracing} operation, we benchmark similarly.

Figure \ref{fig:iterativever} compares the average time taken to perform the aforementioned procedure, between a sequential implementation of the TRS scheme, the corresponding concurrent implementation (utilising four cores), and the (sequential) LRS implementation. Now, the time taken to perform Verify() is proportional to the size of the ring $n$. In addition, we perform more ($n$) Verify() operations as we increase $n$. Consequently, execution time grows quadratically in Figure \ref{fig:iterativever} for each implementation. In addition to the previously outlined speedup, we roughly halve execution time of TRS operations by exploiting concurrency in our implementation. With $n = 100$, our concurrent implementation is 2.25x faster than our sequential implementation, and is 3.86x faster than the sequential LRS implementation.

To perform Trace() between two signatures, $O(n)$ comparisons are performed in a naive implementation, whereas a single comparison is need to link in the TRS scheme. In implementation, we used a hash table, so each new signature required $O(n)$ lookups in the TRS scheme, and one lookup in the LRS scheme. As expected, we observe that the increased overhead of tracing is dominated by the time taken performing Verify() operations, and so speedup does not appear to be affected.

%\begin{figure}%
%    \centering
%	\includegraphics[width=0.9\textwidth]{graphs/pgsignverify2.pdf}
%    \caption{Comparing sequential and concurrent TRS operation performance}%
%    \label{fig:vsconc}
%\end{figure}

%Figure \ref{fig:vsconc} compares the sequential implementation of TRS Sign() and Verify() operations to their respective concurrent implementations, varying the number of cores utilised from 2 to 4. Indeed, there is a noteworthy improvement in performance when introducing a second core to perform operations. The comparatively minor improvements in performance when using more than two cores can largely be attributed to the inherent overhead in spawning threads. As $n$ increases, we still see an increase in operations per second with the three and four core benchmark. When $n = 50$, concurrent verification is 1.92x faster and 2.12x faster than sequential verification with 2 and 4 cores respectively. When $n = 290$, concurrent verification is 1.93x and 2.30x faster than sequential verification with 2 and 4 cores respectively. Thus, we can expect additional efficiency gains if we were to perform operations with 4 cores with very large values of $n$.

%Thus, we have shown that our implementation of Fujisaki's signature scheme is competitive with the implementation of the LRS scheme in kyber, and admits reasonable performance in of itself.

\subsubsection{Threshold encryption} \label{sec:bench:tenc}
Our arbitrary ballot voting scheme, ABE, presented in Appendix~\ref{sec:evoting:abe}, can be instantiated with Shoup and Gennaro's threshold encryption scheme \cite{Shoup2002}. To this end, we present results corresponding to our implementation of their construction.

\begin{table}[htbp]
\begin{center}{\small 
\begin{tabular}{|l|l|}
\hline
\textbf{Operation}               & \textbf{Time to execute (ms)} \\ \hline
Encryption              & $0.407$                \\ \hline
Encryption verification & $0.358$               \\ \hline
Share decryption        & $0.247$                \\ \hline
Share verification      & $0.339$                \\ \hline
\end{tabular}
\caption{Threshold encryption operations}
\label{table:teops}}
\end{center}
\end{table}

Table \ref{table:teops} shows the performance of all operations in the threshold encryption scheme \cite{Shoup2002} as executed by one process, bar decryption itself. As can be seen, all operations can be executed in a reasonable amount of time (less than a millisecond). In our electronic voting protocol, each process performs a single encryption operation, but may perform $O(n)$ operations of the other forms over the protocol's execution. Even when $n$ is relatively large ($ > 100$), we can expect to see acceptable levels of performance.

Let $P$ be a group of $n$ processes. Then, the final operation, \textit{share combination}, combines a group of $k$ valid decryption shares, where $k$ is the threshold required to reconstruct the secret. In our electronic voting protocols, for interoperability with our consensus algorithms, we set $k = t + 1$, where $t$ is, at most, the largest value such that $t < \frac{n}{3}$.

\begin{figure}[htbp]
\centering
\includegraphics[width=0.5\textwidth]{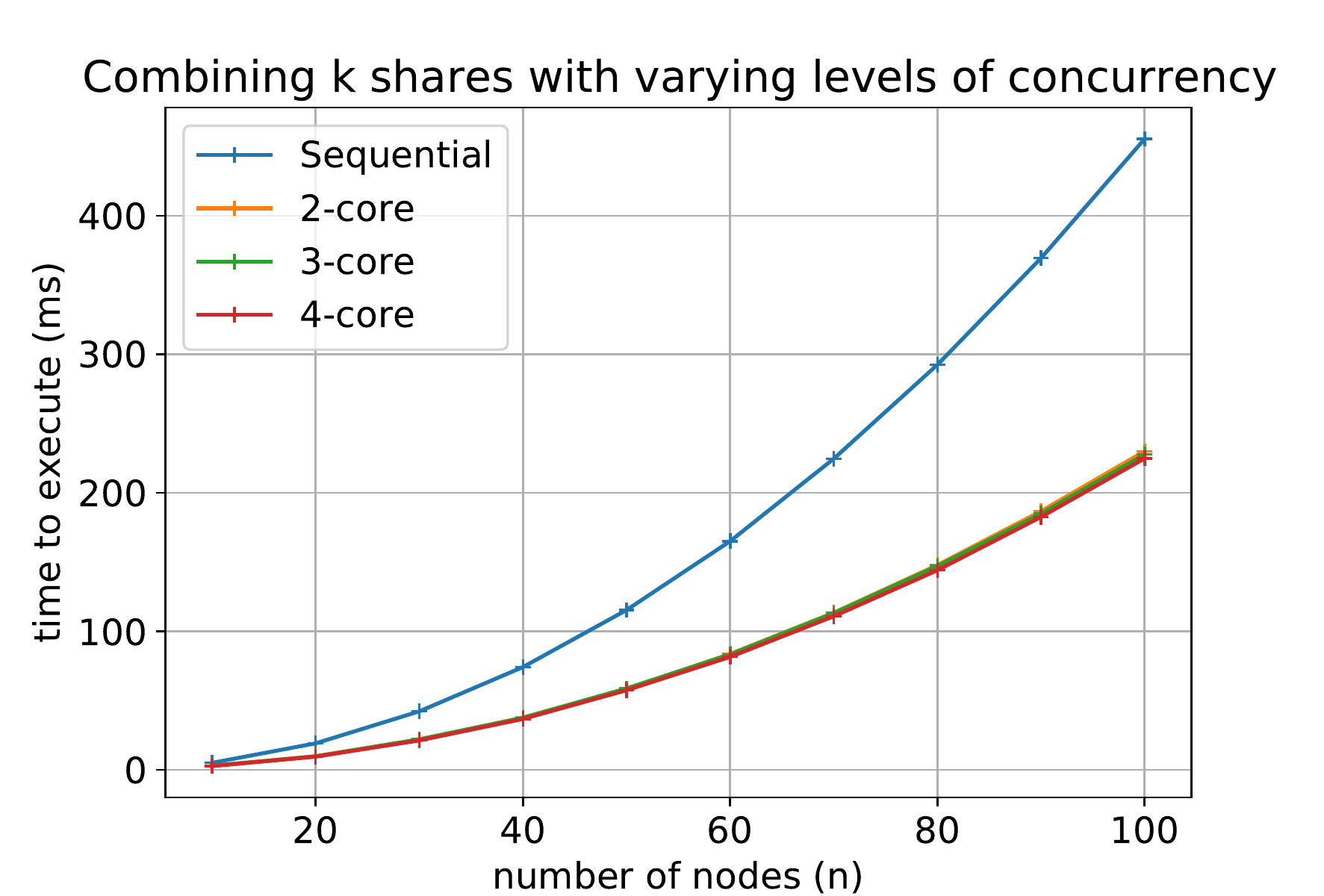}
\caption{Time taken for $k$ partial decryptions (shares) to be combined to decrypt a given message}
\label{fig:combinekshares}
\end{figure}

Figure \ref{fig:combinekshares} represents the time taken to decrypt a message, varying $k$, the number of shares needed to reconstruct the message in increments of 10. As some steps of the reconstruction process can be performed concurrently, e.g. in the derivation of Lagrange coefficients \cite{Shamir1979}, we provide an extension that takes advantage of a multi-core processor. To this end, Figure \ref{fig:combinekshares} graphs the time taken to perform share combination with both our sequential implementation, and the corresponding concurrent implementation utilising two to four cores.

Indeed, producing Lagrange coefficients requires a quadratic amount of work (with respect to $k$), and subsequently decrypting a message takes $O(k)$ effort, which dominates execution time. Both the quadratic and linear components of the share combination can be made concurrent. Consequently, we roughly double our performance with two to four cores running. It is worth noting that the effort required to spawn additional threads in the three and four core case does not translate to a very noticeable improvement in performance for our values of $k$.

%On sum, the performance admitted by our implementation is reasonable. For a referendum with $n = 300$ participants, our implementation with a mid-range laptop processor utilising two cores can decrypt a homomorphically summed tally in less than 250ms with maximal fault tolerance (notwithstanding computing the discrete logarithm \cite{Cramer1997}), which is practically acceptable.

\subsection{Elections with arbitrary ballots}
Arbitrary Ballot Election (ABE) essentially combines AVCP with a threshold encryption scheme.
To perform our benchmarks for the election scheme, we used pre-generated keying material for threshold encryption.
We use Shoup and Gennaro's threshold encryption scheme~\cite{Shoup2002}, which was benchmarked in the previous section.
Our experiment differs from experiments with AVCP in that ring signatures must also contain valid encryptions as per the threshold encryption scheme, and because all processes perform threshold decryption with respect to all decided ballots.

\begin{figure}[htbp]
    \centering
	\includegraphics[width=0.5\textwidth]{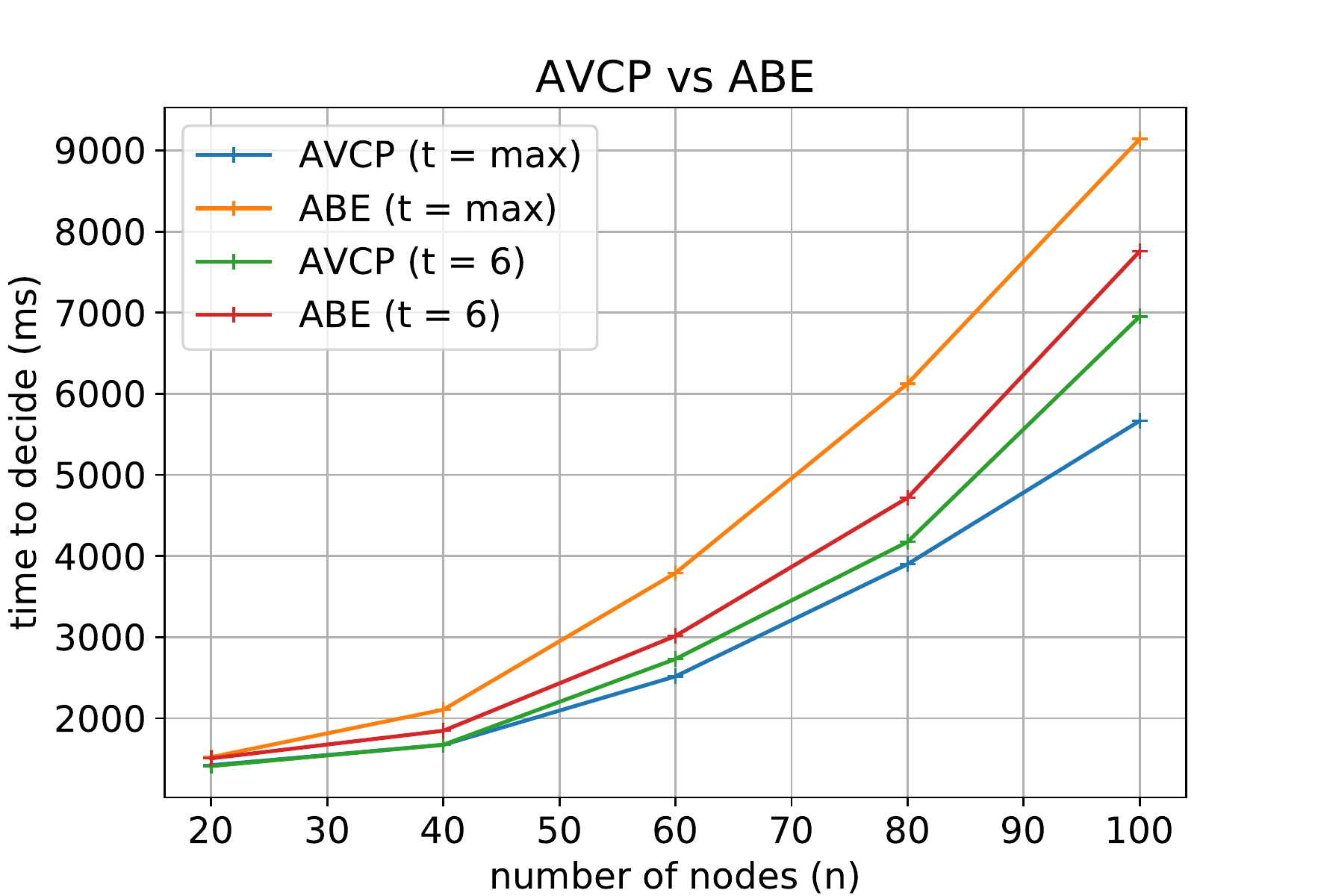}
	\caption{Comparing the performance of AVCP and ABE}
    \label{fig:avc_vs_election}
\end{figure}

Figure \ref{fig:avc_vs_election} compares the performance of AVCP with ABE as described above. As can be seen, there is some, but not a considerable amount, of overhead from introducing the election's necessary cryptographic machinery. The main factors here are an additional message step, additional message verifications, and, requiring the most additional effort, performing threshold decryption.

As can be seen (and was explored in Section \ref{sec:exp}), AVCP performance degrades when $t$ is non-optimal. Despite this, the election still takes longer to perform when $t$ is increased. Consider the case where $(n, t) = (100, 33)$. As shown in Figure \ref{fig:combinekshares}, combining $k = t + 1 = 34$ (valid) partial decryptions together takes roughly 30ms. Since each run of the experiment decides on at least $n - t = 67$ values, processes have to spend almost two seconds combining shares together. In the best case, each process has to verify that $(n - t)t = 2211$ shares are well-formed, which takes roughly 750ms. Thus, in addition to other cryptographic overhead, it is clear that increasing $t$ affects election performance.

Notwithstanding, the election protocol performs well for reasonable values of $n$. At $n = 100$, the election roughly takes between 7.5 and 9 seconds to execute from start to finish. Given that the latency provided by the anonymous channels used by processes is sufficiently low, we can expect to see comparable results in practice.

\fi

\end{document}